\documentclass[letterpaper, 10 pt, conference, onecolumn]{ieeeconf}  

\IEEEoverridecommandlockouts                              
\usepackage{amsmath} 
\usepackage{amssymb}  
\usepackage{amsmath,amssymb,amsfonts,amsthm}
\usepackage{graphicx,float}
\usepackage{algorithm}
\usepackage{algorithmicx}
\usepackage{algpseudocode}
\usepackage{setspace}
\usepackage{autobreak}
\usepackage{subfigure}
\allowdisplaybreaks[4]

\title{\LARGE \bf
  Mean Field Game and Decentralized Intelligent Adaptive Pursuit Evasion Strategy for Massive Multi-Agent System under Uncertain Environment
}

\author{Zejian Zhou and Hao Xu
\thanks{The authors are with the Department
of Electrical and Biomedical Engineering, University of Nevada, Reno,
NV, 89557 USA e-mail: zejianz@nevada.unr.edu; haoxu@unr.edu. The part of this work has been supported by NASA under Grant No. NNX15AIO2H.}}%

\begin{document}

\maketitle
\theoremstyle{definition}
\newtheorem{definition}{Definition}
\newtheorem{theorem}{Theorem}
\newtheorem{remark}{Remark}
\newtheorem{lemma}{Lemma}
\newtheorem{assumption}{Assumption}

\begin{abstract}
In this paper, a novel decentralized intelligent adaptive optimal strategy has been developed to solve the pursuit-evasion game for massive Multi-Agent Systems (MAS) under uncertain environment. Existing strategies for pursuit-evasion games are neither efficient nor practical for large population multi-agent system due to the notorious ``Curse of dimensionality" and communication limit while the agent population is large. To overcome these challenges, the emerging mean field game theory is adopted and further integrated with reinforcement learning to develop a novel decentralized intelligent adaptive strategy with a new type of adaptive dynamic programing architecture named the Actor-Critic-Mass (ACM). Through online approximating the solution of the coupled mean field equations, the developed strategy can obtain the optimal pursuit-evasion policy even for massive MAS under uncertain environment. In the proposed ACM learning based strategy, each agent maintains five neural networks, which are 1) the critic neural network to approximate the solution of the HJI equation for each individual agent; 2) the mass neural network to estimate the population density function (i.e., mass) of the group; 3) the actor neural network to approximate the decentralized optimal strategy, and 4) two more neural networks are designed to estimate the opponents' group mass as well as the optimal cost function. Eventually, a comprehensive numerical simulation has been provided to demonstrate the effectiveness of the designed strategy. 
\end{abstract}

\section{INTRODUCTION}
Pursuit-evasion games have received increasing attention in multi-agent decision-making and control studies (e.g.  \cite{vlahov2018developing}, \cite{ramana2017pursuit}. The problem can be widely found in numerous applications such as quadcopter flight control \cite{camci2016game}, ground vehicle tracking \cite{wilson2017pursuit}, missile guidance system \cite{turetsky2016target} etc. Recently, some of the researches explored a novel type of pursuit-evasion problem for multiple pursuers and evaders due to the enormous gain from the larger population of agents. For instance, \cite{makkapati2018pursuit} studied the pursue evasion problem with two pursuers and one evader; \cite{sun2017multiple} used multiple pursuers, i.e. unmanned aircraft systems (UAS), to capture the ground vehicle. The differential game formulation associated with the Hamilton-Jacobi-Isaacs (HJI) equation is used in those studies to obtain the optimal strategies. However, there are two common limitations in these studies, 1) the agent number cannot be large, 2) a high-quality and reliable communication system is needed for supporting information exchange among distributed agents. In large scale Multi-agent Systems (MAS), these limitations cannot be ignored due to the notorious ``curse of dimensionality", and unreliable communication network in practical (Fig. \ref{fig:challenge_diagram}).

\begin{figure}
\centerline{\includegraphics[width=1\linewidth]{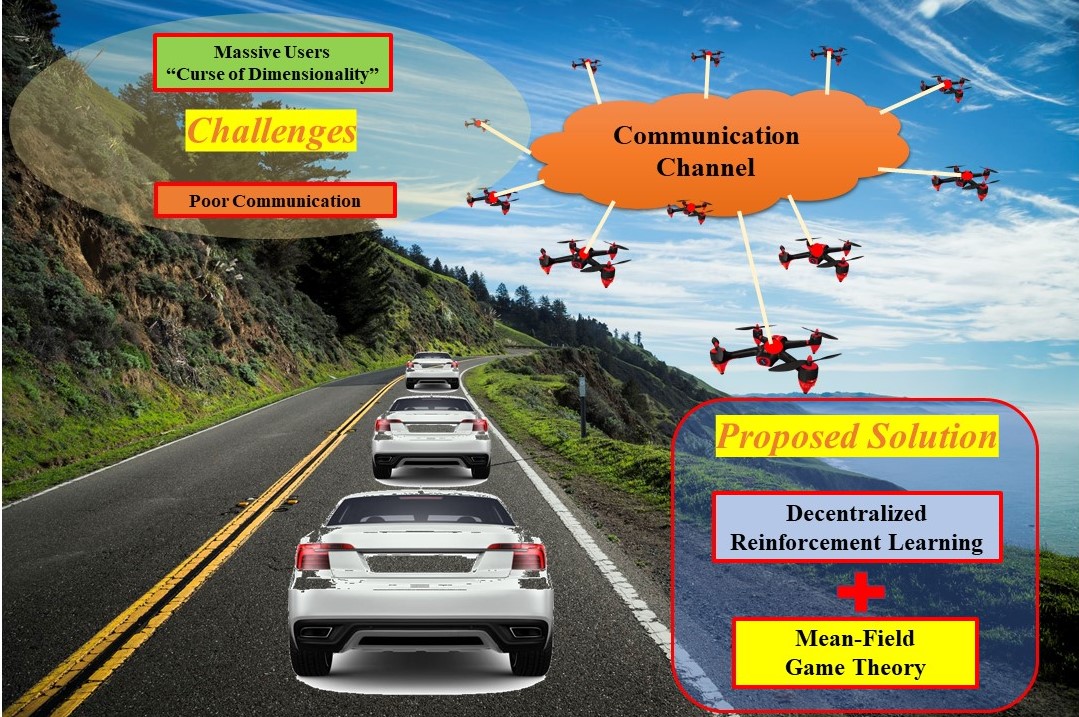}}
\caption{An illustration of challenges and the proposed solution in massive multi-agent pursuit-evasion game, i.e. Massive UASs are tracking the ground vehicles.}
\label{fig:challenge_diagram}
\end{figure}

To overcome these challenges, the emerging Mean Field Game (MFG) theory is adopted and engaged with pursuit-evasion game to develop a decentralized strategy for massive MAS. The key feature of MFG is that a new \textbf{mass} function has been constructed to approximate all the other agents' states through their probability distribution \cite{faguo1}. Different than other mean states based algorithm such as the ``average consensus" algorithm \cite{olfati2007consensus} where the deterministic average states are observed, the MFG estimates the stochastic distribution of all agents' states by solving a Partial Differential Equation (PDE), named Fokker-Planck-Kolmogorov (FPK) equation. The distribution (i.e. mass) is then used to represent the effect from all other agents in the agent's decision-making process. Lasry and Lions \cite{faguo1} first proved that by integrating the \textbf{mass} with the Hamilton-Jacobi-Bellman (HJB) equation from optimal control theory \cite{frankbook}, one can obtain the $\varepsilon-$ Nash equilibrium of the game and further converge to the Nash equilibrium as the agent number goes to infinity. Since the \textbf{mass} is approximated by a PDE which is independent on the agent number, the MFG can be used to tackle the communication limit and the ``curse of dimensionality". In this paper, the pursuers and evaders are using two mass functions to represent the pursuer group density and evader density during the game. Both mass function are integrated into the HJI equation to represent the influence from other agents in the same group.

However, solving Mean Field Game (MFG) is computationally expensive due to the coupled HJI and FPK equations especially with infinite-dimensional states. Meanwhile, the reinforcement learning and adaptive approximate dynamic programming (ADP) techniques \cite{frankbook} have been successfully utilized to solve general HJI equations and learn optimal nonlinear control. Therefore, we extend the ADP technique to a novel Actor-Critic-Mass (ACM) algorithm that can approximate the coupled HJI and FPK equations simultaneously and further obtain the optimal pursuit-evasion strategy. Specifically, five neural networks are designed to approximate the solutions of coupled two HJI equations, two FPK equations, and the optimal control. 

The main contributions of this paper can be summarized as follows:
1) The pursuit-evasion game with massive multi agents has been solved through integrating the Mean Field Game theory which tackles the ``curse of dimentionality" problem as well as requires no communication or observation.
2) A novel reinforcement learning structure named Actor-Critic-Mass (ACM) for differential games has been proposed to numerically solve the optimal strategy for pursuit-evasion game online. The solution of coupled HJI and FPK equations can thus be approximated by ACM.  

\section{Background and Problem Formulation}
Consider a group of pursuers $\mathcal{G}_{1}$ and a group of evaders $\mathcal{G}_{2}$ with identical $N$ agents in each group being travelling in an $l$ dimensional space. The states of individual agent in $\mathcal{G}_{1}$ and $\mathcal{G}_{2}$ are denoted by $x_{g 1,i}\in\mathbb{R}^{l}$ and $x_{g 2,j}\in\mathbb{R}^{l}$, respectively. The system dynamics for each agent are affected by other agents and can be described through a group of stochastic differential equations (SDEs), i.e.:
\begin{align}
&\label{eq:system_dynamics1_d}dx_{g 1, i}=\big[f_{g 1}\left(x_{g 1, i}\right)+g_{g 1}\left(x_{g 1, i}\right) u_{g 1, i}\\\nonumber
&+\boldsymbol{G}_{g 2}\left(\boldsymbol{x}_{g 2}\right)\big]dt+\sigma_{g 1,i} d w_{g 1, i}\\
&\label{eq:system_dynamics2_d}dx_{g 2, j}=\big[f_{g 2}\left(x_{g 2, j}\right)+g_{g 2}\left(x_{g 2, j}\right) u_{g 2, j}\\\nonumber
&+\boldsymbol{G}_{g 1}\left(\boldsymbol{x}_{g 1}\right)\big]dt+\sigma_{g 2,j} d w_{g 2, j}
\end{align}
where $u_{g1,i}\in\mathbb{R}^{l}$ is the control input of the $i$th agent, $w_{g1,i}$ denotes a set of independent Wiener processes representing environment noise for agents in the group $\mathcal{G}_{1}$, $\sigma_{g 1, i}$ is the coefficient matrix of the Wiener process, the functions $f_{g 1, i}\left(x_{g 1, i}\right)$ and $g_{g 1, i}\left(x_{g 1, i}\right)$ represent the intrinsic dynamics of the agents in the group $\mathcal{G}_{1}$, and the $\boldsymbol{G}_{g 2}\left(\boldsymbol{x}_{g 2}\right)$ denotes the influence from the group $\mathcal{G}_{2}$. The parameters in (\ref{eq:system_dynamics2_d}) is similar to those in (\ref{eq:system_dynamics1_d}) but for group $\mathcal{G}_{2}$.

The objective for agents in the pursuer group $\mathcal{G}_{1}$ are to intercept the evader at the fixed time $T$ while the agents in the 
evader group $\mathcal{G}_{2}$ attempts to do the opposite. 
\begin{remark}
Different than the conventional pursuit-evasion problem, which has very limited number of pursuers and evaders, the pursuers' and evaders' groups in this problem has countably infinite number of agents, i.e., $N\rightarrow\infty$. Moreover, the agents in each group can neither communicate nor observe the other agents' states, which indicates a decentralized control problem. 
\end{remark}

Next, two cost functions are constructed to evaluate the performance of agents in different groups. The cost function for agents in the group $\mathcal{G}_{1}$ is defined as:
\begin{align}\label{eq:bare_cost1}
&V_{g 1}\left(x_{g 1, i}(t),u_{g 1, i}(t),m_{g 1},m_{g 2}\right)\\\nonumber
&=\int_{0}^{T}\left[\begin{array}{l}{x_{g 1, i}^TQ_{g1}x_{g 1, i}+u_{g 1,i}^{T}(\tau) R_{g 1, i} u_{g 1, i}(\tau)} \\ {+\Phi_{g 1}\left(m_{g 1}(x_{g 1, i}(\tau),\tau), x_{g 1, i}(\tau)\right)}\\{
-\Phi_{g 2}\left(m_{g 2}(x_{g 1, i}(\tau),\tau), x_{g 1, i}(\tau)\right)}\end{array}\right] d \tau
\end{align}
where $m_{g 1}(\tau)$ and $m_{g 2}(\tau)$ are defined as \textbf{mass}, which are the probability density function of group $\mathcal{G}_{1}$'s and $\mathcal{G}_{2}$'s states, respectively. $\Phi_{g 1}\left(m_{g 1}(x_{g 1, i}(\tau),\tau), x_{g 1, i}(\tau)\right)$ and $\Phi_{g 2}\left(m_{g 2}(x_{g 1, i}(\tau),\tau),x_{g 1, i}(\tau)\right)$ are the Mean Field coupling functions that represent the influence on agent $i$ from group $\mathcal{G}_{1}$ and $\mathcal{G}_{2}$, respectively. $Q$ and $R$ are symmetric positive semi-definite and symmetric positive definite matrices, respectively, with compatible dimensions.

Similarly, the cost function for agents in group $\mathcal{G}_{2}$ is given as:
\begin{align}\label{eq:bare_cost2}
&V_{g 2}\left(x_{g 2, j}(t),u_{g 2, j}(t),m_{g 1},m_{g 2}\right)\\\nonumber
&=\int_{0}^{T}\left[\begin{array}{l}{x_{g 2, j}^TQ_{g2}x_{g 2, j}+u_{g 2, j}^{T}(\tau) R_{g 2} u_{g 2, j}(\tau)} \\ {+\Phi_{g 2}\left(m_{g 2}(x_{g 2, j}(\tau),\tau), x_{g 2, j}(\tau)\right)}\\{
-\Phi_{g 1}\left(m_{g 1}(x_{g 2, j}(\tau),\tau), x_{g 2, j}(\tau)\right)}\end{array}\right] d \tau
\end{align}

\begin{figure*}[!t]
\normalsize
\begin{align}
&\label{eq:hji1}\text{HJI-}\mathcal{G}_{1}:-\frac{\partial V\left(x_{g 1, i},u_{g1,i}\right)}{\partial t}-\frac{\sigma_{g1,i}^2}{2} \frac{\partial^{2} V\left(x_{g_{1}, i},u_{g1,i}\right)}{\partial x_{g 1, i}^{2}}+H_{g1}\left(x_{g 1, i}, \frac{\partial V\left(x_{g_{1}, i},u_{g1,i}\right)}{\partial x_{g_{1, i}}}\right)\nonumber\\
&=\Phi_{g 1, i}\left(m_{g 1}, x_{g 1, i}\right)-\Phi_{g 2, i}\left(m_{g 2},x_{g1,i}\right)\\
&\label{eq:fpk1}\text{FPK-}\mathcal{G}_{1}:\frac{\partial m_{g 1}\left(x_{g_{1}, i}, t\right)}{\partial t}-\frac{\sigma_{g1,i}^2}{2} \frac{\partial^{2} m_{g 1}\left(x_{g 1, i}, t\right)}{\partial x_{g 1, i}^{2}}-\operatorname{div}\left(m_{g 1} D_{p} H\left(x_{g 1, i}, \frac{\partial V\left(x_{g 1, i}, u_{g 1,i}\right)}{\partial x_{g 1, i}}\right)\right)=0\\
&\label{eq:hji2}\text{HJI-}\mathcal{G}_{2}:-\frac{\partial V\left(x_{g 2, j}, u_{g2,j}\right)}{\partial t}-\frac{\sigma_{g1,j}^2}{2} \frac{\partial^{2} V\left(x_{g 2, j}, u_{g2,j}\right)}{\partial x_{g 2, j}^{2}}+H_{g1}\left(x_{g 2, i}, \frac{\partial V\left(x_{g 2, j},u_{g2,j}\right)}{\partial x_{g 2, j}}\right)\nonumber\\
&=\Phi_{g 2, j}\left(m_{g 2}, x_{g 2, j}\right)-\Phi_{g 1, j}\left(m_{g 1},x_{g2,j}\right)\\
&\label{eq:fpk2}\text{FPK-}\mathcal{G}_{2}:\frac{\partial m_{g 2}\left(x_{g 2, j}, t\right)}{\partial t}-\frac{\sigma_{g1,j}^2}{2} \frac{\partial^{2} m_{g 2}\left(x_{g 2, j}, t\right)}{\partial x_{g 2, j}^{2}}-\operatorname{div}\left(m_{g 2} D_{p} H\left(x_{g 2, j}, \frac{\partial V\left(x_{g 2, j}, x_{g 2,-j}, x_{g 1}\right)}{\partial x_{g 2, j}}\right)\right)=0\\
&m_{g 1}\left(x_{g 1, i}, 0\right)=m_{g 1, 0}\left(x_{g 1,i}\right)\nonumber\\
&m_{g 2}\left(x_{g 2, j}, 0\right)=m_{g 2, 0}\left(x_{g 2, j}\right)\nonumber
\end{align}
\hrulefill
\vspace*{2pt}
\end{figure*}

Considering the two groups are competitive while the agents in the same group share the same goal (but non-cooperative), the optimal strategy for one agent must satisfy two conditions: 1) the agent's control input belongs to a joint action set which is the saddle point of the groups' cost function; 2) the agent's control input must reach the Nash equilibrium with other agents in the same group. The two conditions for the pursuers' group $\mathcal{G}_{1}$ are equivalent to the following equation:
\begin{align}\label{eq:optimal1}
    &V^*_{g 1}\left(x_{g 1, i}(t),u^*_{g 1, i}(t),m_{g 1},m_{g 2}\right)\nonumber\\
    &=\inf_{\boldsymbol{u}_{g1}} \sup _{\boldsymbol{u}_{g2}}V_{g 1}\left(x_{g 1, i}, u^*_{g 1, i},m_{g 1},m_{g 2}\right)\nonumber \\
    &\leq V_{g 1}\left(x_{g 1, i}, u_{g 1, i},m_{g 1},m_{g 2}\right)
\end{align}
with $u^*_{g1, i}\in\boldsymbol{u}_{g1}$. The optimal cost function and control input for agents in $\mathcal{G}_{2}$ can be similarly obtained as:

\begin{align}\label{eq:optimal2}
    &V^*_{g 2}\left(x_{g 2, j}(t),u^*_{g 2, j}(t),m_{g 1},m_{g 2}\right)\nonumber\\
    &=\inf_{\boldsymbol{u}_{g1}} \sup _{\boldsymbol{u}_{g2}}V_{g 2}\left(x_{g 2, j}, u^*_{g 2, j},m_{g 1},m_{g 2}\right)\nonumber\\
    &\leq V_{g 2}\left(x_{g 2, j}, u_{g 2, j},m_{g 1},m_{g 2}\right)
\end{align}
where $u^*_{g2, j}\in\boldsymbol{u}_{g2}$.

\begin{figure}
    \centering
    \includegraphics[width=1\linewidth]{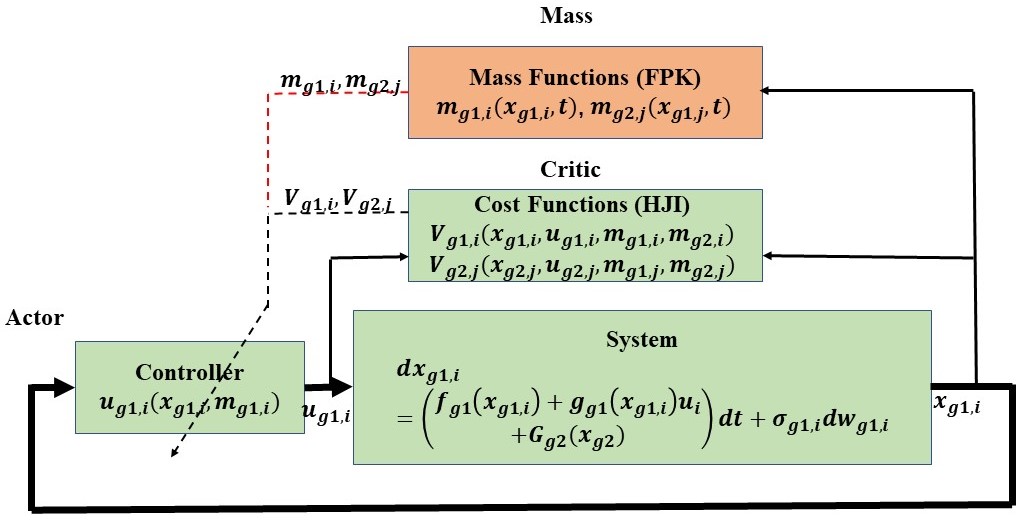}
    \caption{The structure of the ACM algorithm for the pursuers}
    \label{fig:diagramm_acm}
\end{figure}

\section{Mean Field Optimal Strategy for Massive MAS Pursuit-evasion Game}
In this section, the ACM algorithm is introduced in detail. The structure of the proposed algorithm for pursuers is illustrated in \ref{fig:diagramm_acm}.  To obtain the optimal strategies for the agents in two groups (i.e., (\ref{eq:optimal1}) and (\ref{eq:optimal2})), the Mean Field Game theory and Adaptive Dynamic Programming (ADP) has been adopted. The Mean Field Game theory can estimate the mass in (\ref{eq:bare_cost1}) and (\ref{eq:bare_cost2}) by the Fokker-Planck-Kolmogorov (FPK) equation \cite{faguo1}. Inspired by the most recent mean-field game approach such as \cite{zhou2019intelligent} and \cite{zhou2019decentralized} , a coupled HJI-mutli-FPKs equations has been constructed in (\ref{eq:hji1})-(\ref{eq:fpk2}) for obtaining the optimal strategy with large population of multi-agent system. The $H(\cdot)$ functions in (\ref{eq:hji1}) and (\ref{eq:hji2}) are the Hamiltonian which can be represented as:
\begin{align}
    &H_{g1}\left(x_{g 1, i}, \partial V_x\left(x_{g_{1}, i},u_{g1,i}\right)\right)\nonumber\\
    &=\Phi_{g 1}\left(m_{g 1}, x_{g 1, i}\right)-\Phi_{g 2}\left(m_{g 2}, x_{g 1, i}\right)\nonumber\\
    &+x_{g 1, i}^TQ_{g1}x_{g 1, i}+u_{g 1}^{T} R_{g 1, i} u_{g 1, i}+\partial V_x\left(x_{g_{1}, i},u_{g1,i}\right)\dot{x}_{g1,i}\\
    &H_{g2}\left(x_{g 2,j}, \partial V_x\left(x_{g_{2},j},u_{g2,j}\right)\right)\nonumber\\
    &=\Phi_{g 2}\left(m_{g 2}, x_{g 2,j}\right)-\Phi_{g 1}\left(m_{g 1}, x_{g 2, j}\right)\nonumber\\
    &+x_{g 2,j}^TQ_{g2}x_{g 2,j}+u_{g 2}^{T} R_{g 2,j} u_{g 2,j}+\partial V_x\left(x_{g_{2},j},u_{g2,j}\right)\dot{x}_{g2,j}
\end{align}
where $\Phi_{g 1}\left(m_{g 1}, x_{g 1, i}\right)$ and $\Phi_{g 2}\left(m_{g 2}, x_{g 2,j}\right)$ are the mean field function which calculates the affect from all other agents in the same group.

It has been shown by numerous studies (e.g. \cite{nourian2012mean}) that the solution of the coupled HJI-multi-FPKs equations yields the $\varepsilon-$Nash equilibrium, i.e.:
\begin{align}
    &\nonumber V_{g 1}\left(x_{g 1, i},u^{*}_{g1, i},u_{g1,-i}\right)<V_{g 1}\left(x_{g 1, i},u_{g1, i},u_{g1,-i}\right)+\varepsilon_N\\
    &\nonumber V_{g 2}\left(x_{g 2, j},u^{*}_{g2, j},u_{g2,-j}\right)<V_{g 2}\left(x_{g 2, j},u_{g2, j},u_{g2,-j}\right)+\varepsilon_N
\end{align}
where $\varepsilon_N$ is the error that goes to zero as $N$ goes to infinity \cite{nourian2012mean} thus yield (\ref{eq:optimal1}) and (\ref{eq:optimal2}).

Similar to \cite{vamvoudakis2012online}, the optimal control for agents in two groups can be solved separately as:
\begin{align}
    &\label{eq:optimal_control1}u^*_{g1,i}(x_{g1,i})\nonumber\\
    &=-\frac{1}{2} R^{-1}_{g1} g_{g 1}^{T}\left(x_{g 1, i}\right) \frac{\partial V_{g 1, i}\left(x_{g 1, i}, u_{g 1,i},m_{g1},m_{g2}\right)}{\partial x_{g 1, i}}\\
    &\label{eq:optimal_control2}u^*_{g2,j}(x_{g2,j})\nonumber\\
    &=-\frac{1}{2} R^{-1}_{g1} g_{g 2}^{T}\left(x_{g 2, j}\right) \frac{\partial V_{g 2, j}\left(x_{g 2, j,} u_{g 2,j},m_{g1},m_{g2}\right)}{\partial x_{g 2, j}}
\end{align}

\begin{remark}
To obtain the optimal control, the coupled HJI-multi-FPKs equations need to be solved simultaneously. However, the HJI equations ((\ref{eq:hji1}) and (\ref{eq:hji2})) as well as the FPK equation ((\ref{eq:fpk1}) and (\ref{eq:fpk2})) are two complicate infinite-dimensional Partial Differential Equations (PDEs) whose solutions are difficult to solve analytically. Therefore, inspired by adaptive dynamic programming (ADP) and reinforcement learning techniques, a novel neural network based Actor-Critic-Mass algorithm has been developed to learn the coupled HJI-multi-FPKs equations' solution online in this paper.
\end{remark}

\section{Actor-Critic-Mass Based Optimal Pursuit-evasion Strategy Design}
\subsection{Optimal ACM estimator design}
The proposed reinforcement learning ADP algorithm can be implemented into an Actor-Critic-Mass structure which consists of five neural networks for individual agent. For the pursuer agents in the group $\mathcal{G}_{1}$, the actor neural network is utilized to approximate the solution of optimal control (i.e. (\ref{eq:optimal_control1})); the critic is designed to approximate the solution of the HJI equation (i.e. (\ref{eq:hji1})), and the mass neural network is employed to approximate the solution of the FPK equation (i.e. (\ref{eq:fpk1})). Except for the three neural networks for the group $\mathcal{G}_{1}$, the pursuers also needs to estimate the optimal value function, mass, and optimal strategy for evaders since the estimated states and optimal strategy of evaders are also considered in the cost function symmetrically. Similarly, the evader agents in the group $\mathcal{G}_{2}$ admits the same neural network structure and update laws so we will use the agents in the group $\mathcal{G}_{1}$ only to illustrate the controller design.

According to the universal approximation theory of neural network (NN) \cite{b9}, the optimal cost function, decentralized strategy and mass distribution function for pursuers can be approximated as:
\begin{equation}\label{eq:NN_estimate_formula1}
\left\{\begin{matrix}
\begin{aligned}
&\hat{V}_{g1,i}\left(x_{g 1, i}, \hat{u}_{g 1,i},\hat{m}_{g 1,i},\hat{m}_{g 2,i}\right)\\
&=\hat{W}_{V, g 1, i}^{T} \phi_{V, g 1, i}\left(x_{g 1, i},\hat{m}_{g 1,i},\hat{m}_{g 2,i}\right)\\ 
&\hat{u}_{g 1, i}(x_{g 1, i}(t))=\hat{W}_{u 1, i}^{T}(t) \phi_{u, g 1, i}\left(x_{g 1, i}, \hat{m}_{g 1, i}, t\right)\\ 
&\hat{m}_{g 1, i}(x_{g 1, i},t)=\hat{W}_{m, g 1, i}^{T}(t) \phi_{m, g 1, i}\left(x_{g 1, i}, t\right)
\end{aligned}
\end{matrix}\right.
\end{equation}

Besides estimating the evaders' mass distribution required in (\ref{eq:NN_estimate_formula1}), the pursuers also need to maintain two neural networks for the evaders' optimal cost function and mass distribution, i.e.,
\begin{equation}\label{eq:NN_estimate_formula2}
\left\{\begin{matrix}
\begin{aligned}
&\hat{V}_{g2,i}\left(x_{g 2, i}, \hat{u}_{g 2,i},\hat{m}_{g 1,i},\hat{m}_{g 2,i}\right)\\
&=\hat{W}_{V, g 2, i}^{T} \phi_{V, g 2, i}\left(x_{g 2, i},\hat{m}_{g 1,i},\hat{m}_{g 2,i}\right)\\ 
&\hat{u}_{g 2, i}(x_{g 2, i}(t))=\hat{W}_{u 1, i}^{T}(t) \phi_{u, g 2, i}\left(x_{g 2, i}, \hat{m}_{g 2, i}, t\right)\\ 
&\hat{m}_{g 2, i}(x_{g 2, i},t)=\hat{W}_{m, g 2, i}^{T}(t) \phi_{m, g 2, i}\left(x_{g 2, i},  t\right)
\end{aligned}
\end{matrix}\right.
\end{equation}

Substituting (\ref{eq:NN_estimate_formula1}) into (\ref{eq:hji1}), (\ref{eq:optimal_control1}), and (\ref{eq:fpk1}), equations will not hold. The residual errors will be computed and used to tune the actor, critic, and mass NNs along with time, i.e.
\begin{align}
  &e_{HJI1,i}\label{eq:critic_error1_d}=\Phi_{g 1, i}\left(m_{g 1}, x_{g 1, i}\right)-\Phi_{g 2, i}\left(m_{g 2}, x_{g 1, i}\right)\nonumber\\
  &+\hat{W}^T_{V,g1,i}(t)\hat{\Psi}_{V,g1,i}\\
  &e_{FPK1,i}= \begin{aligned}[t] \label{eq:mass_error1_d}
      &\hat{W}^T_{m,g1,i}(t)\hat{\Psi}_{m,g1,i}
       \end{aligned}\\
       &e_{u1,i}= \begin{aligned}[t] \label{eq:control_error1}
      &\hat{W}_{m,g1,i}^T(t)\phi_{u,g1,i}(x_i,\hat{m}_i,t)+\frac{1}{2}R^{-1}_{g1}(x_i)\partial_x\hat{\phi}_{V,g1,i}
       \end{aligned}
\end{align}
where 
\begin{align*}
  &\hat{\Psi}_{V,g1,i} = \begin{aligned}[t]\partial_t \hat{\phi}_{V,g1,i}+\frac{\sigma_{g1,i}^2}{2}\partial_{xx}\hat{\phi}_{V,g1,i}-\hat{H}_{WV}
  \end{aligned}\\
  &\hat{\Psi}_{m,g1,i}=\partial_t \phi_{m,g1,i}-\frac{\sigma_{g1,i}^2}{2}\partial_{xx} \phi_{m,g1,i}-div(\phi_{m,g1,i}D_p\hat{H})  
\end{align*}
with $\hat{\phi}_{V,g1,i}=\phi_{V, g 1, i}\left(x_{g 1, i}, \hat{m}_{g 1, i}, t\right)$,
  $\hat{H}=H\left(x_{g 1, i}, \partial_x (\hat{W}^T_{V,g1,i}\hat{\phi}_{V,g1,i})\right)$ and $\hat{H}_{WV}$ being the left term such that  $\hat{H}=\hat{W}_{V,g1,i}^T\hat{H}_{WV}$.

Next, submitting (\ref{eq:NN_estimate_formula2}) into (\ref{eq:hji2}) and (\ref{eq:optimal_control2}), one obtains:
\begin{align}
  &e_{HJI2,i}\label{eq:critic_error1_3}=\Phi_{g 2, i}\left(m_{g 2}, x_{g 2, i}\right)-\Phi_{g 1, i}\left(m_{g 1}, x_{g 2, i}\right)\nonumber\\
  &+\hat{W}^T_{V,g2,i}(t)\hat{\Psi}_{V,g2,i}\\
  &e_{FPK2,i}= \begin{aligned}[t] \label{eq:mass_error2_d}
      &\hat{W}^T_{m,g2,i}(t)\hat{\Psi}_{m,g2,i}
       \end{aligned}
\end{align}
where $\hat{\Psi}_{V,g2,i}$ and $\hat{\Psi}_{m,g2,i}$ is similarly defined as in (\ref{eq:critic_error1_d}) and (\ref{eq:mass_error1_d}).

By applying the the gradient descent algorithm, the ACM NNs' update laws can be derived as
\begin{align}
  &\text{C}\text{ritic NN-1: }\hat{\dot{W}}_{V g 1, i}=-\alpha_{h,g1,i}\frac{\hat{\Psi}_{V,g1,i}e^T_{HJI1,i}}{1+\|\hat{\Psi}_{V,g1,i}\|^2} \label{eq:critic1_d}\\
   &\text{M}\text{ass NN-1: }\hat{\dot{W}}_{m, g 1, i}=-\alpha_{m,g1,i}\frac{\hat{\Psi}_{m,g1,i}e^T_{FPK1,i}}{1+\|\hat{\Psi}_{m,g1,i}\|^2}\label{eq:mass1_d}\\
  &\text{A}\text{ctor NN-1: }\hat{\dot{W}}_{u, g 1, i}=-\alpha_{u,g1,i}\frac{\phi_{u,g1,i}(x_{g1,i},\hat{m}_{g1,i},t)e^T_{u1,i}}{1+\|\phi_{u,g1,i}(x_{g1,i},\hat{m}_{g1,i},t)\|^2}\label{eq:actor1_d}\\
  &\text{C}\text{ritic NN-2: }\hat{\dot{W}}_{V g2, i}=-\alpha_{h,g2,i}\frac{\hat{\Psi}_{V,g2,i}e^T_{HJI1,i}}{1+\|\hat{\Psi}_{V,g2,i}\|^2} \label{eq:critic2_d}\\
   &\text{M}\text{ass NN-2: }\hat{\dot{W}}_{m, g2, i}=-\alpha_{m,g2,i}\frac{\hat{\Psi}_{m,g2,i}e^T_{FPK1,i}}{1+\|\hat{\Psi}_{m,g2,i}\|^2}\label{eq:mass2_d}
\end{align}
where $\alpha_{h,g1,i}$, $\alpha_{m,g1,i}$, $\alpha_{u,g1,i}$, $\alpha_{h,g2,i}$, $\alpha_{m,g2,i}$, $\alpha_{u,g2,i}$ are the learning rates.
\begin{theorem}  \label{theorem5}
\emph{(Closed-loop Stability}) 
Given an admissible initial control input and let the actor, critic, and mass NNs weights be selected within a compact set. Moreover, the critic, actor, and mass NNs' weight tuning laws for pursuers in $\mathcal{G}_{1}$ are given as (\ref{eq:critic1_d}), (\ref{eq:critic2_d}), (\ref{eq:actor1_d}), (\ref{eq:mass1_d}), and (\ref{eq:mass2_d}), respectively. Then, there exists constants $\alpha_{h,g1,i}$, $\alpha_{m,g1,i}$, $\alpha_{u,g1,i}$, $\alpha_{h,g2,i}$, $\alpha_{m,g2,i}$, $\alpha_{u,g2,i}$, such that the system states $x_{g1,i}$, actor, critic, and  mass NNs weights estimation errors, $\Tilde{W}_{V,g1,i}$, $\Tilde{W}_{m,g1,i}$, $\Tilde{W}_{u,g1,i}$, $\Tilde{W}_{V,g2,i}$, $\Tilde{W}_{m,g2,i}$, and $\Tilde{W}_{u,g2,i}$ are all uniformly ultimately bounded (UUB). In addition, the estimated cost function, mass function and control inputs are all UUB. If the number of neurons and NN architecture has been designed effectively, those NN reconstruction error can be as small as possible and trivial. Furthermore, the system states $x_{g1,i}$, actor, critic, and  mass NNs weights estimation errors, $\Tilde{W}_{V,g1,i}$, $\Tilde{W}_{m,g1,i}$, $\Tilde{W}_{u,g1,i}$, $\Tilde{W}_{V,g2,i}$, $\Tilde{W}_{m,g2,i}$, and $\Tilde{W}_{u,g2,i}$ will still be asymptotically stable. 
\end{theorem}
\begin{proof}
Omitted due to page limitation.
\end{proof}

\section{Simulation Results}
\begin{figure} \centering 
\subfigure[$t=0s$] {
\includegraphics[width=0.46\linewidth]{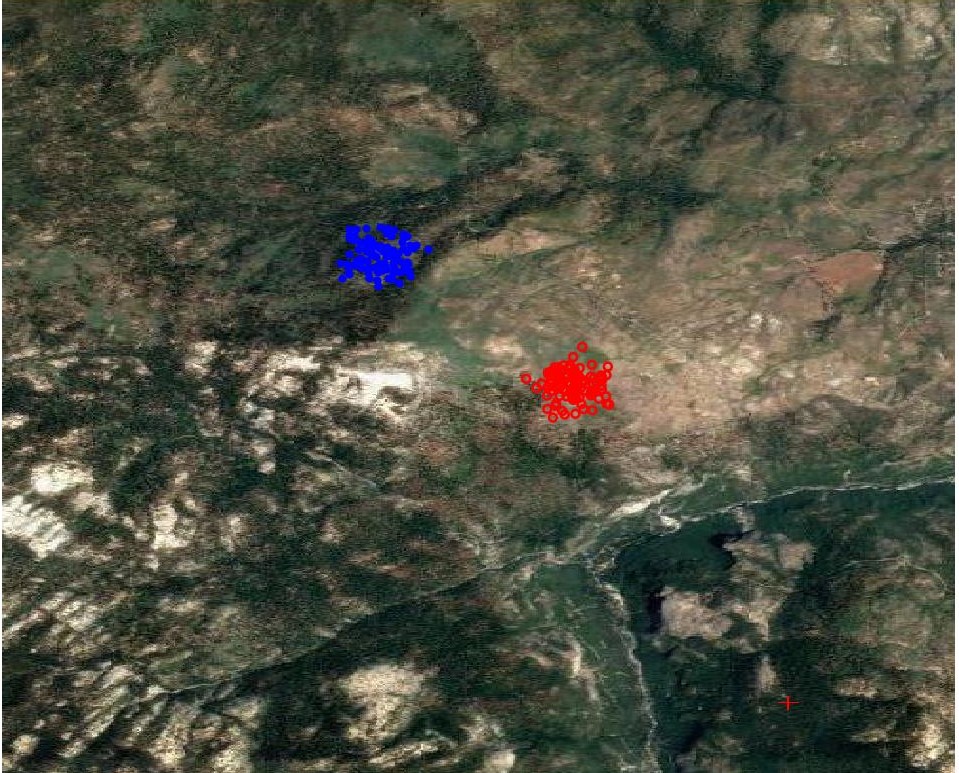}
}
\subfigure[$t=5s$] { 
\includegraphics[width=0.46\linewidth]{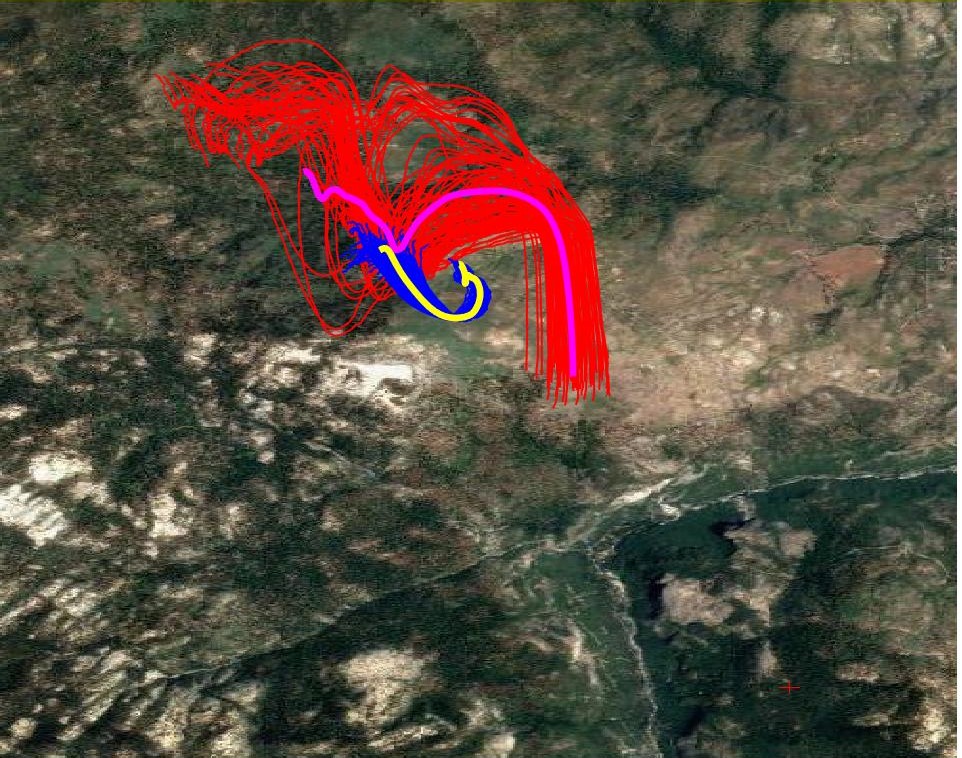} 
}
\subfigure[$t=70s$] {
\includegraphics[width=0.46\linewidth]{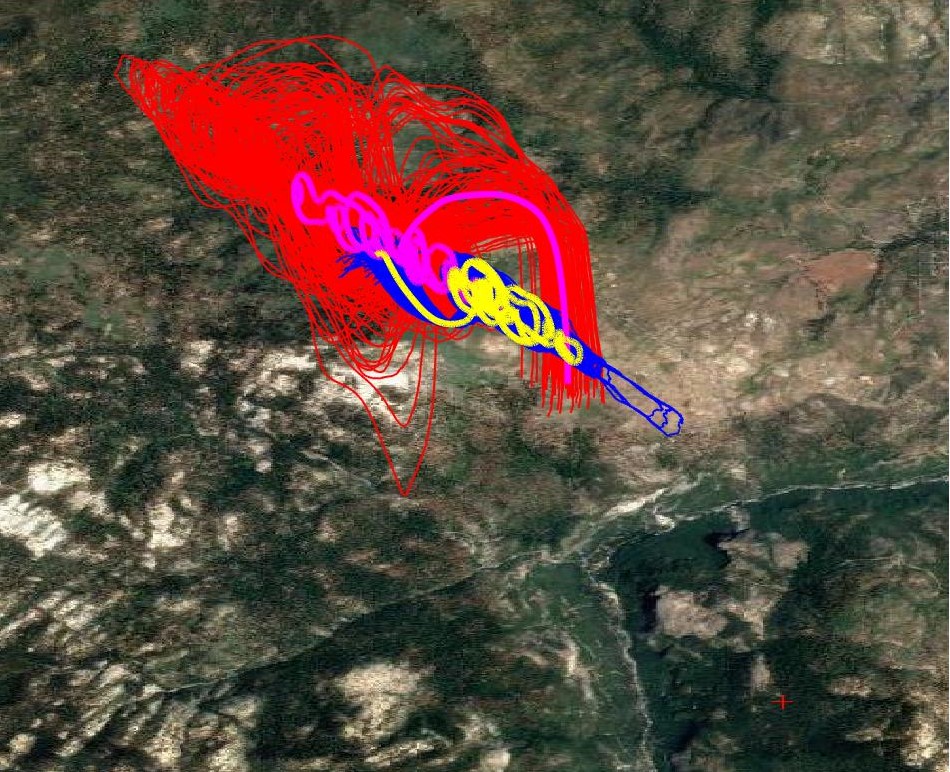}
} 
\subfigure[$t=100s$] { 
\includegraphics[width=0.46\linewidth]{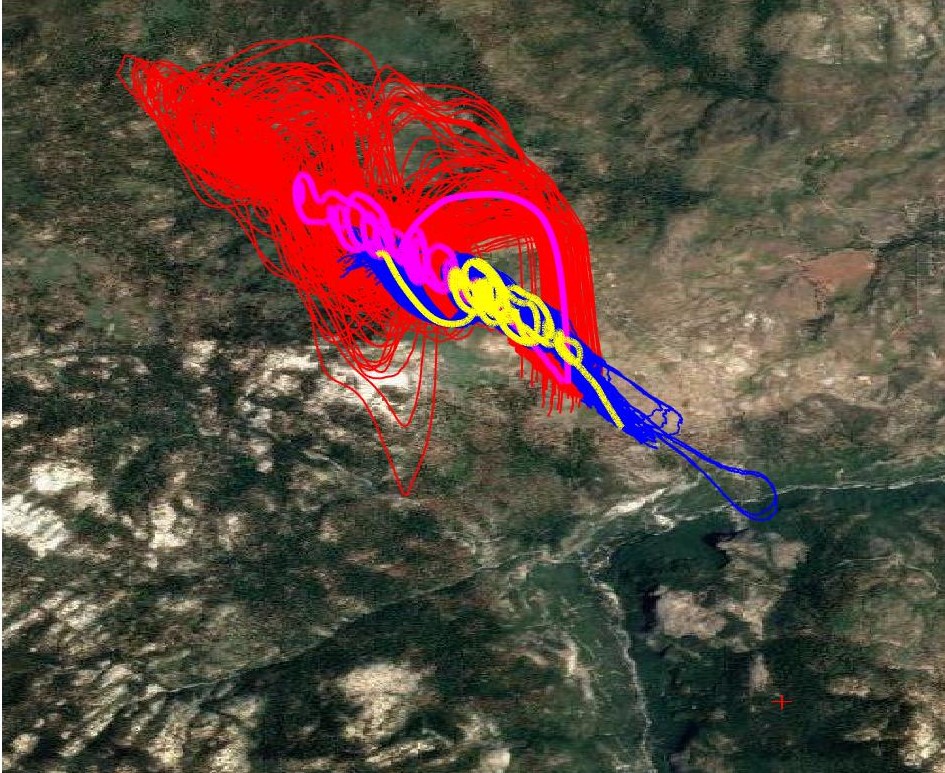} 
}
\caption{Evolution of the overall trajectory at different times. The blue and red curves represent the trajectory of all pursuers and evaders respectively. The magenta and yellow curve represent the average trajectory of pursuers and evaders respectively.}
\label{fig:overal_traj}
\end{figure}

In this section, the proposed decentralized adaptive pursuit evasion strategy has been evaluated under the noised environment. The map we use is the 2-D map of the Yosemite valley in California. A total of 100 pursuer UAVs and 100 evader ground vehicles were employed, with initial velocities set to zero, and positions randomly distributed on the map. The pursuer UAVs intended to intercept the ground vehicles while the evader ground vehicles do the opposite In this paper, we defined a successful interception as the overlap of the centers of  the two groups (i.e. $ \exists t\in[0,T],\quad s.t. \frac{1}{N}\sum_{j=1}^Nx_{g2,j}(t)=\frac{1}{N}\sum_{i=1}^Nx_{g1,i}(t)$).

To demonstrate the effectiveness of the proposed algorithm, we limit each agent's observation ability so that only his own states can be observed. Moreover, all agents are not allowed to communicate in this experiment set.

The nonlinear stochastic system dynamics functions for pursuers are selected as 
\begin{align}
    f_{g1}(x)=\begin{bmatrix}
-x_1+x_2\\ 
-\frac{1}{2}x_1^2-\frac{1}{2}x_2
\end{bmatrix},\quad\quad g_{g1}(x)=\begin{bmatrix}
0\\ 
1
\end{bmatrix}
\end{align}
where $x=[x_1\quad x_2]^T\in\mathbb{R}^2$ represents the agent's position.

The evaders' affect function $\boldsymbol{G}_{g 2}\left(\boldsymbol{x}_{g 2}\right)$ is defined as the average position, i.e., 
\begin{align}
    \boldsymbol{G}_{g 2}\left(\boldsymbol{x}_{g 2}(t)\right)=\frac{1}{N}\sum_{j=1}^Nx_{g2,j}(t)\approx  \mathbb{E}[m_{g2}]
\end{align}
where $m_{g2}$ is the mass function (i.e. probability distribution function of states) for evaders. When $N\rightarrow\infty$, the approximately equal sign can be replaced by equal sign.

Next, the system dynamics functions for evaders are selected as
\begin{align}
       f_{g2}(x)=\begin{bmatrix}
x_1+2x_2\\ 
2x_1+x_2
\end{bmatrix},\quad\quad g_{g2}(x)=\begin{bmatrix}
1\\ 
2
\end{bmatrix} 
\end{align}
Similarly, the pursuers' affect function is defined as
\begin{align}
    \boldsymbol{G}_{g 1}\left(\boldsymbol{x}_{g 1}(t)\right)=\frac{1}{N}\sum_{j=1}^Nx_{g1,j}(t)\approx\int_{\Theta1} x_{g1}m_{g1}dx_{g1}
\end{align}

The diffusion rate in (\ref{eq:system_dynamics1_d}) and (\ref{eq:system_dynamics2_d}) are set to $0.02$ for all agents in both groups. The Mean Field coupling functions in (\ref{eq:bare_cost1}) and  (\ref{eq:bare_cost2}) are defined as
\begin{align*}
    &\Phi_{g1}(m_{g1},x_{g1,i})=\left\|x_{g1,i}-\mathbb{E}[m_{g1}(x_{g1,i},t)]\right\|^2\\
    &\Phi_{g2}(m_{g2},x_{g2,i})=\left\|x_{g2,i}-\mathbb{E}[m_{g2}(x_{g2,i},t)]\right\|^2
\end{align*}
where functions $\Phi_{g1}$ and $\Phi_{g2}$ drive each individual agent to keep cohesion with their population center.  The parameters in the cost functions are selected as $Q_{g1}=Q_{g2}=2I_2$, and $R_{g1}=R_{g2}=2I_2$.

The agents' initial positions were randomly generated by a 2-variant normal distribution. Furthermore, to estimate the solution of HJI equations (i.e., (\ref{eq:hji1}) and (\ref{eq:hji2})), FPK equations (i.e., (\ref{eq:fpk1}) and (\ref{eq:fpk2})), and optimal control input (i.e., (\ref{eq:optimal_control1})), 2 critic NNs, 2 mass NNs, and an actor NN are constructed. Additionally, a random noise is injected to the control input from $0s$ to $50s$ to increase the NN approximators' exploration.

\begin{figure} \centering 
\includegraphics[width=1\linewidth]{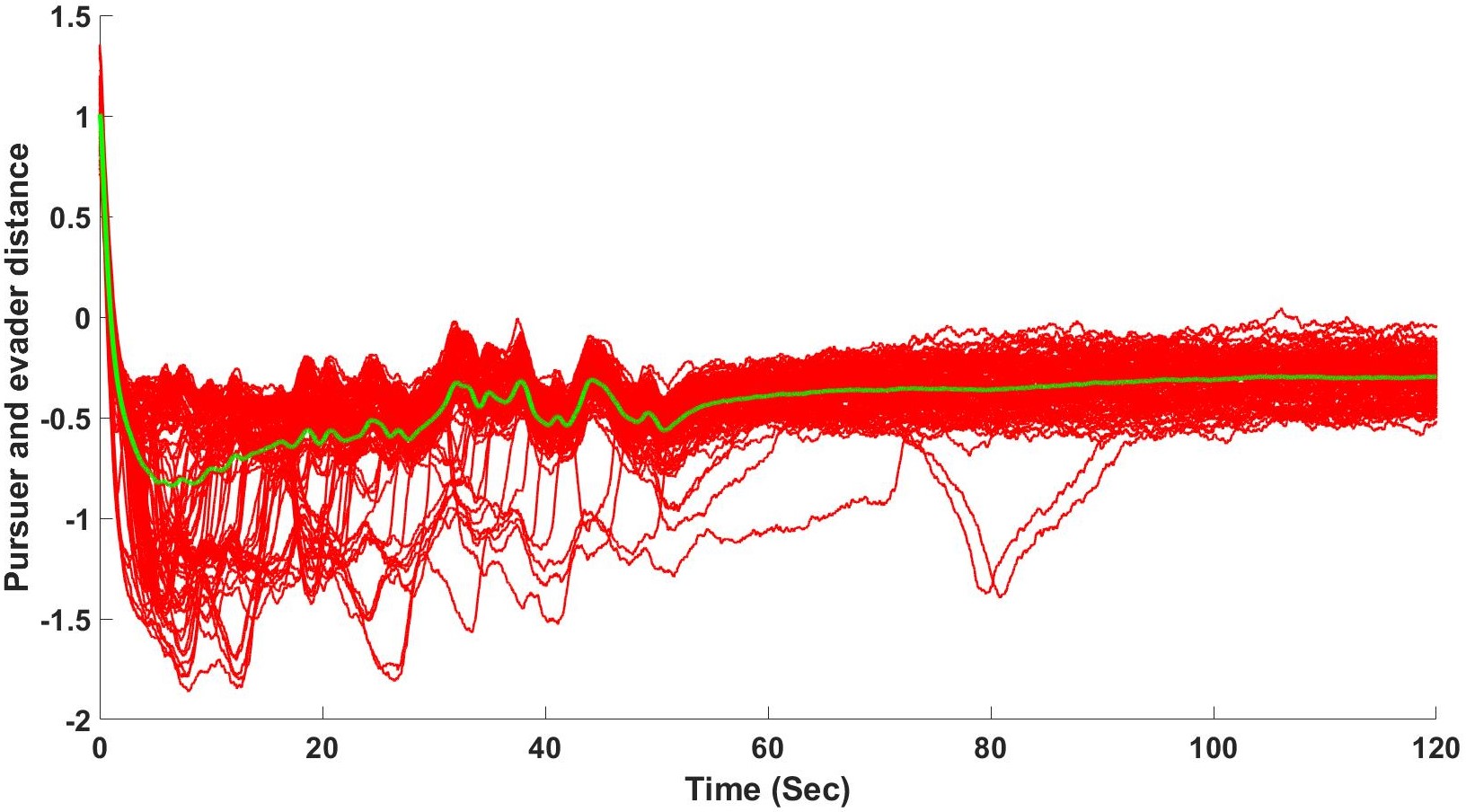}
\caption{States difference of pursuers and evaders. The red curve represent each agent's distance and the green curve is the average distance }
\label{fig:state_error}
\end{figure}

\begin{figure} \centering  
\subfigure[Pursuer 1's FPK equation error] {
\includegraphics[width=1\linewidth]{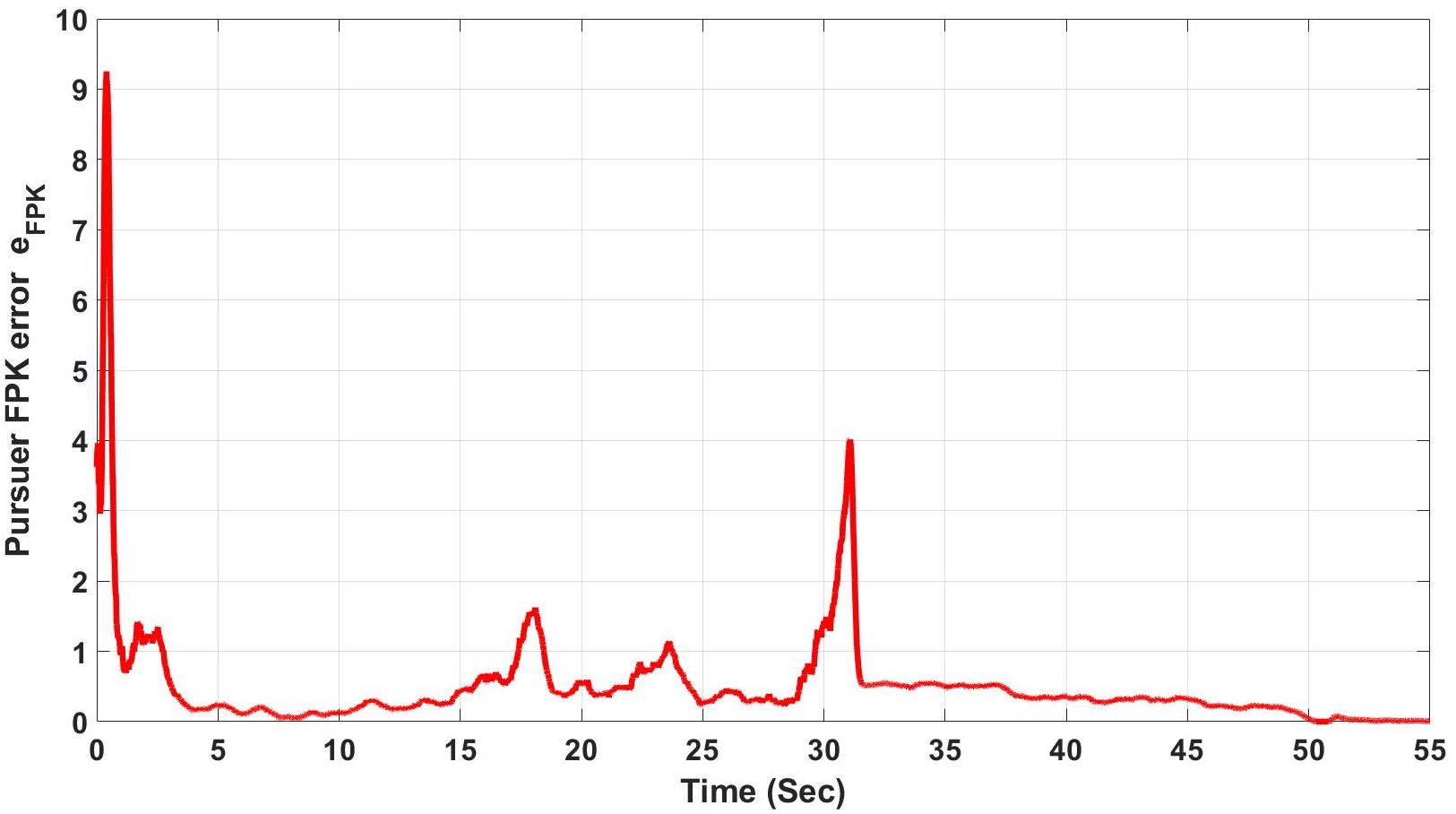}}
\subfigure[Evader 1's FPK equation error] { 
\includegraphics[width=1\linewidth]{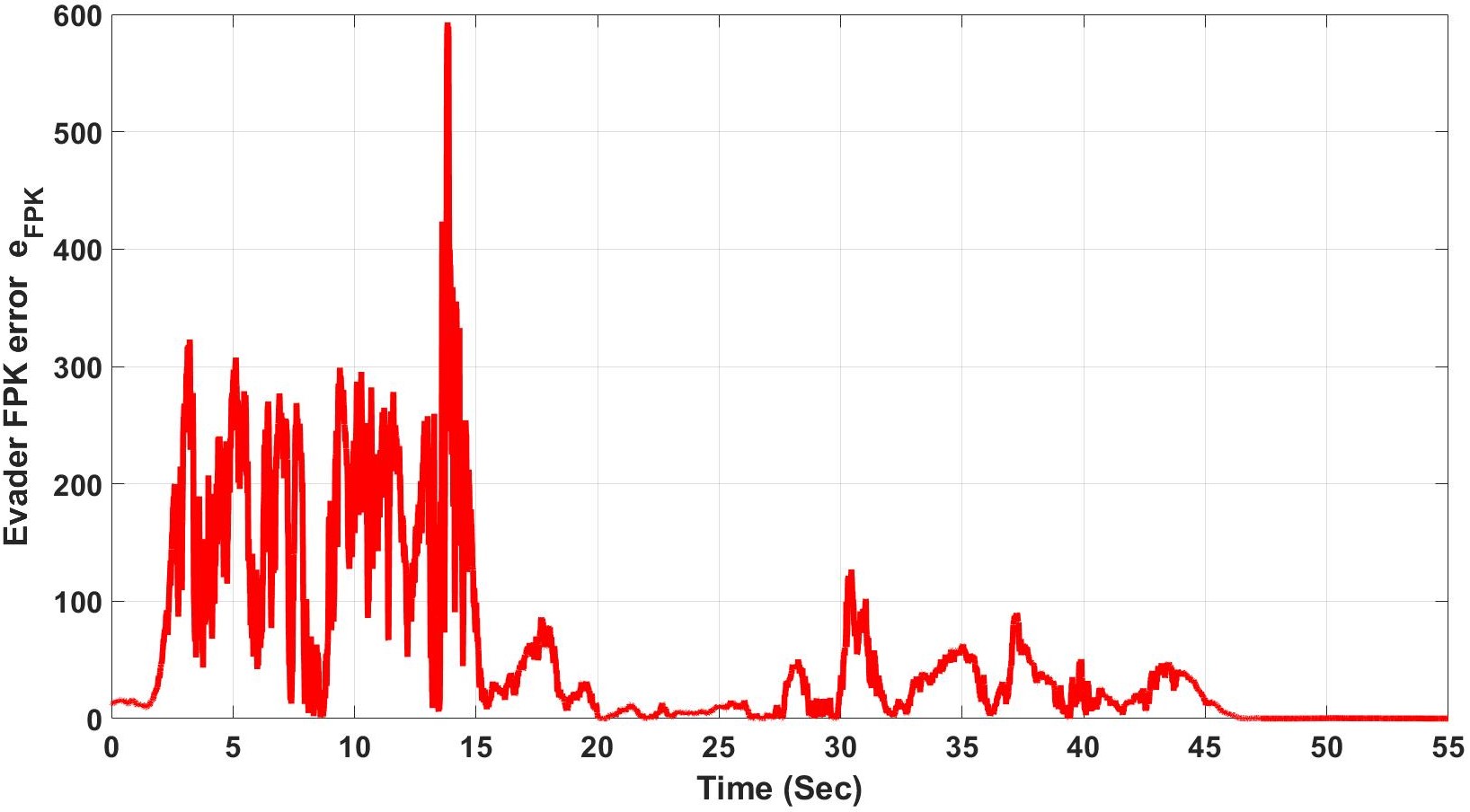} 
}
\caption{FPK equation errors of pursuer 1 and evader 1}
\label{fig:FPK_error}
\end{figure}

The overall trajectory of the pursuers and evaders at different time instants are shown in Fig. \ref{fig:overal_traj}. The initial positions are first shown in Fig. \ref{fig:overal_traj}(a). Then the agents' positions at $5s$, $70s$, and $100s$ are plotted in  Fig. \ref{fig:overal_traj}(b)-(d). From  Fig. \ref{fig:overal_traj}, it's not difficult to observe that the pursuers are able to track the evader and the evaders can escape successfully. However, after $70s$, the distances between pursuers' and evaders' remain similar until the game ends. The reason is that the equilibrium point between two groups (i.e. saddle point of cost function) is reached. We will further analysis the equilibrium point from two aspects: 1) the distance between en two groups, 2) the coupled HJI-multi-FPKs equation error.

Firstly, the distance in x axis between pursuers and evaders are plotted in Fig. \ref{fig:state_error}. The distance in this figure is defined as:
\begin{align*}
    &\text{Individual distance: }\xi_i(t)=x_{g1,i}(t)-x_{g2,i}(t)\\
    &\text{Average distance: }\bar{\xi}(t)=\frac{1}{N}\sum^N_{i=1}x_{g1,i}(t)-\frac{1}{N}\sum^N_{j=1}x_{g2,j}(t)
\end{align*}
The green curve (i.e. average difference) in Fig. \ref{fig:state_error} demonstrates that after $80s$, neither the pursuers nor the evaders can benefit their groups by changing the strategies. This stable point proves that the saddle point (i.e. Nash equilibrium) of the cost function is achieved.

Secondly, the Nash equilibrium point is further examined by the error of the HJI equations (\ref{eq:critic_error1_d}) (\ref{eq:critic_error1_3}). Due to the limit of this paper's size, we only plot pursuer 1 's HJI equation errors in Fig. \ref{fig:HJI_error}. From Fig. \ref{fig:HJI_error} we can clearly observe that the HJI equation errors are bounded near zero after about 53 seconds. The convergence of HJI equation error indicates that the optimal cost function (i.e. Nash equilibrium) is approximated by the critic NN successfully.  
\begin{figure} \centering 
\includegraphics[width=1\linewidth]{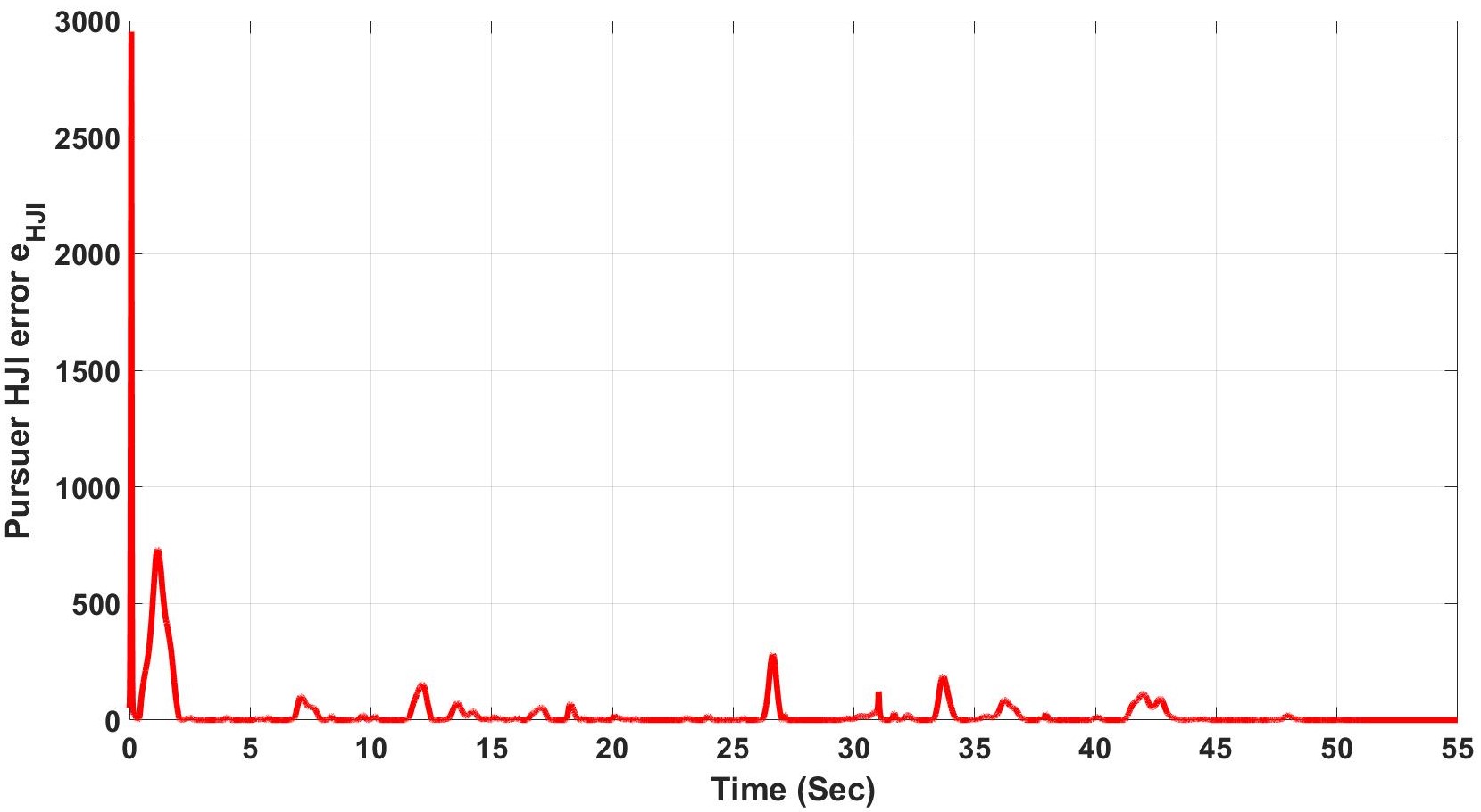}
\caption{HJI equation errors of pursuer 1}
\label{fig:HJI_error}
\end{figure}

Finally, the mass NN's performance is shown by the FPK errors (i.e. equation (\ref{eq:mass_error1_d}) and \ref{eq:mass_error2_d})) plot in Fig. \ref{fig:FPK_error}. Similarly to the HJI equation errors, we only plot pursuer 1's and evader 1's HJI error for convenience. Figure \ref{fig:FPK_error} shows that the FPK equation error converges near zero after $55s$ for both agents. The convergence of both FPK equations and HJI equations proves that a good approximation of the optimal cost function, group population distribution (i.e. mass) has been successfully obtained by the proposed ACM algorithm. Moreover, both the distance and HJI-multi-FPKs equations' error demonstrated the $\epsilon-$ Nash equilibrium point is reached. This proves the fact that the online ACM algorithm can effectively solve the decentralized optimal control for massive multi-agent persuit-evasion games.

\section{Conclusions}
In this paper, the decentralized optimal pursuit-evasion strategies with two large scale groups of pursuers and evaders has been investigated. A novel online Actor-Critic-Mass (ACM) algorithm with five neural networks are designed for individual agent to calculate the decentralized optimal strategy which satisfy the saddle point of the cost function between groups and the $\epsilon-$ Nash equilibrium in the group. The five neural networks can effectively approximate the solution of the HJI equation, the population mass (i.e. the solution of FPK equation), the decentralized optimal control, estimate the mass of the other group, and sample the value function of the opponent's group. The proposed strategy can effectively tackle the ``Curse of dimentionality" as well as eliminating the problem of communication limitation for massive MAS. Moreover, a series of numerical simulations has been conducted to demonstrate the optimality of the strategy. In the future, a pursuer group based on massive UASs will be designed as a testbed to further evaluate the performance of the proposed decentralized optimal pursuit-evasion strategy for massive MAS. 

\appendix[Proof and details]
We discuss the optimal ACM as a pursuer in this manuscript but the evaders can be similarly analyzed. A virtual evader is considered in this manuscript and will be abbreviated as ``evader". Given the system dynamics:
\begin{align}
&\text{Pursuers: }\label{eq:system_dynamics1_p}dx_{g 1, i}=\big[f_{g 1}\left(x_{g 1, i}\right)+g_{g 1}\left(x_{g 1, i}\right) u_{g 1, i}+\boldsymbol{G}_{g 2}\left(\boldsymbol{x}_{g 2}\right)\big]dt+\sigma_{g 1,i} d w_{g 1, i}\\
&\text{Evaders: }\label{eq:system_dynamics2_p}dx_{g 2, j}=\big[f_{g 2}\left(x_{g 2, j}\right)+g_{g 2}\left(x_{g 2, j}\right) u_{g 2, j}+\boldsymbol{G}_{g 1}\left(\boldsymbol{x}_{g 1}\right)\big]dt+\sigma_{g 2,j} d w_{g 2, j}
\end{align}

Neural network representation:
\begin{align}
&\text{Pursuer Critic NN: }V_{g1,i}\left(x_{g 1, i}, u_{g 1,i},m_{g 1,i},m_{g 2,i}\right)=W_{V, g 1, i}^{T} \phi_{V, g 1, i}\left(x_{g 1, i},m_{g 1,i},m_{g 2,i}\right)+\varepsilon_{HJI1,i}\label{eq:critic_repre1}\\ 
&\text{Pursuer Actor NN: }u_{g 1, i}(x_{g 1, i}(t))=W_{u 1, i}^{T}(t) \phi_{u, g 1, i}\left(x_{g 1, i}, m_{g 1, i},m_{g 2, i}\right)+\varepsilon_{m1,i}\label{eq:mass_repre1}\\ 
&\text{Pursuer Mass NN: }m_{g 1, i}(x_{g 1, i},t)=W_{m, g 1, i}^{T}(t) \phi_{m, g 1, i}\left(x_{g 1, i}, t\right)+\varepsilon_{u1,i}\label{eq:actor_repre1}\\
&\text{Evader Critic NN: }V_{g2,j}\left(x_{g 2, j}, u_{g 2,j},m_{g 1,j},m_{g 2,j}\right)=W_{V, g 2, j}^{T} \phi_{V, g 2, j}\left(x_{g 2, j},m_{g 1,j},m_{g 2,j}\right)+\varepsilon_{HJI2,j}\label{eq:critic_repre2}\\ 
&\text{Evader Actor NN: }u_{g 2, j}(x_{g 2, j}(t))=W_{u 2, j}^{T}(t) \phi_{u, g 2, j}\left(x_{g 2, j}, m_{g 1, j},m_{g 2, j}\right)+\varepsilon_{m2,j}\label{eq:mass_repre2}\\ 
&\text{Evader Mass NN: }m_{g 2, j}(x_{g 2, j},t)=W_{m, g 2, j}^{T}(t) \phi_{m, g 2, j}\left(x_{g 2, j}, t\right)+\varepsilon_{u2,j}\label{eq:actor_repre2}
\end{align}
where $\varepsilon_{HJI1,i}$, $\varepsilon_{FPK1,i}$, $\varepsilon_{u1,i}$, $\varepsilon_{HJI2,j}$, $\varepsilon_{FPK2,j}$, and $\varepsilon_{u2,j}$ are the reconstruction errors which are related to the NNs' structures. Note that the evader actor NN is just for proof purpose, not maintained in implementation. 

Neural network estimation representation:
\begin{align}
&\text{Pursuer Critic NN: }\hat{V}_{g1,i}\left(x_{g 1, i}, \hat{u}_{g 1,i},\hat{m}_{g 1,i},\hat{m}_{g 2,i}\right)=\hat{W}_{V, g 1, i}^{T} \hat{\phi}_{V, g 1, i}\left(x_{g 1, i},\hat{m}_{g 1,i},\hat{m}_{g 2,i}\right)\label{eq:critic_estimate1}\\ 
&\text{Pursuer Actor NN: }\hat{u}_{g 1, i}(x_{g 1, i}(t))=\hat{W}_{u 1, i}^{T}(t) \hat{\phi}_{u, g 1, i}\left(x_{g 1, i}, \hat{m}_{g 1, i},\hat{m}_{g 2, i}\right)\label{eq:mass_estimate1}\\ 
&\text{Pursuer Mass NN: }\hat{m}_{g 1, i}(x_{g 1, i},t)=\hat{W}_{m, g 1, i}^{T}(t) \hat{\phi}_{m, g 1, i}\left(x_{g 1, i}, t\right)\label{eq:actor_estimate1} \\

&\text{Evader Critic NN: }\hat{V}_{g2,j}\left(x_{g 2, j}, u_{g 2,j},\hat{m}_{g 1,j},\hat{m}_{g 2,j}\right)=\hat{W}_{V, g 2, j}^{T} \hat{\phi}_{V, g 2, j}\left(x_{g 2, j},\hat{m}_{g 1,j},\hat{m}_{g 2,j}\right) \label{eq:critic_estimate2}\\ 
&\text{Evader Actor NN: }\hat{u}_{g 2, j}(x_{g 2, j}(t))=\hat{W}_{u 2, j}^{T}(t) \hat{\phi}_{u, g 2, j}\left(x_{g 2, j}, \hat{m}_{g 1, j},\hat{m}_{g 2, j}\right) \label{eq:mass_estimate2}\\ 
&\text{Evader Mass NN: }\hat{m}_{g 2, j}(x_{g 2, j},t)=\hat{W}_{m, g 2, j}^{T}(t) \hat{\phi}_{m, g 2, j}\left(x_{g 2, j}, t\right) \label{eq:actor_estimate2}
\end{align}

Estimation error:
\begin{align}
  &e_{HJI1,i}\label{eq:critic_error1_p}=\Phi_{g 1, i}\left(\hat{m}_{g 1,i}, x_{g 1, i}\right)-\Phi_{g 2, i}\left(\hat{m}_{g 2,i}, x_{g 1, i}\right)+\hat{W}^T_{V,g1,i}(t)\hat{\Psi}_{V,g1,i}(x_{g1,i},\hat{m}_{g1,i},\hat{m}_{g2,i})\\
  &e_{FPK1,i}=\hat{W}^T_{m,g1,i}(t)\hat{\Psi}_{m,g1,i}(x_{g1,i},\hat{V}_{g1,i})\label{eq:mass_error1_p}\\
   &e_{u1,i}=\hat{W}_{m,g1,i}^T(t)\hat{\phi}_{u,g1,i}(x_{g1,i},\hat{m}_{g1,i},\hat{m}_{g2,i})+\frac{1}{2}R^{-1}_{g1}(x_{g1,i})\partial_x\hat{V}_{g1,i}\label{eq:actor_error1}\\
   
    &e_{HJI2,j}\label{eq:critic_error2_p}=\Phi_{g 2, j}\left(\hat{m}_{g 2,j}, x_{g 2, j}\right)-\Phi_{g 2, j}\left(\hat{m}_{g 2,j}, x_{g 2, j}\right)+\hat{W}^T_{V,g2,j}(t)\hat{\Psi}_{V,g2,j}(x_{g2,j},\hat{m}_{g2,j},\hat{m}_{g2,j})\\
  &e_{FPK2,j}=\hat{W}^T_{m,g2,j}(t)\hat{\Psi}_{m,g2,j}(x_{g2,j},\hat{V}_{g2,j})\label{eq:mass_error2_p}\\
   &e_{u2,j}=\hat{W}_{m,g2,j}^T(t)\hat{\phi}_{u,g2,j}(x_{g2,j},\hat{m}_{g2,j},\hat{m}_{g2,j})+\frac{1}{2}R^{-1}_{g2}(x_{g2,j})\partial_x\hat{V}_{g2,j}\\
  
\end{align}
where 
\begin{align*}
  &\hat{\Psi}_{V,g1,i}\left(x_{g1,i},\hat{m}_{g1,i},\hat{m}_{g2,i} \right) = \begin{aligned}[t]\partial_t \hat{\phi}_{V,g1,i}+\frac{\sigma_{g1,i}^2}{2}\partial_{xx}\hat{\phi}_{V,g1,i}-\hat{H}_{WV1}
  \end{aligned}\\
  &\hat{\Psi}_{m,g1,i}\left(x_{g1,i},\hat{V}_{g1,i}\right)=\partial_t \hat{\phi}_{m,g1,i}-\frac{\sigma_{g1,i}^2}{2}\partial_{xx} \hat{\phi}_{m,g1,i}-\operatorname{div}(\hat{\phi}_{m,g1,i}D_p\hat{H})\\
  &H_{g1}\left(x_{g 1, i}, \partial V_x\left(x_{g_{1}, i},u_{g1,i}\right)\right)=\Phi_{g 1}\left(m_{g 1},m_{g 2}, x_{g 1, i}\right)-\Phi_{g 2}\left(m_{g 1},m_{g 2}, x_{g 1, i}\right)\nonumber\\
    &+x_{g 1, i}^TQ_{g1}x_{g 1, i}+u_{g 1}^{T} R_{g 1, i} u_{g 1, i}+\partial V_x\left(x_{g_{1}, i},u_{g1,i}\right)\dot{x}_{g1,i}\\
    
&\hat{\Psi}_{V,g2,i}\left(x_{g2,i},\hat{m}_{g1,i},\hat{m}_{g2,i} \right) = \begin{aligned}[t]\partial_t \hat{\phi}_{V,g2,j}+\frac{\sigma_{g2,j}^2}{2}\partial_{xx}\hat{\phi}_{V,g2,j}-\hat{H}_{WV2}
  \end{aligned}\\
  &\hat{\Psi}_{m,g2,j}\left(x_{g2,j},\hat{V}_{g2,j}\right)=\partial_t \hat{\phi}_{m,g2,j}-\frac{\sigma_{g2,j}^2}{2}\partial_{xx} \phi_{m,g2,j}-\operatorname{div}(\hat{\phi}_{m,g2,j}D_p\hat{H})\\
  &H_{g2}\left(x_{g 2, j}, \partial V_x\left(x_{g_{2}, j},u_{g2,j}\right)\right)=\Phi_{g 2}\left(m_{g 1},m_{g 2}, x_{g 2, j}\right)-\Phi_{g 1}\left(m_{g 1},m_{g 2}, x_{g 2, j}\right)\nonumber\\
    &+x_{g 2, j}^TQ_{g2}x_{g 2, j}+u_{g 2}^{T} R_{g 2, j} u_{g 2, j}+\partial V_x\left(x_{g_{2}, j},u_{g2,j}\right)\dot{x}_{g2,j}
\end{align*}
with $\hat{H}_{g1}=H_{g1}\left(x_{g 1, i}, \partial_x (\hat{W}^T_{V,g1,i}\hat{\phi}_{V,g1,i})\right)$, $\hat{H}_{g2}=H_{g2}\left(x_{g2, j}, \partial_x (\hat{W}^T_{V,g2,j}\hat{\phi}_{V,g2,j})\right)$ and $\hat{H}_{WV1}$, $\hat{H}_{WV2}$ being the left term such that  $\hat{H}_{g1}=\hat{W}_{V,g1,i}^T\hat{H}_{WV1}$, $\hat{H}_{g2}=\hat{W}_{V,g2,j}^T\hat{H}_{WV2}$.

The update law for neural networks:
\begin{align}
  &\text{C}\text{ritic NN-pursuer: }\hat{\dot{W}}_{V g 1, i}=-\alpha_h\frac{\hat{\Psi}_{V,g1,i}e^T_{HJI1,i}}{1+\|\hat{\Psi}_{V,g1,i}\|^2} \label{eq:critic1_p}\\
   &\text{M}\text{ass NN-pursuer: }\hat{\dot{W}}_{m, g 1, i}=-\alpha_m\frac{\hat{\Psi}_{m,g1,i}e^T_{FPK1,i}}{1+\|\hat{\Psi}_{m,g1,i}\|^2}\label{eq:mass1_p}\\
  &\text{A}\text{ctor NN-pursuer: }\hat{\dot{W}}_{u, g 1, i}=-\alpha_u\frac{\hat{\phi}_{u,g1,i}(x_{g1,i},\hat{m}_{g1,i},\hat{m}_{g2,i})e^T_{u1,i}}{1+\|\phi_{u,g1,i}(x_{g1,i},\hat{m}_{g1,i})\|^2}\label{eq:actor1_p}\\
  
    &\text{C}\text{ritic NN-evader: }\hat{\dot{W}}_{V g 2, j}=-\alpha_h\frac{\hat{\Psi}_{V,g2,j}e^T_{HJI2,j}}{1+\|\hat{\Psi}_{V,g2,j}\|^2} \label{eq:critic2_p}\\
   &\text{M}\text{ass NN-evader: }\hat{\dot{W}}_{m, g 2, j}=-\alpha_m\frac{\hat{\Psi}_{m,g2,j}e^T_{FPK2,j}}{1+\|\hat{\Psi}_{m,g2,j}\|^2}\label{eq:mass2_p}\\
  &\text{A}\text{ctor NN-evader: }\hat{\dot{W}}_{u, g 2, j}=-\alpha_u\frac{\hat{\phi}_{u,g2,j}(x_{g2,j},\hat{m}_{g1,j},\hat{m}_{g2,j})e^T_{u2,j}}{1+\|\phi_{u,g2,j}(x_{g2,j},\hat{m}_{g2,j})\|^2}\label{eq:actor2}
\end{align}

Because each agent is homogeneous, we drop the subscript of the agent number $i$ and make the following simplification on the notation, $x_{g 1, i}\rightarrow x_1$,
     $f_{g 1}\left(x_{g 1, i}\right)\rightarrow f_1(x_1)$,
     $g_{g 1}\left(x_{g 1, i}\right)\rightarrow g_1(x_1)$,
     $u_{g 1, i}\rightarrow u$,
     $x_{g 1, i}\rightarrow x_1$,
     $W_{V,g 1, i}\rightarrow W_V1$,
     $W_{m,g 1, i}\rightarrow W_m1$,
     $W_{u,g 1, i}\rightarrow W_u1$,
     $\alpha_{h,i}\rightarrow \alpha_h$,
     $\hat{m}_{g1,i}\rightarrow \hat{m}_1$,
     $m_{g1,i}\rightarrow m_1$,
     $\hat{V}_{g1,i}\rightarrow \hat{V}_1$,
     $V_{g1,i}\rightarrow V_1$,
     $\hat{u}_{g1,i}\rightarrow \hat{u}_1$,
     $u_{g1,i}\rightarrow u_1$,
     $e_{HJI1,i}\rightarrow e_{HJI1}$,
     $e_{u1,i}\rightarrow e_{u1}$,
     $e_{FPK1,i}\rightarrow e_{FPK1}$,
     $\varepsilon_{HJI1,i}\rightarrow \varepsilon_{HJI1}$,
     $\varepsilon_{u1,i}\rightarrow \varepsilon_{u1}$,
     $\varepsilon_{FPK1,i}\rightarrow \varepsilon_{FPK1}$,
     $\boldsymbol{G}_{g 1}\left(\boldsymbol{x}_{g 1}\right)\rightarrow G_1$,
     $\sigma_{g 1,i}\rightarrow \sigma_1 $,
     $d w_{g 1, i} \rightarrow dw_1$,
     
    $f_{g 2}\left(x_{g 2, j}\right)\rightarrow f_2(x_2)$,
     $g_{g 2}\left(x_{g 2, j}\right)\rightarrow g_2(x_2)$,
     $u_{g 2, j}\rightarrow u$,
     $x_{g 2, j}\rightarrow x_2$,
     $W_{V,g 2, j}\rightarrow W_V2$,
     $W_{m,g 2, j}\rightarrow W_m2$,
     $W_{u,g 2, j}\rightarrow W_u2$,
     $\alpha_{h,j}\rightarrow \alpha_h$,
     $\hat{m}_{g2,j}\rightarrow \hat{m}_2$,
     $m_{g2,j}\rightarrow m_2$,
     $\hat{V}_{g2,j}\rightarrow \hat{V}_2$,
     $V_{g2,j}\rightarrow V_2$,
     $\hat{u}_{g2,j}\rightarrow \hat{u}_2$,
     $u_{g2,j}\rightarrow u_2$,
     $e_{HJI2,j}\rightarrow e_{HJI2}$,
     $e_{u2,j}\rightarrow e_{u2}$,
     $e_{FPK2,j}\rightarrow e_{FPK2}$,
     $\varepsilon_{HJI2,j}\rightarrow \varepsilon_{HJI2}$,
     $\varepsilon_{u2,j}\rightarrow \varepsilon_{u2}$,
     $\varepsilon_{FPK2,j}\rightarrow \varepsilon_{FPK2}$,
     $\boldsymbol{G}_{g 2}\left(\boldsymbol{x}_{g 2}\right)\rightarrow G_2$,
     $\sigma_{g 2,j}\rightarrow \sigma_2 $,
     $d w_{g 2, j} \rightarrow dw_2$,

\section{Convergence of Critic NN}
\begin{theorem} \label{theorem2}
\emph{(Convergence of pursuer's  Critic NN weights and optimal cost function estimations)}\label{th1}
Given the initial critic NN weights, $\hat{W}_{V1}$, in a compact set, and let the critic NN weights be updated as Eq. \ref{eq:critic1_p} shows. Then, when the critic NN tuning parameters $\alpha_{h}$ satisfies the condition, $\alpha_{h}>0$, the critic NN weights estimation error $\tilde{W}_{V1}$ and the cost function estimation error $\tilde{V}_{1}=V_1-\hat{V}_1$ will be uniformly ultimately bounded (UUB) where the boundedness can be negligible if the NN reconstruction errors are trivial. While the number of neurons and NN architecture has been designed perfectly, the NN reconstruction error can be as small as possible and trivial. Furthermore, the critic NN weights and cost function estimation errors will be asymptotically stable. 
\end{theorem}
\begin{proof}
Consider the following Lyapunov function candidate as:
\begin{align}
    L_{V1}(t)=\frac{1}{2} \operatorname{tr}\left\{\tilde{W}_{V1}^{T}(t) \tilde{W}_{V1}(t)\right\}
\end{align}
Take the first derivative on the Lyapunov function candidate, one obtains:
\begin{align}\label{eq:div_lyapunov}
    \dot{L}_{V1}(t)=\frac{1}{2} \operatorname{tr}\left\{\tilde{W}_{V1}^{T}(t) \dot{\tilde{W}}_{V1}(t)\right\}+\frac{1}{2} \operatorname{tr}\left\{\dot{\tilde{W}}_{V1}^{T}(t) \tilde{W}_{V1}(t)\right\}=\operatorname{tr}\left\{\tilde{W}_{V1}^{T}(t) \dot{\tilde{W}}_{V1}(t)\right\}
\end{align}
Substitute the critic NN weights update law into \eqref{eq:div_lyapunov}, we get
\begin{align}\label{eq20}
    \dot{L}_V1(t)=\alpha_{h} \operatorname{tr}\left\{\tilde{W}_{V1}^{T}(t) \frac{\hat{\Psi}_{V1}\left(x_1, \hat{m}_1,\hat{m}_2\right) e_{HJI1}^{T}}{1+\hat{\Psi}_{V1}^{T}\left(x_1,\hat{m}_1,\hat{m}_2 \right) \hat{\Psi}_{V1}\left(x_1, \hat{m}_1,\hat{m}_2\right)}\right\}
\end{align}

Let $\Phi(x_1,\hat{m}_1,\hat{m}_2)=\Phi_{g 1, i}\left(\hat{m}_{g 1,i}, x_{g 1, i}\right)-\Phi_{g 2, i}\left(\hat{m}_{g 2,i}, x_{g 1, i}\right)$, and $\tilde{\Phi}(x_1,m_1,m_2,\hat{m}_1,\hat{m}_2)=\hat{\Phi}(x_1,\hat{m}_1,\hat{m}_2)-\Phi(x_1,m_1,m_2)$. Substitute $\tilde{\Phi}(x_1,m_1,m_2,\hat{m}_1,\hat{m}_2)$ into critic NN's error function \eqref{eq:critic_error1_p}, we get 
\begin{align}\label{eq:temp1}
    \Phi(x_1,m_1,m_2)+\tilde{\Phi}(x_1,m_1,m_2,\hat{m}_1,\hat{m}_2)+\hat{W}_{V1}^T(t)\hat{\Psi}_{V1}\left(x_1,\hat{m}_1,\hat{m}_2 \right)=e_{HJI1}
\end{align}

Since the correct estimated optimal cost function leads to the HJI equation equals zero, we have
\begin{align}\label{eq:temp3}
    \Phi(x_1,m_1,m_2)+W_V^T(t)\Psi_{V1}\left(x_1,m_1,m_2 \right)=0
\end{align}

Substitute \eqref{eq:temp3} into  \eqref{eq:temp1}, we have
\begin{align} \label{eq:temp2}
    -W_V^T(t)\Psi_{V1}\left(x_1,m_1,m_2 \right)-\varepsilon_{HJI1}+\tilde{\Phi}(x_1,m_1,m_2,\hat{m}_1,\hat{m}_2)-\hat{W}_{V1}^T(t)\hat{\Psi}_{V1}\left(x_1,\hat{m}_1,\hat{m}_2 \right)=e_{HJI1}
\end{align}

Let $\tilde{W}_{V1}(t)=W_V(t)-\hat{W}_{V1}(t)$, and $\tilde{\Psi}_{V1}(x_1,m_1,m_2,\hat{m}_1,\hat{m}_2)=\Psi_V(x_1,m_1,m_2)-\hat{\Psi}_{V1}(x_1,\hat{m}_1,\hat{m}_2)$. After manipulating terms in \eqref{eq:temp2}, we obtain
\begin{align}\label{eq14,1}
    &-W_V^T(t)\left(\hat{\Psi}_{V1}\left(x_1,\hat{m}_1,\hat{m}_2 \right)+ \tilde{\Psi}_{V1}(x_1,m_1,m_2,\hat{m}_1,\hat{m}_2)\right)\nonumber\\
    &-\varepsilon_{HJI1}+\tilde{\Phi}(x_1,m_1,m_2,\hat{m}_1,\hat{m}_2)+\hat{W}_{V1}^T(t)\hat{\Psi}_{V1}\left(x_1,\hat{m}_1,\hat{m}_2 \right)=e_{HJI1}\nonumber\\
    &\tilde{\Phi}(x_1,m_1,m_2,\hat{m}_1,\hat{m}_2)-\tilde{W}^T_V\hat{\Psi}_{V1}\left(x_1,\hat{m}_1,\hat{m}_2\right)-W_V^T\tilde{\Psi}_{V1}(x_1,m_1,m_2,\hat{m}_1,\hat{m}_2)-\varepsilon_{HJI1}=e_{HJI1}
\end{align}
where $\varepsilon_{HJI1}$ is the error resulted from the reconstruction error.

Let's further simplify the notations as: $\hat{\Psi}_{V1}\left(x_1, \hat{m}_1,\hat{m}_2\right)\rightarrow\hat{\Psi}_{V1}$, $\tilde{\Psi}_{V1}(x_1,m_1,m_2,\hat{m}_1,\hat{m}_2)\rightarrow\tilde{\Psi}_{V1}$, $\Psi_{V1}\left(x_1, m_1,m_2\right)\rightarrow\Psi_{V1}$, $\tilde{\Phi}(x_1,m_1,m_2,\hat{m}_1,\hat{m}_2)\rightarrow\tilde{\Phi}$

Substitute \eqref{eq14,1} into \eqref{eq20}, 
\begin{align}\label{eq21}
    &\dot{L}_{V1}(t)=\alpha_{h} \operatorname{tr}\left\{\tilde{W}_{V1}^{T}(t) \frac{\hat{\Psi}_{V1} \left[\tilde{\Phi}-\tilde{W}^T_V\hat{\Psi}_{V1}-W_V^T\tilde{\Psi}_{V1}-\varepsilon_{HJI1}\right]^T}{1+\hat{\Psi}_{V1}^{T} \hat{\Psi}_{V1}}\right\}\nonumber\\
    &=\alpha_{h} \operatorname{tr}\left\{\tilde{W}_{V1}^{T}(t) \frac{\hat{\Psi}_{V1} \tilde{\Phi}^T}{1+\hat{\Psi}_{V1}^{T} \hat{\Psi}_{V1}}\right\}  -\alpha_{h} \operatorname{tr}\left\{\tilde{W}_{V1}^{T}(t) \frac{\hat{\Psi}_{V1}\hat{\Psi}^T_{V1} }{1+\hat{\Psi}_{V1}^{T} \hat{\Psi}_{V1}}\tilde{W}_{V1}(t)\right\}\nonumber\\  &-\alpha_{h} \operatorname{tr}\left\{\tilde{W}_{V1}^{T}(t) \frac{\hat{\Psi}_{V1}\tilde{\Psi}^T_{V1} }{1+\hat{\Psi}_{V1}^{T} \hat{\Psi}_{V1}}W_{V1}^{T}(t)\right\}  -\alpha_{h} \operatorname{tr}\left\{\tilde{W}_{V1}^{T}(t) \frac{\hat{\Psi}_{V1} \varepsilon_{HJI1}^T}{1+\hat{\Psi}_{V1}^{T} \hat{\Psi}_{V1}}\right\}
\end{align}

Apply Cauchy-Schwarz inequality on \eqref{eq21}, 
\begin{align}\label{eq22}
    &\dot{L}_{V1}(t)=\alpha_{h} \operatorname{tr}\left\{\tilde{W}_{V1}^{T}(t) \frac{\hat{\Psi}_{V1} \tilde{\Phi}^T}{1+\hat{\Psi}_{V1}^{T} \hat{\Psi}_{V1}}\right\}  -\alpha_{h} \operatorname{tr}\left\{\tilde{W}_{V1}^{T}(t) \frac{\hat{\Psi}_{V1}\hat{\Psi}^T_{V1} }{1+\hat{\Psi}_{V1}^{T} \hat{\Psi}_{V1}}\tilde{W}_{V1}(t)\right\}\nonumber\\ 
    
    &-\alpha_{h} \operatorname{tr}\left\{\tilde{W}_{V1}^{T}(t) \frac{\hat{\Psi}_{V1}\tilde{\Psi}^T_{V1} }{1+\hat{\Psi}_{V1}^{T} \hat{\Psi}_{V1}}W_{V1}(t)\right\}  -\alpha_{h} \operatorname{tr}\left\{\tilde{W}_{V1}^{T}(t) \frac{\hat{\Psi}_{V1} \varepsilon_{HJI1}^T}{1+\hat{\Psi}_{V1}^{T} \hat{\Psi}_{V1}}\right\}\nonumber\\
    
    &\leq-\frac{\alpha_h}{4}\frac{\left\|\hat{\Psi}_{V1} \right\|^2}{1+\left\| \hat{\Psi}_{V1}\right\|^2}\left\|\tilde{W}_{V1}(t)\right\|^2  -\frac{\alpha_h}{4}\frac{\left\|\hat{\Psi}_{V1} \right\|^2}{1+\left\| \hat{\Psi}_{V1}\right\|^2}\left\|\tilde{W}_{V1}(t)\right\|^2  +\alpha_{h} \operatorname{tr}\left\{\tilde{W}_{V1}^{T}(t) \frac{\hat{\Psi}_{V1} \tilde{\Phi}^T}{1+\hat{\Psi}_{V1}^{T} \hat{\Psi}_{V1}}\right\}\nonumber\\ 
    
    &-\alpha_h\frac{\left\|\tilde{\Phi}\right\|^2}{1+\left\| \hat{\Psi}_{V1}\right\|^2}  +\alpha_h\frac{\left\|\tilde{\Phi}\right\|^2}{1+\left\| \hat{\Psi}_{V1}\right\|^2}  -\frac{\alpha_h}{4}\frac{\left\|\hat{\Psi}_{V1} \right\|^2}{1+\left\| \hat{\Psi}_{V1}\right\|^2}\left\|\tilde{W}_{V1}(t)\right\|^2  -\alpha_{h} \operatorname{tr}\left\{\tilde{W}_{V1}^{T}(t) \frac{\hat{\Psi}_{V1}\tilde{\Psi}^T_{V1} }{1+\hat{\Psi}_{V1}^{T} \hat{\Psi}_{V1}}W_{V1}(t)\right\}\nonumber\\
    
    &-\alpha_h\frac{\left\|W^T_V(t)\tilde{\Psi}_{V1}\right\|}{1+\left\|\hat{\Psi}_{V1}\right\|}  +\alpha_h\frac{\left\|W^T_V(t)\tilde{\Psi}_{V1}\right\|}{1+\left\|\hat{\Psi}_{V1}\right\|}  -\frac{\alpha_h}{4}\frac{\left\|\hat{\Psi}_{V1} \right\|^2}{1+\left\| \hat{\Psi}_{V1}\right\|^2}\left\|\tilde{W}_{V1}(t)\right\|^2 -\alpha_{h} \operatorname{tr}\left\{\tilde{W}_{V1}^{T}(t) \frac{\hat{\Psi}_{V1} \varepsilon_{HJI1}^T}{1+\hat{\Psi}_{V1}^{T} \hat{\Psi}_{V1}}\right\}\nonumber\\
    
    &-\alpha_h\frac{\left\|\varepsilon_{HJI1}\right\|^2}{1+\left\| \hat{\Psi}_{V1}\right\|^2}  +\alpha_h\frac{\left\|\varepsilon_{HJI1}\right\|^2}{1+\left\| \hat{\Psi}_{V1}\right\|^2}
\end{align}

Combining terms in \eqref{eq22}, 
\begin{align}\label{eq23}
    &\dot{L}_{V1}(t)\leq-\frac{\alpha_h}{4}\frac{\left\|\hat{\Psi}_{V1} \right\|^2}{1+\left\| \hat{\Psi}_{V1}\right\|^2}\left\|\tilde{W}_{V1}(t)\right\|^2  -\frac{\alpha_h}{1+\left\| \hat{\Psi}_{V1}\right\|^2}\left\|\frac{\tilde{W}_{V1}(t)\hat{\Psi}_{V1}}{2}  -\tilde{\Phi}\right\|^2-\frac{\alpha_h}{1+\left\|\hat{\Psi}_{V1}\right\|^2}\left\|\frac{\tilde{W}_{V1}(t)\hat{\Psi}_{V1}}{2}  -W^T_V(t)\tilde{\Psi}_{V1}\right\|^2\nonumber\\
    
    &-\frac{\alpha_h}{1+\left\|\hat{\Psi}_{V1}\right\|^2}\left\|\frac{\tilde{W}_{V1}(t)\hat{\Psi}_{V1}}{2}  -\varepsilon_{HJI1}\right\|^2  +\alpha_h\frac{\left\|\tilde{\Phi}\right\|^2}{1+\left\| \hat{\Psi}_{V1}\right\|^2}  +\alpha_h\frac{\left\|\tilde{\Psi}_{V1}\right\|^2}{1+\left\| \hat{\Psi}_{V1}\right\|^2}  +\underbrace{\alpha_h\frac{\left\|\varepsilon_{HJI1}\right\|^2}{1+\left\| \hat{\Psi}_{V1}\right\|^2}}_{\varepsilon_{VHJI}}
\end{align}

Drop the negative terms in the right side of the inequality yields,
\begin{align}\label{eq24}
    &\dot{L}_{V1}(t)\leq-\frac{\alpha_h}{4}\frac{\left\|\hat{\Psi}_{V1} \right\|^2}{1+\left\| \hat{\Psi}_{V1}\right\|^2}\left\|\tilde{W}_{V1}(t)\right\|^2  +\alpha_h\frac{\left\|\tilde{\Phi}\right\|^2}{1+\left\| \hat{\Psi}_{V1}\right\|^2}  +\alpha_h\frac{\left\|\tilde{\Psi}_{V1}\right\|^2}{1+\left\| \hat{\Psi}_{V1}\right\|^2}  +\varepsilon_{VHJI}
\end{align}

Assume that the coupling function $\phi(x_1,m_1,m_2)$, and the function $\Psi_V(x_1,m_1,m_2)$ are Lipschitz and the Lipschitz constant are $L_\Phi$, $L_{\Psi V}$. \eqref{eq24} can be simplified as
\begin{align}\label{eq24.2}
    \dot{L}_{V1}(t)&\leq -\frac{\alpha_h}{4}\frac{\left\|\hat{\Psi}_{V1} \right\|^2}{1+\left\| \hat{\Psi}_{V1}\right\|^2}\left\|\tilde{W}_{V1}(t)\right\|^2   +\alpha_h\frac{\left[ L_\Phi+L_{\Psi V}\|W_V\|^2\right]\|\tilde{m}_1\tilde{m}_2\|^2}{1+\left\| \hat{\Psi}_{V1}\right\|^2}+\varepsilon_{VHJI}\nonumber\\
    &\leq -\frac{\alpha_h}{4}\frac{\left\|\hat{\Psi}_{V1} \right\|^2}{1+\left\| \hat{\Psi}_{V1}\right\|^2}\left\|\tilde{W}_{V1}(t)\right\|^2+B_V(t)
\end{align}

According to the Lyapunov stability analysis, the critic NN weight estimation error will be Uniformly Ultimately Bounded (UUB) with the bound given as
\begin{align}
    \|\tilde{W}_{V1}\|\leq\sqrt{\frac{4(1+\|\hat{\Psi}_{V1}\|^2)}{\alpha_h\|\hat{\Psi}_{V1}\|^2}B_V(t)}\equiv b_{WV}(t)
\end{align}
\end{proof}

We also derive the bound of estimated optimal cost function as follows:

Let $\tilde{V}_1=V_1-\hat{V}_1$, and substitute \eqref{eq:critic_repre1}, \eqref{eq:critic_estimate1}, one obtains,
\begin{align}
    \tilde{V}_1(t)&=W^T_{V1}(t)\phi_{V1}-\hat{W}_{V1}^T(t)\hat{\phi}_{V1}+\varepsilon_{HJI1}\nonumber\\
    &=W^T_{V1}(t)(\tilde{\phi}_{V1}+\hat{\phi}_{V1})-\hat{W}_{V1}(t)^T\hat{\phi}_{V1}+\varepsilon_{HJI1}\nonumber\\
    &=\tilde{W}^T_{V1}(t)\hat{\phi}_{V1}+W_{V1}^T(t)\tilde{\phi}_{V1}+\varepsilon_{HJI1}
\end{align}

Assume the critic NN activation function is Lipschitz, and the Lipschitz constant is denoted as $L_{\phi v}$. The value function estimation error can be represented as:
\begin{align}\label{eq27}
    \|\tilde{V}_1(t)\|&=\|\tilde{W}^T_{V1}(t)\hat{\phi}_{V1}+W_{V1}^T(t)\tilde{\phi}_{V1}+\varepsilon_{HJI1}\|\nonumber\\
    &\leq \|\tilde{W}_{V1}(t)\|\|\hat{\phi}_{V1}\|+L_{\phi v}\|W_{V1}(t)\|\|\tilde{m_1}\tilde{m_2}\|+\|\varepsilon_{HJI1}\|\nonumber\\
    &\leq b_{WV}(t)\|\hat{\phi}_{V1}\|+L_{\phi v}\|W_{V1}(t)\|\|\tilde{m_1}\tilde{m_2}\|+\|\varepsilon_{HJI1}\|\equiv b_{V1}(t)
\end{align}

\begin{theorem}
\emph{(Convergence of virtual evader's  Critic NN weights and optimal cost function estimations)}\label{th2_p}
Given the initial critic NN weights, $\hat{W}_{V2}$, in a compact set, and let the critic NN weights be updated as Eq. \ref{eq:critic2_p} shows. Then, when the critic NN tuning parameters $\alpha_{h}$ satisfies the condition, $\alpha_{h}>0$, the critic NN weights estimation error $\tilde{W}_{V2}$ and the cost function estimation error $\tilde{V}_{2}=V_2-\hat{V}_2$ will be uniformly ultimately bounded (UUB) where the boundedness can be negligible if the NN reconstruction errors are trivial. While the number of neurons and NN architecture has been designed perfectly, the NN reconstruction error can be as small as possible and trivial. Furthermore, the critic NN weights and cost function estimation errors will be asymptotically stable. 
\end{theorem}
\begin{proof}
Similar to above.
\end{proof}

\section{Convergence of Mass NN}
\begin{theorem}  \label{th3}
\emph{(Convergence of pursuer's Mass NN weights and mass function estimation)}:
Given the initial mass NN weights, $\hat{W}_{m1}(t)$, in a compact set, and let the mass NN weights be updated as Eq. \ref{eq:mass1_p} shows. Then, when the mass NN tuning parameter $\alpha_{m}$ satisfies the condition, $\alpha_{m}>0$, the mass NN weights estimation error $\tilde{W}_{m1}$ and mass function estimation error $\tilde{m}_1=m_1-\hat{m}_1$ will be uniformly ultimately bounded (UUB) where the boundedness can be negligible if the NN reconstruction errors are trivial. While the number of neurons and NN architecture has been designed perfectly, the NN reconstruction error can be as small as possible and trivial. Furthermore, the mass NN weights and mass function estimation errors will be asymptotically stable. 
\end{theorem}

\begin{proof}
Consider the following Lyapunov function 
\begin{align}
    L_{m1}(t)=\frac{1}{2} \operatorname{tr}\left\{\tilde{W}_{m1}^{T}(t) \tilde{W}_{m1}(t)\right\}
\end{align}
Take the first derivative on the Lyapunov function candidate, one obtains:
\begin{align}\label{eq:div_lyapunov_actor}
    \dot{L}_{m1}(t)=\frac{1}{2} \operatorname{tr}\left\{\tilde{W}_{m1}^{T}(t) \dot{\tilde{W}}_{m1}(t)\right\}+\frac{1}{2} \operatorname{tr}\left\{\dot{\tilde{W}}_{m1}^{T}(t) \tilde{W}_{m1}(t)\right\}=\operatorname{tr}\left\{\tilde{W}_{m1}^{T}(t) \dot{\tilde{W}}_{m1}(t)\right\}
\end{align}

Since the correct estimated optimal cost function leads to the FPK equation equals zero, we have
\begin{align}\label{eq:temp4}
    W_{m1}^T(t)\Psi_{m1}\left(x_1,V_1 \right)+\varepsilon_{FPK1}=0
\end{align}

Combine \eqref{eq:temp4} and \eqref{eq:mass_error1_p}, we have
\begin{align} \label{eq:temp5}
    -W_{m1}^T(t)\Psi_{m1}\left(x_1,V_1 \right)-\varepsilon_{FPK1}-\hat{W}_{m1}^T(t)\hat{\Psi}_{m1}\left(x_1,\hat{V}_1 \right)=e_{FPK1}
\end{align}

Let $\tilde{W}_{m1}(t)=W_{m1}(t)-\hat{W}_{m1}(t)$, and $\tilde{\Psi}_{m1}(x_1,V_1,\hat{V}_1)=\Psi_{m1}(x_1,V_1)-\hat{\Psi}_{m1}(x_1,\hat{V}_1)$. After manipulating terms in \eqref{eq:temp5}, we obtain
\begin{align}\label{eq14,2}
    &-W_{m1}^T(t)\left(\hat{\Psi}_{m1}\left(x_1,\hat{V}_1 \right)+ \tilde{\Psi}_{m1}(x_1,V_1,\hat{V}_1)\right)-\varepsilon_{FPK1}+\hat{W}_{m1}^T(t)\hat{\Psi}_{m1}\left(x_1,\hat{V}_1 \right)=e_{FPK1}\nonumber\\
    &-\tilde{W}^T_{m1}\hat{\Psi}_{m1}\left(x_1,\hat{V}_1\right)-W_{m1}^T\tilde{\Psi}_{m1}(x_1,V_1,\hat{V}_1)-\varepsilon_{FPK1}=e_{FPK1}
\end{align}
where $\varepsilon_{FPK1}$ is the error resulted from the reconstruction error.

Let's further simplify the notations as: $\hat{\Psi}_{m1}\left(x_1, \hat{V}_1\right)\rightarrow\hat{\Psi}_{m1}$, $\tilde{\Psi}_{m1}(x_1,V_1,\hat{V}_1)\rightarrow\tilde{\Psi}_{m1}$, $\Psi_{m1}\left(x_1, V_1\right)\rightarrow\Psi_{m1}$

Substitute \eqref{eq14,2} into \eqref{eq:div_lyapunov_actor}, 
\begin{align}\label{eq31}
    &\dot{L}_{m1}(t)=\alpha_{m} \operatorname{tr}\left\{\tilde{W}_{m1}^{T}(t) \frac{\hat{\Psi}_{m1} \left[-\tilde{W}^T_{m1}\hat{\Psi}_{m1}-W_{m1}^T\tilde{\Psi}_{m1}-\varepsilon_{FPK1}\right]^T}{1+\hat{\Psi}_{m1}^{T} \hat{\Psi}_{m1}}\right\}\nonumber\\
    &=-\alpha_{m} \operatorname{tr}\left\{\tilde{W}_{m1}^{T}(t) \frac{\hat{\Psi}_{m1}\hat{\Psi}^T_{m1} }{1+\hat{\Psi}_{m1}^{T} \hat{\Psi}_{m1}}\tilde{W}_{m1}(t)\right\}  -\alpha_{m} \operatorname{tr}\left\{\tilde{W}_{m1}^{T}(t) \frac{\hat{\Psi}_{m1}\tilde{\Psi}^T_{m1} }{1+\hat{\Psi}_{m1}^{T} \hat{\Psi}_{m1}}W_{m1}(t)\right\}  -\alpha_{m} \operatorname{tr}\left\{\tilde{W}_{m1}^{T}(t) \frac{\hat{\Psi}_{m1} \varepsilon_{FPK1}^T}{1+\hat{\Psi}_{m1}^{T} \hat{\Psi}_{m1}}\right\}
\end{align}

Apply Cauchy-Schwarz inequality on \eqref{eq31}, 
\begin{align}\label{eq32}
    &\dot{L}_{m1}(t)= -\alpha_{m} \operatorname{tr}\left\{\tilde{W}_{m1}^{T}(t) \frac{\hat{\Psi}_{m1}\hat{\Psi}^T_{m1} }{1+\hat{\Psi}_{m1}^{T} \hat{\Psi}_{m1}}\tilde{W}_{m1}(t)\right\}  -\alpha_{m} \operatorname{tr}\left\{\tilde{W}_{m1}^{T}(t) \frac{\hat{\Psi}_{m1}\tilde{\Psi}^T_{m1} }{1+\hat{\Psi}_{m1}^{T} \hat{\Psi}_{m1}}W_{m1}(t)\right\}  \nonumber\\
    &-\alpha_{m} \operatorname{tr}\left\{\tilde{W}_{m1}^{T}(t) \frac{\hat{\Psi}_{m1} \varepsilon_{FPK1}^T}{1+\hat{\Psi}_{m1}^{T} \hat{\Psi}_{m1}}\right\}\nonumber\\
    
    &\leq-\frac{\alpha_{m}}{2}\frac{\left\|\hat{\Psi}_{m1} \right\|^2}{1+\left\| \hat{\Psi}_{m1}\right\|^2}\left\|\tilde{W}_{m1}(t)\right\|^2  -\frac{\alpha_{m}}{4}\frac{\left\|\hat{\Psi}_{m1} \right\|^2}{1+\left\| \hat{\Psi}_{m1}\right\|^2}\left\|\tilde{W}_{m1}(t)\right\|^2  -\alpha_{m} \operatorname{tr}\left\{\tilde{W}_{m1}^{T}(t) \frac{\hat{\Psi}_{m1}\tilde{\Psi}^T_{m1} }{1+\hat{\Psi}_{m1}^{T} \hat{\Psi}_{m1}}W_{m1}(t)\right\}\nonumber\\
    
    &-\alpha_{m}\frac{\left\|W^T_{m1}(t)\tilde{\Psi}_{m1}\right\|}{1+\left\|\hat{\Psi}_{m1}\right\|}  +\alpha_{m}\frac{\left\|W^T_{m1}(t)\tilde{\Psi}_{m1}\right\|}{1+\left\|\hat{\Psi}_{m1}\right\|}  -\frac{\alpha_{m}}{4}\frac{\left\|\hat{\Psi}_{m1} \right\|^2}{1+\left\| \hat{\Psi}_{m1}\right\|^2}\left\|\tilde{W}_{m1}(t)\right\|^2 -\alpha_{m} \operatorname{tr}\left\{\tilde{W}_{m1}^{T}(t) \frac{\hat{\Psi}_{m1} \varepsilon_{FPK1}^T}{1+\hat{\Psi}_{m1}^{T} \hat{\Psi}_{m1}}\right\}\nonumber\\
    
    &-\alpha_{m}\frac{\left\|\varepsilon_{FPK1}\right\|^2}{1+\left\| \hat{\Psi}_{m1}\right\|^2}  +\alpha_{m}\frac{\left\|\varepsilon_{FPK1}\right\|^2}{1+\left\| \hat{\Psi}_{m1}\right\|^2}
\end{align}

Combining terms in \eqref{eq32}, 
\begin{align}\label{eq33}
    &\dot{L}_{m1}(t)\leq-\frac{\alpha_{m}}{2}\frac{\left\|\hat{\Psi}_{m1} \right\|^2}{1+\left\| \hat{\Psi}_{m1}\right\|^2}\left\|\tilde{W}_{m1}(t)\right\|^2  -\frac{\alpha_{m}}{1+\left\|\hat{\Psi}_{m1}\right\|^2}\left\|\frac{\tilde{W}_{m1}(t)\hat{\Psi}_{m1}}{2}  -W^T_{m1}(t)\tilde{\Psi}_{m1}\right\|^2\nonumber\\
    
    &-\frac{\alpha_{m}}{1+\left\|\hat{\Psi}_{m1}\right\|^2}\left\|\frac{\tilde{W}_{m1}(t)\hat{\Psi}_{m1}}{2}  -\varepsilon_{FPK1}\right\|^2   +\alpha_{m}\frac{\left\|\tilde{\Psi}_{m1}\right\|^2}{1+\left\| \hat{\Psi}_{m1}\right\|^2}  +\underbrace{\alpha_{m}\frac{\left\|\varepsilon_{FPK1}\right\|^2}{1+\left\| \hat{\Psi}_{m1}\right\|^2}}_{\varepsilon_{NFPK1}}
\end{align}

Drop the negative terms in the right side of the inequality yields,
\begin{align}\label{eq34}
    &\dot{L}_{m1}(t)\leq-\frac{\alpha_{m}}{2}\frac{\left\|\hat{\Psi}_{m1} \right\|^2}{1+\left\| \hat{\Psi}_{m1}\right\|^2}\left\|\tilde{W}_{m1}(t)\right\|^2   +\alpha_{m}\frac{\left\|\tilde{\Psi}_{m1}\right\|^2}{1+\left\| \hat{\Psi}_{m1}\right\|^2}  +\varepsilon_{NFPK1}
\end{align}

Assume that the function $\Psi_{m1}(x_1,V_1)$ are Lipschitz and the Lipschitz constant is $L_{\Psi m}$. \eqref{eq34} can be simplified as
\begin{align}\label{eq34.2}
    \dot{L}_{m1}(t)&\leq -\frac{\alpha_{m}}{2}\frac{\left\|\hat{\Psi}_{m1} \right\|^2}{1+\left\| \hat{\Psi}_{m1}\right\|^2}\left\|\tilde{W}_{m1}(t)\right\|^2   +\alpha_{m}\frac{ L_{\Psi m}\|W_{m1}\|^2\|\tilde{V}_1\|^2}{1+\left\| \hat{\Psi}_{m1}\right\|^2}+\varepsilon_{NFPK1}\nonumber\\
    
    &\leq -\frac{\alpha_{m}}{2}\frac{\left\|\hat{\Psi}_{m1} \right\|^2}{1+\left\| \hat{\Psi}_{m1}\right\|^2}\left\|\tilde{W}_{m1}(t)\right\|^2+B_{m1}(t)
\end{align}

According to the Lyapunov stability analysis, the mass NN weight estimation error will be Uniformly Ultimately Bounded (UUB) with the bound given as
\begin{align}
    \|\tilde{W}_{m1}\|\leq\sqrt{\frac{2(1+\|\hat{\Psi}_{m1}\|^2)}{\alpha_{m}\|\hat{\Psi}_{m1}\|^2}B_{m1}(t)}\equiv b_{Wm}(t)
\end{align}
\end{proof}

We also derive the bound of estimated mass function as follows:

Let $\tilde{m}_1=m_1-\hat{m}_1$, and substitute \eqref{eq:mass_repre1}, \eqref{eq:mass_estimate1}, one obtains,
\begin{align}
    \tilde{m}_1(t)&=W^T_{m1}(t)\phi_{m1}-\hat{W}^T\phi_{m1}+\varepsilon_{FPK1}\nonumber\\
    &=\tilde{W}^T_{m1}(t)\phi_{m1}+\varepsilon_{FPK1}
\end{align}

The PDF estimation error can be represented as:
\begin{align}\label{eq37}
    \|\tilde{m}_1(t)\|&=\|\tilde{W}^T_{m1}(t)\phi_{m1}+\varepsilon_{FPK1}\|\nonumber\\
    &\leq \|\tilde{W}_{m1}(t)\|\|\phi_{m1}\|+\|\varepsilon_{FPK1}\|\nonumber\\
    &\leq b_{Wm}(t)\|\hat{\phi}_{m1}\|+\|\varepsilon_{FPK1}\|\equiv b_{m1}(t)
\end{align}

\begin{theorem}  \label{th4}
\emph{(Convergence of virtual evader's Mass NN weights and mass function estimation)}:
Given the initial mass NN weights, $\hat{W}_{m2}(t)$, in a compact set, and let the mass NN weights be updated as Eq. \ref{eq:mass2_p} shows. Then, when the mass NN tuning parameter $\alpha_{m}$ satisfies the condition, $\alpha_{m}>0$, the mass NN weights estimation error $\tilde{W}_{m2}$ and mass function estimation error $\tilde{m}_2=m_2-\hat{m}_2$ will be uniformly ultimately bounded (UUB) where the boundedness can be negligible if the NN reconstruction errors are trivial. While the number of neurons and NN architecture has been designed perfectly, the NN reconstruction error can be as small as possible and trivial. Furthermore, the mass NN weights and mass function estimation errors will be asymptotically stable. 
\end{theorem}
\begin{proof}
Similar to above.
\end{proof}

\section{Convergence of Actor NN}
\begin{theorem}\label{th5}
\emph{(Convergence of pursuer's Actor NN weights and optimal Mean Field type of control estimation errors)}: Given the initial mass NN weights, $\hat{W}_{u1}$, in a compact set, and let the actor NN weights be updated as Eq. \ref{eq:actor1_p} shows. Then, when the actor NN tuning parameter $\alpha_{u}$ satisfies the condition, $\alpha_{u}>0$, the actor NN weights estimation error $\tilde{W}_{u1}$ and optimal control estimation error $\tilde{u}_1=u_1-\hat{u}_1$ will be uniformly ultimately bounded (UUB) where the boundedness can be negligible if the NN reconstruction errors are trivial. While the number of neurons and NN architecture has been designed perfectly, the NN Reconstruction error can be as small as possible and trivial. Furthermore, the mass NN weights and actor function estimation errors will be asymptotically stable. 
\end{theorem}
\begin{proof}
Consider the following Lyapunov function 
\begin{align}
    L_{u1}(t)=\frac{1}{2} \operatorname{tr}\left\{\tilde{W}_{u1}^{T}(t) \tilde{W}_{u1}(t)\right\}
\end{align}
Take the first derivative on the Lyapunov function candidate, one obtains:
\begin{align}\label{eq:div_lyapunov_actor2}
    \dot{L}_{u1}(t)=\frac{1}{2} \operatorname{tr}\left\{\tilde{W}_{u1}^{T}(t) \dot{\tilde{W}}_{u1}(t)\right\}+\frac{1}{2} \operatorname{tr}\left\{\dot{\tilde{W}}_{u1}^{T}(t) \tilde{W}_{u1}(t)\right\}=\operatorname{tr}\left\{\tilde{W}_{u1}^{T}(t) \dot{\tilde{W}}_{u1}(t)\right\}
\end{align}

Since the correct estimated optimal cost function leads to the optimal control equation equals zero, we have
\begin{align}\label{eq:temp6}
    W_{u1}^T(t)\phi_{u1}\left(x_1,m_1,m_2 \right)+\frac{1}{2}R^{-1}_{g1}g_1(x_1)\frac{\partial \hat{V}_1(x_1,\hat{m}_1,\hat{m}_2)}{\partial x_1}+\varepsilon_{u1}=0
\end{align}

Let $\tilde{W}_{u1}(t)=W_{u1}(t)-\hat{W}_{u1}(t)$, and $\tilde{\phi}_{u1}(x_1,m_1,m_2,\hat{m}_1,\hat{m}_2)=\phi_{u1}(x_1,m_1,m_2)-\hat{\phi}_{u1}(x_1,\hat{m}_1,\hat{m}_2)$. Similar to the critic and actor NNs, after manipulating terms, we obtain
\begin{align}\label{eq14,3}
    &-\tilde{W}^T_{u1}\hat{\phi}_{u1}\left(x_1,\hat{m}_1,\hat{m}_2\right)-W_{u1}^T\tilde{\phi}_{u1}(x_1,m_1,m_2,\hat{m}_1,\hat{m}_2)-\frac{1}{2}R^{-1}_{g1}g_1(x_1)\frac{\partial \hat{V}_1(x_1,\hat{m}_1,\hat{m}_2)}{\partial x_1}-\varepsilon_{u1}=e_{u1}
\end{align}
where $\varepsilon_{u1}$ is the error resulted from the reconstruction error.

Let's further simplify the notations as: $\hat{\phi}_{u1}\left(x_1,\hat{m}_1,\hat{m}_2\right)\rightarrow\hat{\phi}_{u1}$, $\tilde{\phi}_{u1}(x_1,m_1,m_2,\hat{m}_1,\hat{m}_2)\rightarrow\tilde{\phi}_{u1}$, $\phi_{u1}\left(x_1, m_1,m_2\right)\rightarrow\phi_{u1}$

Substitute \eqref{eq14,3} into \eqref{eq:div_lyapunov_actor2}, 
\begin{align}\label{eq41}
    &\dot{L}_{u1}(t)=\alpha_{u} \operatorname{tr}\left\{\tilde{W}_{u1}^{T}(t) \frac{\hat{\phi}_{u1} \left[-\tilde{W}^T_{u1}\hat{\phi}_{u1}-W_{u1}^T\tilde{\phi}_{u1}-\frac{1}{2}R^{-1}_{g1}g_1(x_1)\frac{\partial \hat{V}_1}{\partial x_1}-\varepsilon_{u1}\right]^T}{1+\hat{\phi}_{u1}^{T} \hat{\phi}_{u1}}\right\}\nonumber\\
    
    &=-\alpha_{u} \operatorname{tr}\left\{\tilde{W}_{u1}^{T}(t) \frac{\hat{\phi}_{u1}\hat{\phi}^T_{u1} }{1+\hat{\phi}_{u1}^{T} \hat{\phi}_{u1}}\tilde{W}_{u1}(t)\right\}  -\alpha_{u} \operatorname{tr}\left\{\tilde{W}_{u1}^{T}(t) \frac{\hat{\phi}_{u1}\tilde{\phi}^T_{u1} }{1+\hat{\phi}_{u1}^{T} \hat{\phi}_{u1}}W_{u1}(t)\right\}\nonumber\\
    
    &-\alpha_{u} \operatorname{tr}\left\{\tilde{W}_{u1}^{T}(t) \frac{\hat{\phi}_{u1} \left[\frac{1}{2}R^{-1}_{g1}g_1(x_1)\frac{\partial \hat{V}_1}{\partial x_1}\right]^T}{1+\hat{\phi}_{u1}^{T} \hat{\phi}_{u1}}\right\}  -\alpha_{u} \operatorname{tr}\left\{\tilde{W}_{u1}^{T}(t) \frac{\hat{\phi}_{u1} \varepsilon_{u1}^T}{1+\hat{\phi}_{u1}^{T} \hat{\phi}_{u1}}\right\}
\end{align}

Apply Cauchy-Schwarz inequality on \eqref{eq41}, 
\begin{align}\label{eq42}
    &\dot{L}_{u1}(t)=-\alpha_{u} \operatorname{tr}\left\{\tilde{W}_{u1}^{T}(t) \frac{\hat{\phi}_{u1}\hat{\phi}^T_{u1} }{1+\hat{\phi}_{u1}^{T} \hat{\phi}_{u1}}\tilde{W}_{u1}(t)\right\}  -\alpha_{u} \operatorname{tr}\left\{\tilde{W}_{u1}^{T}(t) \frac{\hat{\phi}_{u1}\tilde{\phi}^T_{u1} }{1+\hat{\phi}_{u1}^{T} \hat{\phi}_{u1}}W_{u1}(t)\right\}\nonumber\\
    
    &-\alpha_{u} \operatorname{tr}\left\{\tilde{W}_{u1}^{T}(t) \frac{\hat{\phi}_{u1} \left[\frac{1}{2}R^{-1}_{g1}g_1(x_1)\frac{\partial \hat{V}_1}{\partial x_1}\right]^T}{1+\hat{\phi}_{u1}^{T} \hat{\phi}_{u1}}\right\}  -\alpha_{u} \operatorname{tr}\left\{\tilde{W}_{u1}^{T}(t) \frac{\hat{\phi}_{u1} \varepsilon_{u1}^T}{1+\hat{\phi}_{u1}^{T} \hat{\phi}_{u1}}\right\}\nonumber\\
    
    &\leq-\frac{\alpha_{u}}{4}\frac{\left\|\hat{\phi}_{u1} \right\|^2}{1+\left\| \hat{\phi}_{u1}\right\|^2}\left\|\tilde{W}_{u1}(t)\right\|^2  -\frac{\alpha_{u}}{4}\frac{\left\|\hat{\phi}_{u1} \right\|^2}{1+\left\| \hat{\phi}_{u1}\right\|^2}\left\|\tilde{W}_{u1}(t)\right\|^2  -\alpha_{u} \operatorname{tr}\left\{\tilde{W}_{u1}^{T}(t) \frac{\hat{\phi}_{u1} \left[\frac{1}{2}R^{-1}_{g1}g_1(x_1)\frac{\partial \hat{V}_1}{\partial x_1}\right]^T}{1+\hat{\phi}_{u1}^{T} \hat{\phi}_{u1}}\right\}\nonumber\\ 
    
    &-\alpha_{u}\frac{\left\|\frac{1}{2}R^{-1}_{g1}g_1(x_1)\frac{\partial \hat{V}_1}{\partial x_1}\right\|^2}{1+\left\| \hat{\phi}_{u1}\right\|^2}  +\alpha_{u}\frac{\left\|\frac{1}{2}R^{-1}_{g1}g_1(x_1)\frac{\partial \hat{V}_1}{\partial x_1}\right\|^2}{1+\left\| \hat{\phi}_{u1}\right\|^2}  -\frac{\alpha_{u}}{4}\frac{\left\|\hat{\phi}_{u1} \right\|^2}{1+\left\| \hat{\phi}_{u1}\right\|^2}\left\|\tilde{W}_{u1}(t)\right\|^2  -\alpha_{u} \operatorname{tr}\left\{\tilde{W}_{u1}^{T}(t) \frac{\hat{\phi}_{u1}\tilde{\phi}^T_{u1} }{1+\hat{\phi}_{u1}^{T} \hat{\phi}_{u1}}W_{u1}(t)\right\}\nonumber\\
    
    &-\alpha_{u}\frac{\left\|W^T_{u1}(t)\tilde{\phi}_{u1}\right\|}{1+\left\|\hat{\phi}_{u1}\right\|}  +\alpha_{u}\frac{\left\|W^T_{u1}(t)\tilde{\phi}_{u1}\right\|}{1+\left\|\hat{\phi}_{u1}\right\|}  -\frac{\alpha_{u}}{4}\frac{\left\|\hat{\phi}_{u1} \right\|^2}{1+\left\| \hat{\phi}_{u1}\right\|^2}\left\|\tilde{W}_{u1}(t)\right\|^2 -\alpha_{u} \operatorname{tr}\left\{\tilde{W}_{u1}^{T}(t) \frac{\hat{\phi}_{u1} \varepsilon_{u1}^T}{1+\hat{\phi}_{u1}^{T} \hat{\phi}_{u1}}\right\}\nonumber\\
    
    &-\alpha_{u}\frac{\left\|\varepsilon_{u1}\right\|^2}{1+\left\| \hat{\phi}_{u1}\right\|^2}  +\alpha_{u}\frac{\left\|\varepsilon_{u1}\right\|^2}{1+\left\| \hat{\phi}_{u1}\right\|^2}
\end{align}

Combining terms in \eqref{eq42}, 
\begin{align}\label{eq43}
    &\dot{L}_{u1}(t)\leq-\frac{\alpha_{u}}{4}\frac{\left\|\hat{\phi}_{u1} \right\|^2}{1+\left\| \hat{\phi}_{u1}\right\|^2}\left\|\tilde{W}_{u1}(t)\right\|^2  -\frac{\alpha_{u}}{1+\left\|\hat{\phi}_{u1}\right\|^2}\left\|\frac{\tilde{W}_{u1}(t)\hat{\phi}_{u1}}{2}  -W^T_{u1}(t)\tilde{\phi}_{u1}\right\|^2\nonumber\\
    
    &-\frac{\alpha_{u}}{1+\left\|\hat{\phi}_{u1}\right\|^2}\left\|\frac{\tilde{W}_{u1}(t)\hat{\phi}_{u1}}{2}  -\frac{1}{2}R^{-1}_{g1}g_1(x_1)\frac{\partial \hat{V}_1}{\partial x_1}\right\|^2\nonumber\\
    
    &-\frac{\alpha_{u}}{1+\left\|\hat{\phi}_{u1}\right\|^2}\left\|\frac{\tilde{W}_{u1}(t)\hat{\phi}_{u1}}{2}  -\varepsilon_{u1}\right\|^2   +\frac{\alpha_{u}}{4}\frac{\left\|R^{-1}_{g1}g_1(x_1)\frac{\partial \hat{V}_1}{\partial x_1}\right\|^2}{1+\left\| \hat{\phi}_{u1}\right\|^2}    +\underbrace{\alpha_{u}\frac{\left\|\varepsilon_{u1}\right\|^2}{1+\left\| \hat{\phi}_{u1}\right\|^2}}_{\varepsilon_{Nu1}}
\end{align}

Drop the negative terms in the right side of the inequality yields,
\begin{align}\label{eq44}
    &\dot{L}_{u1}(t)\leq-\frac{\alpha_{u}}{4}\frac{\left\|\hat{\phi}_{u1} \right\|^2}{1+\left\| \hat{\phi}_{u1}\right\|^2}\left\|\tilde{W}_{u1}(t)\right\|^2   +\frac{\alpha_{u}}{4}\frac{\left\|R^{-1}_{g1}g_1(x_1)\frac{\partial \hat{V}_1}{\partial x_1}\right\|^2}{1+\left\| \hat{\phi}_{u1}\right\|^2}   +\varepsilon_{Nu1}\nonumber\\

    &\leq -\frac{\alpha_{u}}{4}\frac{\left\|\hat{\phi}_{u1} \right\|^2}{1+\left\| \hat{\phi}_{u1}\right\|^2}\left\|\tilde{W}_{u1}(t)\right\|^2   +\alpha_{u}\frac{\|R^{-1}_{g1}g_1(x_1)\|^2\|\tilde{V}_1\|^2}{1+\left\| \hat{\phi}_{u1}\right\|^2}+\varepsilon_{Nu1}\nonumber\\
    
    &\leq -\frac{\alpha_{u}}{4}\frac{\left\|\hat{\phi}_{u1} \right\|^2}{1+\left\| \hat{\phi}_{u1}\right\|^2}\left\|\tilde{W}_{u1}(t)\right\|^2+B_{u1}(t)
\end{align}

According to the Lyapunov stability analysis, the actor NN weight estimation error will be Uniformly Ultimately Bounded (UUB) with the bound given as
\begin{align}
    \|\tilde{W}_{u1}\|\leq\sqrt{\frac{4(1+\|\hat{\phi}_{u1}\|^2)}{\alpha_{u}\|\hat{\phi}_{u1}\|^2}B_{u1}(t)}\equiv b_{Wm}(t)
\end{align}
\end{proof}

We also derive the bound of estimated optimal control function as follows:

Similarly, let $\tilde{u}_1=m_1-\hat{u}_1$, and substitute \eqref{eq:actor_repre1}, \eqref{eq:actor_estimate2}, one obtains,
\begin{align}
    \tilde{u}_1(t)&=\tilde{W}^T_{u1}(t)\phi_{u1}+W^T_{u1}(t)\tilde{\phi}_{u1}+\varepsilon_{u1}
\end{align}

The optimal control estimation error can be represented as:
\begin{align}\label{eq47}
    \|\tilde{u}_1(t)\|&=\|\tilde{W}^T_{u1}(t)\phi_{u1}+W^T_{u1}(t)\tilde{\phi}_{u1}+\varepsilon_{u1}\|\nonumber\\
    &\leq \|\tilde{W}_{u1}(t)\|\|\hat{\phi}_{u1}\|+L_{\phi u}\|W_{u1}\|\|\tilde{m_1}\tilde{m_2}\|+\|\varepsilon_{u1}\|\nonumber\\
    &\leq b_{Wu}(t)\|\hat{\phi}_{u1}\|+L_{\phi u}\|W_{u1}\|\|\tilde{m_1}\tilde{m_2}\|+\|\varepsilon_{u1}\|\equiv b_{u1}(t)
\end{align}
where $L_{\phi u}$ is the Lipschitz constant of the actor NN's activation function. 

\begin{theorem}\label{th6}
\emph{(Convergence of virtual evader's Actor NN weights and optimal Mean Field type of control estimation errors)}: Given the initial mass NN weights, $\hat{W}_{u2}$, in a compact set, and let the actor NN weights be updated as Eq. \ref{eq:actor2} shows. Then, when the actor NN tuning parameter $\alpha_{u}$ satisfies the condition, $\alpha_{u}>0$, the actor NN weights estimation error $\tilde{W}_{u2}$ and optimal control estimation error $\tilde{u}_2=u_2-\hat{u}_2$ will be uniformly ultimately bounded (UUB) where the boundedness can be negligible if the NN reconstruction errors are trivial. While the number of neurons and NN architecture has been designed perfectly, the NN Reconstruction error can be as small as possible and trivial. Furthermore, the mass NN weights and actor function estimation errors will be asymptotically stable. 
\end{theorem}

\begin{proof}
Similar to above.
\end{proof}

\section{Closed-loop Stability}
Before prove the closed-loop stability, a lemma is needed. 

\begin{lemma}\label{lm1}
Consider the system dynamics given in \eqref{eq:system_dynamics1_p}, there exists an optimal mean-field type of optimal control, $u^*_1$, such that the closed-loop system dynamics, $f_1\left(x_1\right)+g_1\left(x_1\right) u^*_1+G_2+\sigma_1 \frac{dw_1}{dt}$
\begin{align}
    x_1^T\left[f_1\left(x_1\right)+g_1\left(x_1\right) u^*+G_2+\sigma_1 \frac{dw_1}{dt}\right]\leq -\gamma_1\|x_1\|^2
\end{align}
where $\gamma_1>0$ is a constant.
\end{lemma}

\begin{lemma}\label{lm2}
Consider the system dynamics given in \eqref{eq:system_dynamics2_p}, there exists an optimal mean-field type of optimal control, $u^*_2$, such that the closed-loop system dynamics, $f_2\left(x_2\right)+g_2\left(x_2\right) u^*_2+G_1+\sigma_2 \frac{dw_2}{dt}$
\begin{align}
    x_2^T\left[f_2\left(x_2\right)+g_2\left(x_2\right) u^*+G_2+\sigma_2 \frac{dw_2}{dt}\right]\leq -\gamma_2\|x_2\|^2
\end{align}
where $\gamma_2>0$ is a constant.
\end{lemma}

\begin{theorem}  \label{theorem_closed_loop}
\emph{(Closed-loop Stability}) 
Given an admissible initial control input and let the actor, critic, and mass NNs weights be selected within a compact set. Moreover, the critic, actor, and mass NNs' weight tuning laws for pursuers in $\mathcal{G}_{1}$ are given as (\ref{eq:critic1_p}), (\ref{eq:critic2_p}), (\ref{eq:actor1_p}), (\ref{eq:mass1_p}), and (\ref{eq:mass2_p}), respectively. Then, there exists constants $\alpha_h$, $\alpha_m$, and $\alpha_u$, such that the system states $x_1$, $x_2$, actor, critic, and  mass NNs weights estimation errors, $\Tilde{W}_{V1}$, $\Tilde{W}_{m1}$, $\Tilde{W}_{u1}$, $\Tilde{W}_{V2}$, $\Tilde{W}_{m2}$, and $\Tilde{W}_{u2}$ are all uniformly ultimately bounded (UUB). In addition, the estimated cost function, mass function and control inputs are all UUB. If the number of neurons and NN architecture has been designed effectively, those NN reconstruction error can be as small as possible and trivial. Furthermore, the system states $x_1$, $x_2$, actor, critic, and  mass NNs weights estimation errors, $\Tilde{W}_{V1}$, $\Tilde{W}_{m1}$, $\Tilde{W}_{u1}$, $\Tilde{W}_{V2}$, $\Tilde{W}_{m2}$, and $\Tilde{W}_{u2}$ will still be asymptotically stable. 
\end{theorem}

\begin{proof}
Consider the Lyapunov function candidate as: 
\begin{align}
    L_{s y s m}(t)&=\frac{\beta_{1}}{2} \operatorname{tr}\left\{x_1^{T}(t) x_{1}(t)\right\}+\frac{\beta_{2}}{2} \operatorname{tr}\left\{\tilde{W}_{V1}^{T}(t) \tilde{W}_{V1}(t)\right\}+\frac{\beta_{3}}{2}\operatorname{tr}\left\{\tilde{W}_{m1}^{T}(t) \tilde{W}_{m1}(t)\right\}+\frac{\beta_{4}}{2}\operatorname{tr}\left\{\tilde{W}_{u1}^{T}(t) \tilde{W}_{u1}(t)\right\}\nonumber\\
    
    &+\frac{\beta_{5}}{2} \operatorname{tr}\left\{x_2^{T}(t) x_{2}(t)\right\}+\frac{\beta_{6}}{2} \operatorname{tr}\left\{\tilde{W}_{V2}^{T}(t) \tilde{W}_{V2}(t)\right\}+\frac{\beta_{7}}{2}\operatorname{tr}\left\{\tilde{W}_{m2}^{T}(t) \tilde{W}_{m2}(t)\right\} +\frac{\beta_{8}}{2}\operatorname{tr}\left\{\tilde{W}_{u2}^{T}(t) \tilde{W}_{u2}(t)\right\}
\end{align}
According to the Lyapunov stability method, taking the first derivative of the selected Lyapunov function candidate
\begin{align}\label{eq:eq50}
&{{\dot L}_{sysm}}(t) = \frac{{{\beta _1}}}{2}{\mathop{\rm tr}\nolimits} \left\{ {x_1^T(t){{\dot x}_1}(t)} \right\} + \frac{{{\beta _1}}}{2}{\mathop{\rm tr}\nolimits} \left\{ {\dot x_1^T(t){x_1}(t)} \right\}+ \frac{{{\beta _2}}}{2}{\mathop{\rm tr}\nolimits} \left\{ {\tilde W_{V1}^T(t){{\dot{ \tilde W}}_{V1}}(t)} \right\} + \frac{{{\beta _2}}}{2}{\mathop{\rm tr}\nolimits} \left\{ {\dot{ \tilde W}_{V1}^T(t){{\tilde W}_{V1}}(t)} \right\}\nonumber\\

&+ \frac{{{\beta _3}}}{2}{\mathop{\rm tr}\nolimits} \left\{ {\tilde W_{m1}^T(t){{\dot{ \tilde W}}_{m1}}(t)} \right\} + \frac{{{\beta _3}}}{2}{\mathop{\rm tr}\nolimits} \left\{ {\dot{ \tilde W}_{m1}^T(t){{\tilde W}_{m1}}(t)} \right\}+ \frac{{{\beta _4}}}{2}{\mathop{\rm tr}\nolimits} \left\{ {\tilde W_{u1}^T(t){{\dot{ \tilde W}}_{u1}}(t)} \right\} + \frac{{{\beta _4}}}{2}{\mathop{\rm tr}\nolimits} \left\{ {\dot{ \tilde W}_{u1}^T(t){{\tilde W}_{u1}}(t)} \right\}\nonumber\\

&+\frac{{{\beta _5}}}{2}{\mathop{\rm tr}\nolimits} \left\{ {x_2^T(t){{\dot x}_2}(t)} \right\} + \frac{{{\beta _5}}}{2}{\mathop{\rm tr}\nolimits} \left\{ {\dot x_2^T(t){x_2}(t)} \right\}+ \frac{{{\beta _6}}}{2}{\mathop{\rm tr}\nolimits} \left\{ {\tilde W_{V2}^T(t){{\dot{ \tilde W}}_{V2}}(t)} \right\} + \frac{{{\beta _6}}}{2}{\mathop{\rm tr}\nolimits} \left\{ {\dot{ \tilde W}_{V2}^T(t){{\tilde W}_{V2}}(t)} \right\}\nonumber\\

&+ \frac{{{\beta _7}}}{2}{\mathop{\rm tr}\nolimits} \left\{ {\tilde W_{m2}^T(t){{\dot{ \tilde W}}_{m2}}(t)} \right\} + \frac{{{\beta _7}}}{2}{\mathop{\rm tr}\nolimits} \left\{ {\dot{ \tilde W}_{m2}^T(t){{\tilde W}_{m2}}(t)} \right\} +\frac{{{\beta _8}}}{2}\operatorname{tr}\left\{ {\tilde W_{u2}^T(t){{\dot{ \tilde W}}_{u2}}(t)} \right\} + \frac{{{\beta _8}}}{2}{\mathop{\rm tr}\nolimits} \left\{ {\dot{ \tilde W}_{u2}^T(t){{\tilde W}_{u2}}(t)} \right\}  \nonumber\\

&={\beta _1}{\mathop{\rm tr}\nolimits} \left\{ {x_1^T(t){{\dot x}_1}(t)} \right\} + {\beta _2}{\mathop{\rm tr}\nolimits} \left\{ {\tilde W_{V1}^T(t){{\dot{ \tilde W}}_{V1}}(t)} \right\}+{\beta _3}{\mathop{\rm tr}\nolimits} \left\{ {\tilde W_{m1}^T(t){{\dot{ \tilde W}}_{m1}}(t)} \right\} + {\beta _4}{\mathop{\rm tr}\nolimits} \left\{ {\tilde W_{u1}^T(t){{\dot{ \tilde W}}_{u1}}(t)} \right\}   \nonumber\\

&+ \beta _5{\mathop{\rm tr}\nolimits} \left\{ {\tilde x_2^T(t) x_2(t)} \right\} + \beta _6{\mathop{\rm tr}\nolimits} \left\{ {\dot{ \tilde W}_{V2}^T(t){{\tilde W}_{V2}}(t)} \right\}+ \beta _7{\mathop{\rm tr}\nolimits} \left\{ {\tilde W_{m2}^T(t){{\dot{ \tilde W}}_{m2}}(t)} \right\} +{\beta _8}{\mathop{\rm tr}\nolimits} \left\{ {\tilde W_{u2}^T(t){{\dot{ \tilde W}}_{u2}}(t)} \right\}
\end{align}

Recall to Lemmas \ref{lm1}, \ref{lm2}, Theorems \ref{th1}-\ref{th6}, and equations \eqref{eq24.2}, \eqref{eq34.2}, \eqref{eq44}, \eqref{eq:eq50} can be represented as:
\begin{align}\label{eq:eq51}
&{{\dot L}_{sysm}}(t) ={\beta _1}{\mathop{\rm tr}\nolimits} \left\{ {x_1^T(t){{\dot x}_1}(t)} \right\} + {\beta _2}{\mathop{\rm tr}\nolimits} \left\{ {\tilde W_{V1}^T(t){{\dot{ \tilde W}}_{V1}}(t)} \right\}+{\beta _3}{\mathop{\rm tr}\nolimits} \left\{ {\tilde W_{m1}^T(t){{\dot{ \tilde W}}_{m1}}(t)} \right\} + {\beta _4}{\mathop{\rm tr}\nolimits} \left\{ {\tilde W_{u1}^T(t){{\dot{ \tilde W}}_{u1}}(t)} \right\}   \nonumber\\

&+ \beta _5{\mathop{\rm tr}\nolimits} \left\{ {\tilde x_2^T(t) x_2(t)} \right\} + \beta _6{\mathop{\rm tr}\nolimits} \left\{ {\dot{ \tilde W}_{V2}^T(t){{\tilde W}_{V2}}(t)} \right\}+ \beta _7{\mathop{\rm tr}\nolimits} \left\{ {\tilde W_{m2}^T(t){{\dot{ \tilde W}}_{m2}}(t)} \right\} +{\beta _8}{\mathop{\rm tr}\nolimits} \left\{ {\tilde W_{u2}^T(t){{\dot{ \tilde W}}_{u2}}(t)} \right\}\nonumber\\ \nonumber\\

&\leq \beta_{1} \operatorname{tr}\left\{x_{1}^{T}\left[f_1\left(x_{1}\right)+g_{1}\left(x_{1}\right) u^{*}_1+\sigma_1 \frac{d w_{1}}{d t}\right]\right\}-\beta_{1} \operatorname{tr}\left\{x_{1}^{T} g_{1}\left(x_{1}\right) \tilde{u}_1\right\}-\frac{2 \beta_{1}}{\gamma_1}\left\|g_{1}\left(x_{1}\right) \tilde{u}_1\right\|^{2}+\frac{2 \beta_{1}}{\gamma_1}\left\|g_{1}\left(x_{1}\right) \tilde{u}_1\right\|^{2}  \nonumber\\

&-\frac{\alpha_{h} \beta_{2}}{4}  \frac{\left\|\hat{\Psi}_{V1}\right\|^{2}}{1+\left\|\hat{\Psi}_{V1}\right\|^{2}}\left\|\tilde{W}_{V1}\right\|^{2}   +\alpha_{h} \frac{\beta_{2}\left[L_\Phi+L_{\Psi V1}\left\|W_{V1}\right\|^{2}\right]\left\|\tilde{m}_{1}\tilde{m}_2\right\|^{2}}{1+\left\|\hat{\Psi}_{V1}\right\|^{2}}+\beta_{2} \varepsilon_{VHJI1}  \nonumber\\

&-\frac{\alpha_{m} \beta_{3}}{2} \frac{\left\|\hat{\Psi}_{m1}\right\|^{2}}{1+\left\|\hat{\Psi}_{m1}\right\|^{2}}\left\|\tilde{W}_{m1}\right\|^{2}+\alpha_{m} \frac{\beta_{3} L_{\Psi_{m1}}\left\|W_{m1}\right\|^{2}\left\|\tilde{V}_{1}\right\|^{2}}{1+\left\|\hat{\Psi}_{m1}\right\|^{2}}+\beta_{3} \varepsilon_{N F P K1} \nonumber\\

&-\frac{\alpha_{u} \beta_{4}}{4} \frac{\left\|\hat{\phi}_{u1}\right\|^{2}}{1+\left\|\hat{\phi}_{u1}\right\|^{2}}\left\|\tilde{W}_{u1}\right\|^{2}+\alpha_u \beta_{4} \frac{\left\|R_1^{-1} g_1^{T}\left(x_1\right)\right\|^{2}\left\|\tilde{V}_{1}\right\|^{2}}{1+\left\|\hat{\phi}_{u1}\right\|^{2}}+\beta_{4} \varepsilon_{Nu1}  \nonumber\\

&+\beta_{5} \operatorname{tr}\left\{x_{2}^{T}\left[f_2\left(x_{2}\right)+g_{2}\left(x_{2}\right) u^{*}_2+\sigma_2 \frac{d w_{2}}{d t}\right]\right\}-\beta_{5} \operatorname{tr}\left\{x_{2}^{T} g_{2}\left(x_{2}\right) \tilde{u}_2\right\}-\frac{2 \beta_{5}}{\gamma_2}\left\|g_{2}\left(x_{2}\right) \tilde{u}_2\right\|^{2}+\frac{2 \beta_{5}}{\gamma_2}\left\|g_{2}\left(x_{2}\right) \tilde{u}_2\right\|^{2}  \nonumber\\

&-\frac{\alpha_{h} \beta_{6}}{4}  \frac{\left\|\hat{\Psi}_{V2}\right\|^{2}}{1+\left\|\hat{\Psi}_{V2}\right\|^{2}}\left\|\tilde{W}_{V2}\right\|^{2}   +\alpha_{h} \frac{\beta_{6}\left[L_\Phi+L_{\Psi V2}\left\|W_{V2}\right\|^{2}\right]\left\|\tilde{m}_{1}\tilde{m}_2\right\|^{2}}{1+\left\|\hat{\Psi}_{V2}\right\|^{2}}+\beta_{6} \varepsilon_{VHJI2}  \nonumber\\

&-\frac{\alpha_{m} \beta_{7}}{2} \frac{\left\|\hat{\Psi}_{m2}\right\|^{2}}{1+\left\|\hat{\Psi}_{m2}\right\|^{2}}\left\|\tilde{W}_{m2}\right\|^{2}+\alpha_{m} \frac{\beta_{7} L_{\Psi_{m2}}\left\|W_{m2}\right\|^{2}\left\|\tilde{V}_{2}\right\|^{2}}{1+\left\|\hat{\Psi}_{m2}\right\|^{2}}+\beta_{7} \varepsilon_{N F P K2} \nonumber\\

&-\frac{\alpha_{u} \beta_{8}}{4} \frac{\left\|\hat{\phi}_{u2}\right\|^{2}}{1+\left\|\hat{\phi}_{u2}\right\|^{2}}\left\|\tilde{W}_{u2}\right\|^{2}+\alpha_u \beta_{8} \frac{\left\|R_2^{-1} g_2^{T}\left(x_2\right)\right\|^{2}\left\|\tilde{V}_2\right\|^{2}}{1+\left\|\hat{\phi}_{u2}\right\|^{2}}+\beta_{8} \varepsilon_{Nu1}  \nonumber\\\nonumber \\

&\leq-\frac{\gamma_1 \beta_{1}}{2}\left\|x_1\right\|^{2}-\frac{\gamma_1 \beta_{1}}{2}\left\|x_1\right\|^{2}-\beta_{1} \operatorname{tr}\left\{x_{1}^{T} g_{1}\left(x_{1}\right) \tilde{u}_1\right\}-\frac{2 \beta_{1}}{\gamma_1}\left\|g_{1}\left(x_{1}\right) \tilde{u}_1\right\|^{2}+\frac{2 \beta_{1}}{\gamma_1}\left\|g_{1}\left(x_{1}\right) \tilde{u}_1\right\|^{2}\nonumber \\

&-\frac{\alpha_{h} \beta_{2}}{4}  \frac{\left\|\hat{\Psi}_{V1}\right\|^{2}}{1+\left\|\hat{\Psi}_{V1}\right\|^{2}}\left\|\tilde{W}_{V1}\right\|^{2}   +\alpha_{h} \frac{\beta_{2}\left[L_\Phi+L_{\Psi V1}\left\|W_{V1}\right\|^{2}\right]\left\|\tilde{m}_{1}\tilde{m}_2\right\|^{2}}{1+\left\|\hat{\Psi}_{V1}\right\|^{2}}  \nonumber\\

&-\frac{\alpha_{m} \beta_{3}}{2} \frac{\left\|\hat{\Psi}_{m1}\right\|^{2}}{1+\left\|\hat{\Psi}_{m1}\right\|^{2}}\left\|\tilde{W}_{m1}\right\|^{2}+\alpha_{m} \frac{\beta_{3} L_{\Psi_{m1}}\left\|W_{m1}\right\|^{2}\left\|\tilde{V}_{1}\right\|^{2}}{1+\left\|\hat{\Psi}_{m1}\right\|^{2}} \nonumber\\

&-\frac{\alpha_{u} \beta_{4}}{4} \frac{\left\|\hat{\phi}_{u1}\right\|^{2}}{1+\left\|\hat{\phi}_{u1}\right\|^{2}}\left\|\tilde{W}_{u1}\right\|^{2}+\alpha_u \beta_{4} \frac{\left\|R_1^{-1} g_1^{T}\left(x_1\right)\right\|^{2}\left\|\tilde{V}_{1}\right\|^{2}}{1+\left\|\hat{\phi}_{u1}\right\|^{2}}+\beta_{4} \varepsilon_{Nu1} +\beta_{3} \varepsilon_{N F P K1}+\beta_{2} \varepsilon_{VHJI1}  \nonumber\\

&-\frac{\gamma_2 \beta_{5}}{2}\left\|x_2\right\|^{2}-\frac{\gamma_2 \beta_{5}}{2}\left\|x_2\right\|^{2}-\beta_{5} \operatorname{tr}\left\{x_{2}^{T} g_{2}\left(x_{2}\right) \tilde{u}_2\right\}-\frac{2 \beta_{5}}{\gamma_2}\left\|g_{2}\left(x_{2}\right) \tilde{u}_2\right\|^{2}+\frac{2 \beta_{5}}{\gamma_2}\left\|g_{2}\left(x_{2}\right) \tilde{u}_2\right\|^{2}  \nonumber\\

&-\frac{\alpha_{h} \beta_{6}}{4}  \frac{\left\|\hat{\Psi}_{V2}\right\|^{2}}{1+\left\|\hat{\Psi}_{V2}\right\|^{2}}\left\|\tilde{W}_{V2}\right\|^{2}   +\alpha_{h} \frac{\beta_{6}\left[L_\Phi+L_{\Psi V2}\left\|W_{V2}\right\|^{2}\right]\left\|\tilde{m}_{1}\tilde{m}_2\right\|^{2}}{1+\left\|\hat{\Psi}_{V2}\right\|^{2}}  \nonumber\\

&-\frac{\alpha_{m} \beta_{7}}{2} \frac{\left\|\hat{\Psi}_{m2}\right\|^{2}}{1+\left\|\hat{\Psi}_{m2}\right\|^{2}}\left\|\tilde{W}_{m2}\right\|^{2}+\alpha_{m} \frac{\beta_{7} L_{\Psi_{m2}}\left\|W_{m2}\right\|^{2}\left\|\tilde{V}_{2}\right\|^{2}}{1+\left\|\hat{\Psi}_{m2}\right\|^{2}}\nonumber \\

&-\frac{\alpha_{u} \beta_{8}}{4} \frac{\left\|\hat{\phi}_{u2}\right\|^{2}}{1+\left\|\hat{\phi}_{u2}\right\|^{2}}\left\|\tilde{W}_{u2}\right\|^{2}+\alpha_u \beta_{8} \frac{\left\|R_2^{-1} g_2^{T}\left(x_2\right)\right\|^{2}\left\|\tilde{V}_2\right\|^{2}}{1+\left\|\hat{\phi}_{u2}\right\|^{2}} +\beta_{8} \varepsilon_{Nu2}+\beta_{7} \varepsilon_{N F P K2}+\beta_{6} \varepsilon_{VHJI2}  \nonumber\\ \nonumber\\

&\leq-\frac{\gamma_1 \beta_{1}}{2}\left\|x_1\right\|^{2}-\beta_1\left[\sqrt{\frac{\gamma_1}{2}}\|x_1\|+\sqrt{\frac{2}{\gamma_1}}\|g_1(x_1)\tilde{u}_1 \right]^2+\frac{2g^2_{M1}\beta_1}{\gamma_1}\|\tilde{u}_1\|^2  \nonumber \\

&-\frac{\alpha_{h} \beta_{2}}{4}  \frac{\left\|\hat{\Psi}_{V1}\right\|^{2}}{1+\left\|\hat{\Psi}_{V1}\right\|^{2}}\left\|\tilde{W}_{V1}\right\|^{2}   +\alpha_{h} \frac{\beta_{2}\left[L_\Phi+L_{\Psi V1}\left\|W_{V1}\right\|^{2}\right]\left\|\tilde{m}_{1}\tilde{m}_2\right\|^{2}}{1+\left\|\hat{\Psi}_{V1}\right\|^{2}}  \nonumber\\

&-\frac{\alpha_{m} \beta_{3}}{2} \frac{\left\|\hat{\Psi}_{m1}\right\|^{2}}{1+\left\|\hat{\Psi}_{m1}\right\|^{2}}\left\|\tilde{W}_{m1}\right\|^{2}  +\alpha_{m} \frac{\beta_{3} L_{\Psi_{m1}}\left\|W_{m1}\right\|^{2}\left\|\tilde{V}_{1}\right\|^{2}}{1+\left\|\hat{\Psi}_{m1}\right\|^{2}} \nonumber\\

&-\frac{\alpha_{u} \beta_{4}}{4} \frac{\left\|\hat{\phi}_{u1}\right\|^{2}}{1+\left\|\hat{\phi}_{u1}\right\|^{2}}\left\|\tilde{W}_{u1}\right\|^{2} +\alpha_u \beta_{4} \frac{\left\|R_1^{-1} g_1^{T}\left(x_1\right)\right\|^{2}\left\|\tilde{V}_{1}\right\|^{2}}{1+\left\|\hat{\phi}_{u1}\right\|^{2}}  +  \beta_{4} \varepsilon_{Nu1} +\beta_{3} \varepsilon_{N F P K1}+\beta_{2} \varepsilon_{VHJI1}  \nonumber\\

&-\frac{\gamma_2 \beta_{5}}{2}\left\|x_2\right\|^{2}-\beta_5\left[\sqrt{\frac{\gamma_2}{2}}\|x_2\|+\sqrt{\frac{2}{\gamma_2}}\|g_2(x_2)\tilde{u}_2 \right]^2+\frac{2g^2_{M2}\beta_5}{\gamma_2}\|\tilde{u}_2\|^2  \nonumber\\

&-\frac{\alpha_{h} \beta_{6}}{4}  \frac{\left\|\hat{\Psi}_{V2}\right\|^{2}}{1+\left\|\hat{\Psi}_{V2}\right\|^{2}}\left\|\tilde{W}_{V2}\right\|^{2}   +\alpha_{h} \frac{\beta_{6}\left[L_\Phi+L_{\Psi V2}\left\|W_{V2}\right\|^{2}\right]\left\|\tilde{m}_{1}\tilde{m}_2\right\|^{2}}{1+\left\|\hat{\Psi}_{V2}\right\|^{2}}  \nonumber\\

&-\frac{\alpha_{m} \beta_{7}}{2} \frac{\left\|\hat{\Psi}_{m2}\right\|^{2}}{1+\left\|\hat{\Psi}_{m2}\right\|^{2}}\left\|\tilde{W}_{m2}\right\|^{2}+\alpha_{m} \frac{\beta_{7} L_{\Psi_{m2}}\left\|W_{m2}\right\|^{2}\left\|\tilde{V}_{2}\right\|^{2}}{1+\left\|\hat{\Psi}_{m2}\right\|^{2}}  \nonumber\\

&-\frac{\alpha_{u} \beta_{8}}{8} \frac{\left\|\hat{\phi}_{u2}\right\|^{2}}{1+\left\|\hat{\phi}_{u2}\right\|^{2}}\left\|\tilde{W}_{u2}\right\|^{2} +\alpha_u \beta_{8} \frac{\left\|R_2^{-1} g_2^{T}\left(x_2\right)\right\|^{2}\left\|\tilde{V}_2\right\|^{2}}{1+\left\|\hat{\phi}_{u2}\right\|^{2}} +\beta_{8} \varepsilon_{Nu2} +\beta_{7}\varepsilon_{N F P K2}+\beta_{6} \varepsilon_{VHJI2}  \nonumber\\ \nonumber\\

&\leq-\frac{\gamma_1}{2}\beta_1\|x_1\|^2+\frac{2g^2_{M1}\beta_1}{\gamma_1}\|\tilde{u}_1\|^2-\frac{\alpha_{h} \beta_{2}}{4}  \frac{\left\|\hat{\Psi}_{V1}\right\|^{2}}{1+\left\|\hat{\Psi}_{V1}\right\|^{2}}\left\|\tilde{W}_{V1}\right\|^{2}   +\alpha_{h} \frac{\beta_{2}\left[L_\Phi+L_{\Psi V1}\left\|W_{V1}\right\|^{2}\right]\left\|\tilde{m}_{1}\tilde{m}_2\right\|^{2}}{1+\left\|\hat{\Psi}_{V1}\right\|^{2}}  \nonumber\\

&-\frac{\alpha_{m} \beta_{3}}{2} \frac{\left\|\hat{\Psi}_{m1}\right\|^{2}}{1+\left\|\hat{\Psi}_{m1}\right\|^{2}}\left\|\tilde{W}_{m1}\right\|^{2}  -\frac{\alpha_{u} \beta_{4}}{4} \frac{\left\|\hat{\phi}_{u1}\right\|^{2}}{1+\left\|\hat{\phi}_{u1}\right\|^{2}}\left\|\tilde{W}_{u1}\right\|^{2}  +\beta_{4}\varepsilon_{Nu1} +\beta_{3} \varepsilon_{N F P K1}+\beta_{2} \varepsilon_{VHJI1}  \nonumber\\
&+\left[\alpha_{m} \frac{\beta_{3} L_{\Psi_{m1}}\left\|W_{m1}\right\|^{2}}{1+\left\|\hat{\Psi}_{m1}\right\|^{2}}  +\alpha_{u1} \beta_{4} \frac{\left\|R_1^{-1} g_1^{T}\left(x_1\right)\right\|^{2}}{1+\left\|\hat{\phi}_{u1}\right\|^{2}}\right]  \left\|\tilde{V}_{1}\right\|^{2}  \nonumber \\

&-\frac{\gamma_2}{2}\beta_5\|x_2\|^2+\frac{2g^2_{M2}\beta_5}{\gamma_2}\|\tilde{u}_2\|^2-\frac{\alpha_{h} \beta_6}{4}  \frac{\left\|\hat{\Psi}_{V2}\right\|^{2}}{1+\left\|\hat{\Psi}_{V2}\right\|^{2}}\left\|\tilde{W}_{V2}\right\|^{2}   +\alpha_{h} \frac{\beta_6\left[L_\Phi+L_{\Psi V2}\left\|W_{V2}\right\|^{2}\right]\left\|\tilde{m}_{1}\tilde{m}_2\right\|^{2}}{1+\left\|\hat{\Psi}_{V2}\right\|^{2}}  \nonumber\\

&-\frac{\alpha_{m} \beta_7}{2} \frac{\left\|\hat{\Psi}_{m2}\right\|^{2}}{1+\left\|\hat{\Psi}_{m2}\right\|^{2}}\left\|\tilde{W}_{m2}\right\|^{2} -\frac{\alpha_{u} \beta_{8}}{4} \frac{\left\|\hat{\phi}_{u2}\right\|^{2}}{1+\left\|\hat{\phi}_{u2}\right\|^{2}}\left\|\tilde{W}_{u2}\right\|^{2} +\beta_{8} \varepsilon_{Nu2} +\beta_7 \varepsilon_{N F P K2}+\beta_6 \varepsilon_{VHJI2}  \nonumber\\

&+\left[\alpha_{m} \frac{\beta_{7} L_{\Psi_{m2}}\left\|W_{m2}\right\|^{2}}{1+\left\|\hat{\Psi}_{m2}\right\|^{2}}  +\alpha_u \beta_{8} \frac{\left\|R_2^{-1} g_2^{T}\left(x_2\right)\right\|^{2}}{1+\left\|\hat{\phi}_{u2}\right\|^{2}}\right]  \left\|\tilde{V}_2\right\|^{2}
\end{align}
where $g^2_{M1}$ is the upper bound of $g^2_1(x_1)$, $g^2_{M2}$ is the upper bound of $g^2_2(x_2)$

Next, substituting \eqref{eq27} into \eqref{eq:eq51}, \eqref{eq:eq51} can be represented as
\begin{align}\label{eq52}
    &\cdot{L}_{sys}(t)\leq -\frac{\gamma_1}{2}\beta_1\|x_1\|^2+\frac{2g^2_{M1}\beta_1}{\gamma_1}\|\tilde{u}_1\|^2-\frac{\alpha_{h} \beta_{2}}{4}  \frac{\left\|\hat{\Psi}_{V1}\right\|^{2}}{1+\left\|\hat{\Psi}_{V1}\right\|^{2}}\left\|\tilde{W}_{V1}\right\|^{2}   +\alpha_{h} \frac{\beta_{2}\left[L_\Phi+L_{\Psi V1}\left\|W_{V1}\right\|^{2}\right]\left\|\tilde{m}_{1}\tilde{m}_2\right\|^{2}}{1+\left\|\hat{\Psi}_{V1}\right\|^{2}}  \nonumber\\

&-\frac{\alpha_{m} \beta_{3}}{2} \frac{\left\|\hat{\Psi}_{m1}\right\|^{2}}{1+\left\|\hat{\Psi}_{m1}\right\|^{2}}\left\|\tilde{W}_{m1}\right\|^{2}  -\frac{\alpha_{u} \beta_{4}}{4} \frac{\left\|\hat{\phi}_{u1}\right\|^{2}}{1+\left\|\hat{\phi}_{u1}\right\|^{2}}\left\|\tilde{W}_{u1}\right\|^{2}  +\beta_{4}\varepsilon_{Nu1} +\beta_{3} \varepsilon_{N F P K1}+\beta_{2} \varepsilon_{VHJI1}  \nonumber\\
&+\left[\alpha_{m} \frac{\beta_{3} L_{\Psi_{m1}}\left\|W_{m1}\right\|^{2}}{1+\left\|\hat{\Psi}_{m1}\right\|^{2}}  +\alpha_u \beta_{4} \frac{\left\|R_1^{-1} g_1^{T}\left(x_1\right)\right\|^{2}}{1+\left\|\hat{\phi}_{u1}\right\|^{2}}\right]
\left[\|\tilde{W}_{V1}(t)\|\|\hat{\phi}_{V1}\|+L_{\phi v1}\|W_{V1}\|\|\tilde{m_1}\tilde{m_2}\|+\|\varepsilon_{HJI1}\| \right]^2  \nonumber\\

&-\frac{\gamma_2}{2}\beta_5\|x_2\|^2+\frac{2g^2_{M2}\beta_5}{\gamma_2}\|\tilde{u}_2\|^2-\frac{\alpha_{h} \beta_6}{4}  \frac{\left\|\hat{\Psi}_{V2}\right\|^{2}}{1+\left\|\hat{\Psi}_{V2}\right\|^{2}}\left\|\tilde{W}_{V2}\right\|^{2}   +\alpha_{h} \frac{\beta_6\left[L_\Phi+L_{\Psi V2}\left\|W_{V2}\right\|^{2}\right]\left\|\tilde{m}_{1}\tilde{m}_2\right\|^{2}}{1+\left\|\hat{\Psi}_{V2}\right\|^{2}}  \nonumber\\

&-\frac{\alpha_{m} \beta_7}{2} \frac{\left\|\hat{\Psi}_{m2}\right\|^{2}}{1+\left\|\hat{\Psi}_{m2}\right\|^{2}}\left\|\tilde{W}_{m2}\right\|^{2} -\frac{\alpha_{u} \beta_{8}}{4} \frac{\left\|\hat{\phi}_{u2}\right\|^{2}}{1+\left\|\hat{\phi}_{u2}\right\|^{2}}\left\|\tilde{W}_{u2}\right\|^{2} +\beta_{8} \varepsilon_{Nu2} +\beta_7 \varepsilon_{N F P K2}+\beta_6 \varepsilon_{VHJI2}    \nonumber\\

&+\left[\alpha_{m} \frac{\beta_{7} L_{\Psi_{m2}}\left\|W_{m2}\right\|^{2}}{1+\left\|\hat{\Psi}_{m2}\right\|^{2}}  +\alpha_u \beta_{8} \frac{\left\|R_2^{-1} g_2^{T}\left(x_2\right)\right\|^{2}}{1+\left\|\hat{\phi}_{u2}\right\|^{2}}\right]  \left[\|\tilde{W}_{V2}(t)\|\|\hat{\phi}_{V2}\|+L_{\phi v2}\|W_{V2}\|\|\tilde{m_1}\tilde{m_2}\| +\|\varepsilon_{HJI2}\| \right]^2  \nonumber\\\nonumber\\

&\leq -\frac{\gamma_1}{2}\beta_1\|x_1\|^2+\frac{2g^2_{M1}\beta_1}{\gamma_1}\|\tilde{u}_1\|^2-\frac{\alpha_{h} \beta_{2}}{4}  \frac{\left\|\hat{\Psi}_{V1}\right\|^{2}}{1+\left\|\hat{\Psi}_{V1}\right\|^{2}}\left\|\tilde{W}_{V1}\right\|^{2}\nonumber\\

&-\frac{\alpha_{m} \beta_{3}}{2} \frac{\left\|\hat{\Psi}_{m1}\right\|^{2}}{1+\left\|\hat{\Psi}_{m1}\right\|^{2}}\left\|\tilde{W}_{m1}\right\|^{2}  -\frac{\alpha_{u} \beta_{4}}{4} \frac{\left\|\hat{\phi}_{u1}\right\|^{2}}{1+\left\|\hat{\phi}_{u1}\right\|^{2}}\left\|\tilde{W}_{u1}\right\|^{2}  +\beta_{4}\varepsilon_{Nu1} +\beta_{3} \varepsilon_{N F P K1}+\beta_{2} \varepsilon_{VHJI1}  \nonumber\\
&+3\left[\alpha_{m} \frac{\beta_{3} L_{\Psi_{m1}}\left\|W_{m1}\right\|^{2}}{1+\left\|\hat{\Psi}_{m1}\right\|^{2}}  +\alpha_u \beta_{4} \frac{\left\|R_1^{-1} g_1^{T}\left(x_1\right)\right\|^{2}}{1+\left\|\hat{\phi}_{u1}\right\|^{2}}\right]  \| \tilde{W}_{V1}(t)\|^2\|\hat{\phi}_{V1}\|^2\nonumber\\

&+\left[ 3\left[\alpha_{m} \frac{\beta_{3} L_{\Psi_{m1}}\left\|W_{m1}\right\|^{2}}{1+\left\|\hat{\Psi}_{m1}\right\|^{2}}  +\alpha_u \beta_{4} \frac{\left\|R_1^{-1} g_1^{T}\left(x_1\right)\right\|^{2}}{1+\left\|\hat{\phi}_{u1}\right\|^{2}}\right]  L^2_{\phi v1}\|W_{V1}\|^2  +\alpha_{h} \frac{\beta_{2}\left[L_\Phi+L_{\Psi V1}\left\|W_{V1}\right\|^{2}\right]}{1+\left\|\hat{\Psi}_{V1}\right\|^{2}} \right]  \left\|\tilde{m}_{1}\tilde{m}_2\right\|^{2}  \nonumber\\

&+3\left[\alpha_{m} \frac{\beta_{3} L_{\Psi_{m1}}\left\|W_{m1}\right\|^{2}}{1+\left\|\hat{\Psi}_{m1}\right\|^{2}}  +\alpha_u \beta_{4} \frac{\left\|R_1^{-1} g_1^{T}\left(x_1\right)\right\|^{2}}{1+\left\|\hat{\phi}_{u1}\right\|^{2}}\right]  \|\varepsilon_{HJI1}\|^2  \nonumber\\

&-\frac{\gamma_2}{2}\beta_5\|x_2\|^2+\frac{2g^2_{M2}\beta_5}{\gamma_2}\|\tilde{u}_2\|^2-\frac{\alpha_{h} \beta_6}{4}  \frac{\left\|\hat{\Psi}_{V2}\right\|^{2}}{1+\left\|\hat{\Psi}_{V2}\right\|^{2}}\left\|\tilde{W}_{V2}\right\|^{2}     \nonumber\\

&-\frac{\alpha_{m} \beta_7}{2} \frac{\left\|\hat{\Psi}_{m2}\right\|^{2}}{1+\left\|\hat{\Psi}_{m2}\right\|^{2}}\left\|\tilde{W}_{m2}\right\|^{2} -\frac{\alpha_{u} \beta_{8}}{4} \frac{\left\|\hat{\phi}_{u2}\right\|^{2}}{1+\left\|\hat{\phi}_{u2}\right\|^{2}}\left\|\tilde{W}_{u2}\right\|^{2} +\beta_{8} \varepsilon_{Nu2}+\beta_7 \varepsilon_{N F P K2}+\beta_6 \varepsilon_{VHJI2}    \nonumber\\

&+3\left[\alpha_{m} \frac{\beta_{7} L_{\Psi_{m2}}\left\|W_{m2}\right\|^{2}}{1+\left\|\hat{\Psi}_{m2}\right\|^{2}}  +\alpha_u \beta_{8} \frac{\left\|R_2^{-1} g_2^{T}\left(x_2\right)\right\|^{2}}{1+\left\|\hat{\phi}_{u2}\right\|^{2}}\right]  \|\tilde{W}_{V2}(t)\|^2\|\hat{\phi}_{V2}\|^2\nonumber\\

&+\left[3\left[\alpha_{m} \frac{\beta_{7} L_{\Psi_{m2}}\left\|W_{m2}\right\|^{2}}{1+\left\|\hat{\Psi}_{m2}\right\|^{2}}  +\alpha_u \beta_{8} \frac{\left\|R_2^{-1} g_2^{T}\left(x_2\right)\right\|^{2}}{1+\left\|\hat{\phi}_{u2}\right\|^{2}}\right] L^2_{\phi v2}\|W_{V2}\|^2  +\alpha_{h} \frac{\beta_6\left[L_\Phi+L_{\Psi V2}\left\|W_{V2}\right\|^{2}\right]}{1+\left\|\hat{\Psi}_{V2}\right\|^{2}} \right] \|\tilde{m_1}\tilde{m_2}\|^2  \nonumber\\

&+3\left[\alpha_{m} \frac{\beta_{7} L_{\Psi_{m2}}\left\|W_{m2}\right\|^{2}}{1+\left\|\hat{\Psi}_{m2}\right\|^{2}}  +\alpha_u \beta_{8} \frac{\left\|R_2^{-1} g_2^{T}\left(x_2\right)\right\|^{2}}{1+\left\|\hat{\phi}_{u2}\right\|^{2}}\right]  \|\varepsilon_{HJI2}\|^2
\end{align}

Furthermore, substituting \eqref{eq37} into \eqref{eq52}, \eqref{eq52} can be represented as 
\begin{align}\label{eq53}
    &\cdot{L}_{sys}(t)\leq -\frac{\gamma_1}{2}\beta_1\|x_1\|^2+\frac{2g^2_{M1}\beta_1}{\gamma_1}\|\tilde{u}_1\|^2-\frac{\alpha_{h} \beta_{2}}{4}  \frac{\left\|\hat{\Psi}_{V1}\right\|^{2}}{1+\left\|\hat{\Psi}_{V1}\right\|^{2}}\left\|\tilde{W}_{V1}\right\|^{2}\nonumber\\
&+\left[\begin{aligned}
 &3\left[\alpha_{m} \frac{\beta_{3} L_{\Psi_{m1}}\left\|W_{m1}\right\|^{2}}{1+\left\|\hat{\Psi}_{m1}\right\|^{2}}  +\alpha_u \beta_{4} \frac{\left\|R_1^{-1} g_1^{T}\left(x_1\right)\right\|^{2}}{1+\left\|\hat{\phi}_{u1}\right\|^{2}}\right]  L^2_{\phi v1}\|W_{V1}\|^2 \nonumber\\
&+\alpha_{h} \frac{\beta_{2}\left[L_\Phi+L_{\Psi V1}\left\|W_{V1}\right\|^{2}\right]}{1+\left\|\hat{\Psi}_{V1}\right\|^{2}} \nonumber
\end{aligned}\right] 
\|\tilde{m}_2\|^2\left[\|\tilde{W}_{m1}(t)\|\|\phi_{m1}\|+\|\varepsilon_{FPK1}\|\right]^2  \nonumber\\

&-\frac{\alpha_{m} \beta_{3}}{2} \frac{\left\|\hat{\Psi}_{m1}\right\|^{2}}{1+\left\|\hat{\Psi}_{m1}\right\|^{2}}\left\|\tilde{W}_{m1}\right\|^{2}  -\frac{\alpha_{u} \beta_{4}}{4} \frac{\left\|\hat{\phi}_{u1}\right\|^{2}}{1+\left\|\hat{\phi}_{u1}\right\|^{2}}\left\|\tilde{W}_{u1}\right\|^{2}  +\beta_{4}\varepsilon_{Nu1} +\beta_{3} \varepsilon_{N F P K1}+\beta_{2} \varepsilon_{VHJI1}  \nonumber\\

&+3\left[\alpha_{m} \frac{\beta_{3} L_{\Psi_{m1}}\left\|W_{m1}\right\|^{2}}{1+\left\|\hat{\Psi}_{m1}\right\|^{2}}  +\alpha_u \beta_{4} \frac{\left\|R_1^{-1} g_1^{T}\left(x_1\right)\right\|^{2}}{1+\left\|\hat{\phi}_{u1}\right\|^{2}}\right]  \| \tilde{W}_{V1}(t)\|^2\|\hat{\phi}_{V1}\|^2\nonumber\\

&+3\left[\alpha_{m} \frac{\beta_{3} L_{\Psi_{m1}}\left\|W_{m1}\right\|^{2}}{1+\left\|\hat{\Psi}_{m1}\right\|^{2}}  +\alpha_u \beta_{4} \frac{\left\|R_1^{-1} g_1^{T}\left(x_1\right)\right\|^{2}}{1+\left\|\hat{\phi}_{u1}\right\|^{2}}\right]  \|\varepsilon_{HJI1}\|^2  \nonumber\\

&-\frac{\gamma_2}{2}\beta_5\|x_2\|^2+\frac{2g^2_{M2}\beta_5}{\gamma_2}\|\tilde{u}_2\|^2-\frac{\alpha_{h} \beta_6}{4}  \frac{\left\|\hat{\Psi}_{V2}\right\|^{2}}{1+\left\|\hat{\Psi}_{V2}\right\|^{2}}\left\|\tilde{W}_{V2}\right\|^{2}     \nonumber\\
&+\left[\begin{aligned}
&3\left[\alpha_{m} \frac{\beta_{7} L_{\Psi_{m2}}\left\|W_{m2}\right\|^{2}}{1+\left\|\hat{\Psi}_{m2}\right\|^{2}}  +\alpha_u \beta_{8} \frac{\left\|R_2^{-1} g_2^{T}\left(x_2\right)\right\|^{2}}{1+\left\|\hat{\phi}_{u2}\right\|^{2}}\right] L^2_{\phi v2}\|W_{V2}\|^2\nonumber\\

&+\alpha_{h} \frac{\beta_6\left[L_\Phi+L_{\Psi V2}\left\|W_{V2}\right\|^{2}\right]}{1+\left\|\hat{\Psi}_{V2}\right\|^{2}}
\end{aligned}\right] \|\tilde{m}_1\|^2\left[\|\tilde{W}_{m2}(t)\|\|\phi_{m2}\|+\|\varepsilon_{FPK2}\|\right]^2  \nonumber\\

&-\frac{\alpha_{m} \beta_7}{2} \frac{\left\|\hat{\Psi}_{m2}\right\|^{2}}{1+\left\|\hat{\Psi}_{m2}\right\|^{2}}\left\|\tilde{W}_{m2}\right\|^{2} -\frac{\alpha_{u} \beta_{8}}{4} \frac{\left\|\hat{\phi}_{u2}\right\|^{2}}{1+\left\|\hat{\phi}_{u2}\right\|^{2}}\left\|\tilde{W}_{u2}\right\|^{2} +\beta_{8} \varepsilon_{Nu2} +\beta_7 \varepsilon_{N F P K2}+\beta_6 \varepsilon_{VHJI2}    \nonumber\\

&+3\left[\alpha_{m} \frac{\beta_{7} L_{\Psi_{m2}}\left\|W_{m2}\right\|^{2}}{1+\left\|\hat{\Psi}_{m2}\right\|^{2}}  +\alpha_u \beta_{8} \frac{\left\|R_2^{-1} g_2^{T}\left(x_2\right)\right\|^{2}}{1+\left\|\hat{\phi}_{u2}\right\|^{2}}\right]  \|\tilde{W}_{V2}(t)\|^2\|\hat{\phi}_{V2}\|^2\nonumber\\

&+3\left[\alpha_{m} \frac{\beta_{7} L_{\Psi_{m2}}\left\|W_{m2}\right\|^{2}}{1+\left\|\hat{\Psi}_{m2}\right\|^{2}}  +\alpha_u \beta_{8} \frac{\left\|R_2^{-1} g_2^{T}\left(x_2\right)\right\|^{2}}{1+\left\|\hat{\phi}_{u2}\right\|^{2}}\right]  \|\varepsilon_{HJI2}\|^2

\nonumber\\\nonumber\\

&\leq -\frac{\gamma_1}{2}\beta_1\|x_1\|^2+\frac{2g^2_{M1}\beta_1}{\gamma_1}\|\tilde{u}_1\|^2-\frac{\alpha_{h} \beta_{2}}{4}  \frac{\left\|\hat{\Psi}_{V1}\right\|^{2}}{1+\left\|\hat{\Psi}_{V1}\right\|^{2}}\left\|\tilde{W}_{V1}\right\|^{2}\nonumber\\
&+2\left[\begin{aligned}
 &3\left[\alpha_{m} \frac{\beta_{3} L_{\Psi_{m1}}\left\|W_{m1}\right\|^{2}}{1+\left\|\hat{\Psi}_{m1}\right\|^{2}}  +\alpha_u \beta_{4} \frac{\left\|R_1^{-1} g_1^{T}\left(x_1\right)\right\|^{2}}{1+\left\|\hat{\phi}_{u1}\right\|^{2}}\right]  L^2_{\phi v1}\|W_{V1}\|^2 \nonumber\\
&+\alpha_{h} \frac{\beta_{2}\left[L_\Phi+L_{\Psi V1}\left\|W_{V1}\right\|^{2}\right]}{1+\left\|\hat{\Psi}_{V1}\right\|^{2}}   \nonumber
\end{aligned}\right] \|\tilde{m}_2\|^2\|\tilde{W}_{m1}(t)\|^2\|\phi_{m1}\|^2  \nonumber\\

&+2\left[\begin{aligned}
 &3\left[\alpha_{m} \frac{\beta_{3} L_{\Psi_{m1}}\left\|W_{m1}\right\|^{2}}{1+\left\|\hat{\Psi}_{m1}\right\|^{2}}  +\alpha_u \beta_{4} \frac{\left\|R_1^{-1} g_1^{T}\left(x_1\right)\right\|^{2}}{1+\left\|\hat{\phi}_{u1}\right\|^{2}}\right]  L^2_{\phi v1}\|W_{V1}\|^2 \nonumber\\
&+\alpha_{h} \frac{\beta_{2}\left[L_\Phi+L_{\Psi V1}\left\|W_{V1}\right\|^{2}\right]}{1+\left\|\hat{\Psi}_{V1}\right\|^{2}}   \nonumber
\end{aligned}\right] \|\tilde{m}_2\|^2\|\varepsilon_{FPK}\|^2\nonumber \\

&-\frac{\alpha_{m} \beta_{3}}{2} \frac{\left\|\hat{\Psi}_{m1}\right\|^{2}}{1+\left\|\hat{\Psi}_{m1}\right\|^{2}}\left\|\tilde{W}_{m1}\right\|^{2}  -\frac{\alpha_{u} \beta_{4}}{4} \frac{\left\|\hat{\phi}_{u1}\right\|^{2}}{1+\left\|\hat{\phi}_{u1}\right\|^{2}}\left\|\tilde{W}_{u1}\right\|^{2}  +\beta_{4}\varepsilon_{Nu1} +\beta_{3} \varepsilon_{N F P K1}+\beta_{2} \varepsilon_{VHJI1}  \nonumber\\

&+3\left[\alpha_{m} \frac{\beta_{3} L_{\Psi_{m1}}\left\|W_{m1}\right\|^{2}}{1+\left\|\hat{\Psi}_{m1}\right\|^{2}}  +\alpha_u \beta_{4} \frac{\left\|R_1^{-1} g_1^{T}\left(x_1\right)\right\|^{2}}{1+\left\|\hat{\phi}_{u1}\right\|^{2}}\right]  \| \tilde{W}_{V1}(t)\|^2\|\hat{\phi}_{V1}\|^2\nonumber\\

&+3\left[\alpha_{m} \frac{\beta_{3} L_{\Psi_{m1}}\left\|W_{m1}\right\|^{2}}{1+\left\|\hat{\Psi}_{m1}\right\|^{2}}  +\alpha_u \beta_{4} \frac{\left\|R_1^{-1} g_1^{T}\left(x_1\right)\right\|^{2}}{1+\left\|\hat{\phi}_{u1}\right\|^{2}}\right]  \|\varepsilon_{HJI1}\|^2  \nonumber\\

&-\frac{\gamma_2}{2}\beta_5\|x_2\|^2+\frac{2g^2_{M2}\beta_5}{\gamma_2}\|\tilde{u}_2\|^2-\frac{\alpha_{h} \beta_6}{4}  \frac{\left\|\hat{\Psi}_{V2}\right\|^{2}}{1+\left\|\hat{\Psi}_{V2}\right\|^{2}}\left\|\tilde{W}_{V2}\right\|^{2}     \nonumber\\
&+2\left[\begin{aligned}
&3\left[\alpha_{m} \frac{\beta_{7} L_{\Psi_{m2}}\left\|W_{m2}\right\|^{2}}{1+\left\|\hat{\Psi}_{m2}\right\|^{2}}  +\alpha_u \beta_{8} \frac{\left\|R_2^{-1} g_2^{T}\left(x_2\right)\right\|^{2}}{1+\left\|\hat{\phi}_{u2}\right\|^{2}}\right] L^2_{\phi v2}\|W_{V2}\|^2\nonumber\\

&+\alpha_{h} \frac{\beta_6\left[L_\Phi+L_{\Psi V2}\left\|W_{V2}\right\|^{2}\right]}{1+\left\|\hat{\Psi}_{V2}\right\|^{2}}
\end{aligned}\right] \|\tilde{m}_1\|^2\|\tilde{W}_{m2}(t)\|^2 \|\phi_{m2}\|^2  \nonumber\\

&+2\left[\begin{aligned}
&3\left[\alpha_{m} \frac{\beta_{7} L_{\Psi_{m2}}\left\|W_{m2}\right\|^{2}}{1+\left\|\hat{\Psi}_{m2}\right\|^{2}}  +\alpha_u \beta_{8} \frac{\left\|R_2^{-1} g_2^{T}\left(x_2\right)\right\|^{2}}{1+\left\|\hat{\phi}_{u2}\right\|^{2}}\right] L^2_{\phi v2}\|W_{V2}\|^2\nonumber\\   

&+\alpha_{h} \frac{\beta_6\left[L_\Phi+L_{\Psi V2}\left\|W_{V2}\right\|^{2}\right]}{1+\left\|\hat{\Psi}_{V2}\right\|^{2}}
\end{aligned}\right] \|\tilde{m}_1\|^2\|\varepsilon_{FPK}\|^2  \nonumber\\

&-\frac{\alpha_{m} \beta_7}{2} \frac{\left\|\hat{\Psi}_{m2}\right\|^{2}}{1+\left\|\hat{\Psi}_{m2}\right\|^{2}}\left\|\tilde{W}_{m2}\right\|^{2} -\frac{\alpha_{u} \beta_{8}}{4} \frac{\left\|\hat{\phi}_{u2}\right\|^{2}}{1+\left\|\hat{\phi}_{u2}\right\|^{2}}\left\|\tilde{W}_{u2}\right\|^{2} +\beta_{8} \varepsilon_{Nu2} +\beta_7 \varepsilon_{N F P K2}+\beta_6 \varepsilon_{VHJI2}    \nonumber\\

&+3\left[\alpha_{m} \frac{\beta_{7} L_{\Psi_{m2}}\left\|W_{m2}\right\|^{2}}{1+\left\|\hat{\Psi}_{m2}\right\|^{2}}  +\alpha_u \beta_{8} \frac{\left\|R_2^{-1} g_2^{T}\left(x_2\right)\right\|^{2}}{1+\left\|\hat{\phi}_{u2}\right\|^{2}}\right]  \|\tilde{W}_{V2}(t)\|^2\|\hat{\phi}_{V2}\|^2\nonumber\\

&+3\left[\alpha_{m} \frac{\beta_{7} L_{\Psi_{m2}}\left\|W_{m2}\right\|^{2}}{1+\left\|\hat{\Psi}_{m2}\right\|^{2}}  +\alpha_u \beta_{8} \frac{\left\|R_2^{-1} g_2^{T}\left(x_2\right)\right\|^{2}}{1+\left\|\hat{\phi}_{u2}\right\|^{2}}\right]  \|\varepsilon_{HJI2}\|^2
\end{align}

Next, substitute \eqref{eq47} into \eqref{eq53}, \eqref{eq53} can be represented as:
\begin{align}\label{eq54}
    &\cdot{L}_{sys}(t)\leq -\frac{\gamma_1}{2}\beta_1\|x_1\|^2+\frac{2g^2_{M1}\beta_1}{\gamma_1}\left[\|\tilde{W}_u(t)\|\|\hat{\phi}_{u1}\|+L_{\phi u}\|W_u\|\|\tilde{m_1}\tilde{m_2}\|+\|\varepsilon_{u1}\|\right]^2-\frac{\alpha_{h} \beta_{2}}{4}  \frac{\left\|\hat{\Psi}_{V1}\right\|^{2}}{1+\left\|\hat{\Psi}_{V1}\right\|^{2}}\left\|\tilde{W}_{V1}\right\|^{2}\nonumber\\

&+2\left[\begin{aligned}
 &3\left[\alpha_{m} \frac{\beta_{3} L_{\Psi_{m1}}\left\|W_{m1}\right\|^{2}}{1+\left\|\hat{\Psi}_{m1}\right\|^{2}}  +\alpha_u \beta_{4} \frac{\left\|R_1^{-1} g_1^{T}\left(x_1\right)\right\|^{2}}{1+\left\|\hat{\phi}_{u1}\right\|^{2}}\right]  L^2_{\phi v1}\|W_{V1}\|^2 \nonumber\\
&+\alpha_{h} \frac{\beta_{2}\left[L_\Phi+L_{\Psi V1}\left\|W_{V1}\right\|^{2}\right]}{1+\left\|\hat{\Psi}_{V1}\right\|^{2}}   \nonumber
\end{aligned}\right] \|\tilde{m}_2\|^2\|\tilde{W}_{m1}(t)\|^2\|\phi_{m1}\|^2  \nonumber\\

&+2\left[\begin{aligned}
 &3\left[\alpha_{m} \frac{\beta_{3} L_{\Psi_{m1}}\left\|W_{m1}\right\|^{2}}{1+\left\|\hat{\Psi}_{m1}\right\|^{2}}  +\alpha_u \beta_{4} \frac{\left\|R_1^{-1} g_1^{T}\left(x_1\right)\right\|^{2}}{1+\left\|\hat{\phi}_{u1}\right\|^{2}}\right]  L^2_{\phi v1}\|W_{V1}\|^2 \nonumber\\
&+\alpha_{h} \frac{\beta_{2}\left[L_\Phi+L_{\Psi V1}\left\|W_{V1}\right\|^{2}\right]}{1+\left\|\hat{\Psi}_{V1}\right\|^{2}}   \nonumber
\end{aligned}\right] \|\tilde{m}_2\|^2\|\varepsilon_{FPK}\|^2\nonumber \\

&-\frac{\alpha_{m} \beta_{3}}{2} \frac{\left\|\hat{\Psi}_{m1}\right\|^{2}}{1+\left\|\hat{\Psi}_{m1}\right\|^{2}}\left\|\tilde{W}_{m1}\right\|^{2}  -\frac{\alpha_{u} \beta_{4}}{4} \frac{\left\|\hat{\phi}_{u1}\right\|^{2}}{1+\left\|\hat{\phi}_{u1}\right\|^{2}}\left\|\tilde{W}_{u1}\right\|^{2}  +\beta_{4}\varepsilon_{Nu1} +\beta_{3} \varepsilon_{N F P K1}+\beta_{2} \varepsilon_{VHJI1}  \nonumber\\

&+3\left[\alpha_{m} \frac{\beta_{3} L_{\Psi_{m1}}\left\|W_{m1}\right\|^{2}}{1+\left\|\hat{\Psi}_{m1}\right\|^{2}}  +\alpha_u \beta_{4} \frac{\left\|R_1^{-1} g_1^{T}\left(x_1\right)\right\|^{2}}{1+\left\|\hat{\phi}_{u1}\right\|^{2}}\right]  \| \tilde{W}_{V1}(t)\|^2\|\hat{\phi}_{V1}\|^2\nonumber\\

&+3\left[\alpha_{m} \frac{\beta_{3} L_{\Psi_{m1}}\left\|W_{m1}\right\|^{2}}{1+\left\|\hat{\Psi}_{m1}\right\|^{2}}  +\alpha_u \beta_{4} \frac{\left\|R_1^{-1} g_1^{T}\left(x_1\right)\right\|^{2}}{1+\left\|\hat{\phi}_{u1}\right\|^{2}}\right]  \|\varepsilon_{HJI1}\|^2  \nonumber\\

&-\frac{\gamma_2}{2}\beta_5\|x_2\|^2+\frac{2g^2_{M2}\beta_5}{\gamma_2} \left[\|\tilde{W}_{u2}(t)\|\|\hat{\phi}_{u2}\|+L_{\phi u2}\|W_{u2}\|\|\tilde{m_1}\tilde{m_2}\|+\|\varepsilon_{u2}\|\right]^2 -\frac{\alpha_{h} \beta_6}{4}  \frac{\left\|\hat{\Psi}_{V2}\right\|^{2}}{1+\left\|\hat{\Psi}_{V2}\right\|^{2}}\left\|\tilde{W}_{V2}\right\|^{2}     \nonumber\\

&+2\left[\begin{aligned}
&3\left[\alpha_{m} \frac{\beta_{7} L_{\Psi_{m2}}\left\|W_{m2}\right\|^{2}}{1+\left\|\hat{\Psi}_{m2}\right\|^{2}}  +\alpha_u \beta_{8} \frac{\left\|R_2^{-1} g_2^{T}\left(x_2\right)\right\|^{2}}{1+\left\|\hat{\phi}_{u2}\right\|^{2}}\right] L^2_{\phi v2}\|W_{V2}\|^2\nonumber\\

&+\alpha_{h} \frac{\beta_6\left[L_\Phi+L_{\Psi V2}\left\|W_{V2}\right\|^{2}\right]}{1+\left\|\hat{\Psi}_{V2}\right\|^{2}}
\end{aligned}\right] \|\tilde{m}_1\|^2\|\tilde{W}_{m2}(t)\|^2 \|\phi_{m2}\|^2  \nonumber\\

&+2\left[\begin{aligned}
&3\left[\alpha_{m} \frac{\beta_{7} L_{\Psi_{m2}}\left\|W_{m2}\right\|^{2}}{1+\left\|\hat{\Psi}_{m2}\right\|^{2}}  +\alpha_u \beta_{8} \frac{\left\|R_2^{-1} g_2^{T}\left(x_2\right)\right\|^{2}}{1+\left\|\hat{\phi}_{u2}\right\|^{2}}\right] L^2_{\phi v2}\|W_{V2}\|^2\nonumber\\   

&+\alpha_{h} \frac{\beta_6\left[L_\Phi+L_{\Psi V2}\left\|W_{V2}\right\|^{2}\right]}{1+\left\|\hat{\Psi}_{V2}\right\|^{2}}
\end{aligned}\right] \|\tilde{m}_1\|^2\|\varepsilon_{FPK}\|^2  \nonumber\\

&-\frac{\alpha_{m} \beta_7}{2} \frac{\left\|\hat{\Psi}_{m2}\right\|^{2}}{1+\left\|\hat{\Psi}_{m2}\right\|^{2}}\left\|\tilde{W}_{m2}\right\|^{2} -\frac{\alpha_{u} \beta_{8}}{4} \frac{\left\|\hat{\phi}_{u2}\right\|^{2}}{1+\left\|\hat{\phi}_{u2}\right\|^{2}}\left\|\tilde{W}_{u2}\right\|^{2} +\beta_{8} \varepsilon_{Nu2} +\beta_7 \varepsilon_{N F P K2}+\beta_6 \varepsilon_{VHJI2}    \nonumber\\

&+3\left[\alpha_{m} \frac{\beta_{7} L_{\Psi_{m2}}\left\|W_{m2}\right\|^{2}}{1+\left\|\hat{\Psi}_{m2}\right\|^{2}}  +\alpha_u \beta_{8} \frac{\left\|R_2^{-1} g_2^{T}\left(x_2\right)\right\|^{2}}{1+\left\|\hat{\phi}_{u2}\right\|^{2}}\right]  \|\tilde{W}_{V2}(t)\|^2\|\hat{\phi}_{V2}\|^2\nonumber\\

&+3\left[\alpha_{m} \frac{\beta_{7} L_{\Psi_{m2}}\left\|W_{m2}\right\|^{2}}{1+\left\|\hat{\Psi}_{m2}\right\|^{2}}  +\alpha_u \beta_{8} \frac{\left\|R_2^{-1} g_2^{T}\left(x_2\right)\right\|^{2}}{1+\left\|\hat{\phi}_{u2}\right\|^{2}}\right]  \|\varepsilon_{HJI2}\|^2 \nonumber\\\nonumber\\

&\leq -\frac{\gamma_1}{2}\beta_1\|x_1\|^2 +\frac{6g^2_{M1}\beta_1}{\gamma_1}\|\hat{\phi}_{u1}\|^2\|\tilde{W}_u(t)\|^2 -\frac{\alpha_{h} \beta_{2}}{4}  \frac{\left\|\hat{\Psi}_{V1}\right\|^{2}}{1+\left\|\hat{\Psi}_{V1}\right\|^{2}}\left\|\tilde{W}_{V1}\right\|^{2}\nonumber\\

&+2\left[\begin{aligned}
 &3\left[\alpha_{m} \frac{\beta_{3} L_{\Psi_{m1}}\left\|W_{m1}\right\|^{2}}{1+\left\|\hat{\Psi}_{m1}\right\|^{2}}  +\alpha_u \beta_{4} \frac{\left\|R_1^{-1} g_1^{T}\left(x_1\right)\right\|^{2}}{1+\left\|\hat{\phi}_{u1}\right\|^{2}}\right]  L^2_{\phi v1}\|W_{V1}\|^2 \nonumber\\
&+\alpha_{h} \frac{\beta_{2}\left[L_\Phi+L_{\Psi V1}\left\|W_{V1}\right\|^{2}\right]}{1+\left\|\hat{\Psi}_{V1}\right\|^{2}} +\frac{6g^2_{M1}\beta_1}{\gamma}L^2_{\phi u}\|W_u\|^2  \nonumber
\end{aligned}\right] \|\tilde{m}_2\|^2\|\tilde{W}_{m1}(t)\|^2\|\phi_{m1}\|^2  \nonumber\\

&+2\left[\begin{aligned}
 &3\left[\alpha_{m} \frac{\beta_{3} L_{\Psi_{m1}}\left\|W_{m1}\right\|^{2}}{1+\left\|\hat{\Psi}_{m1}\right\|^{2}}  +\alpha_u \beta_{4} \frac{\left\|R_1^{-1} g_1^{T}\left(x_1\right)\right\|^{2}}{1+\left\|\hat{\phi}_{u1}\right\|^{2}}\right]  L^2_{\phi v1}\|W_{V1}\|^2 \nonumber\\
&+\alpha_{h} \frac{\beta_{2}\left[L_\Phi+L_{\Psi V1}\left\|W_{V1}\right\|^{2}\right]}{1+\left\|\hat{\Psi}_{V1}\right\|^{2}}  +\frac{6g^2_{M1}\beta_1}{\gamma}L^2_{\phi u}\|W_u\|^2  \nonumber
\end{aligned}\right] \|\tilde{m}_2\|^2\|\varepsilon_{FPK}\|^2   +\frac{6g^2_{M1}\beta_1}{\gamma}\|\varepsilon_{u1}\|^2\nonumber \\

&-\frac{\alpha_{m} \beta_{3}}{2} \frac{\left\|\hat{\Psi}_{m1}\right\|^{2}}{1+\left\|\hat{\Psi}_{m1}\right\|^{2}}\left\|\tilde{W}_{m1}\right\|^{2}  -\frac{\alpha_{u} \beta_{4}}{4} \frac{\left\|\hat{\phi}_{u1}\right\|^{2}}{1+\left\|\hat{\phi}_{u1}\right\|^{2}}\left\|\tilde{W}_{u1}\right\|^{2}  +\beta_{4}\varepsilon_{Nu1} +\beta_{3} \varepsilon_{N F P K1}+\beta_{2} \varepsilon_{VHJI1}  \nonumber\\

&+3\left[\alpha_{m} \frac{\beta_{3} L_{\Psi_{m1}}\left\|W_{m1}\right\|^{2}}{1+\left\|\hat{\Psi}_{m1}\right\|^{2}}  +\alpha_u \beta_{4} \frac{\left\|R_1^{-1} g_1^{T}\left(x_1\right)\right\|^{2}}{1+\left\|\hat{\phi}_{u1}\right\|^{2}}\right]  \| \tilde{W}_{V1}(t)\|^2\|\hat{\phi}_{V1}\|^2\nonumber\\

&+3\left[\alpha_{m} \frac{\beta_{3} L_{\Psi_{m1}}\left\|W_{m1}\right\|^{2}}{1+\left\|\hat{\Psi}_{m1}\right\|^{2}}  +\alpha_u \beta_{4} \frac{\left\|R_1^{-1} g_1^{T}\left(x_1\right)\right\|^{2}}{1+\left\|\hat{\phi}_{u1}\right\|^{2}}\right]  \|\varepsilon_{HJI1}\|^2  \nonumber\\

&-\frac{\gamma_2}{2}\beta_5\|x_2\|^2+\frac{6g^2_{M2}\beta_5}{\gamma_2}\|\hat{\phi}_{u2}\|^2\|\tilde{W}_{u2}(t)\|^2-\frac{\alpha_{h} \beta_6}{4}  \frac{\left\|\hat{\Psi}_{V2}\right\|^{2}}{1+\left\|\hat{\Psi}_{V2}\right\|^{2}}\left\|\tilde{W}_{V2}\right\|^{2}     \nonumber\\

&+2\left[\begin{aligned}
&3\left[\alpha_{m} \frac{\beta_{7} L_{\Psi_{m2}}\left\|W_{m2}\right\|^{2}}{1+\left\|\hat{\Psi}_{m2}\right\|^{2}}  +\alpha_u \beta_{8} \frac{\left\|R_2^{-1} g_2^{T}\left(x_2\right)\right\|^{2}}{1+\left\|\hat{\phi}_{u2}\right\|^{2}}\right] L^2_{\phi v2}\|W_{V2}\|^2\nonumber\\

&+\alpha_{h} \frac{\beta_6\left[L_\Phi+L_{\Psi V2}\left\|W_{V2}\right\|^{2}\right]}{1+\left\|\hat{\Psi}_{V2}\right\|^{2}} +\frac{6g^2_{M2}\beta_5}{\gamma}L^2_{\phi {u2}}\|W_{u2}\|^2 
\end{aligned}\right] \|\tilde{m}_1\|^2\|\tilde{W}_{m2}(t)\|^2 \|\phi_{m2}\|^2  \nonumber\\

&+2\left[\begin{aligned}
&\left[\alpha_{m} \frac{\beta_{7} L_{\Psi_{m2}}\left\|W_{m2}\right\|^{2}}{1+\left\|\hat{\Psi}_{m2}\right\|^{2}}  +\alpha_u \beta_{8} \frac{\left\|R_2^{-1} g_2^{T}\left(x_2\right)\right\|^{2}}{1+\left\|\hat{\phi}_{u2}\right\|^{2}}\right] L^2_{\phi v2}\|W_{V2}\|^2\nonumber\\   

&+\alpha_{h} \frac{\beta_6\left[L_\Phi+L_{\Psi V2}\left\|W_{V2}\right\|^{2}\right]}{1+\left\|\hat{\Psi}_{V2}\right\|^{2}} +\frac{6g^2_{M2}\beta_5}{\gamma}L^2_{\phi u2}\|W_{u2}\|^2
\end{aligned}\right] \|\tilde{m}_1\|^2\|\varepsilon_{FPK}\|^2 +\frac{6g^2_{M2}\beta_5}{\gamma}\|\varepsilon_{u2}\|^2  \nonumber\\

&-\frac{\alpha_{m} \beta_7}{2} \frac{\left\|\hat{\Psi}_{m2}\right\|^{2}}{1+\left\|\hat{\Psi}_{m2}\right\|^{2}}\left\|\tilde{W}_{m2}\right\|^{2}-\frac{\alpha_{u} \beta_{8}}{4} \frac{\left\|\hat{\phi}_{u2}\right\|^{2}}{1+\left\|\hat{\phi}_{u2}\right\|^{2}}\left\|\tilde{W}_{u2}\right\|^{2} +\beta_{8} \varepsilon_{Nu2} +\beta_7 \varepsilon_{N F P K2}+\beta_6 \varepsilon_{VHJI2}    \nonumber\\

&+3\left[\alpha_{m} \frac{\beta_{7} L_{\Psi_{m2}}\left\|W_{m2}\right\|^{2}}{1+\left\|\hat{\Psi}_{m2}\right\|^{2}}  +\alpha_u \beta_{8} \frac{\left\|R_2^{-1} g_2^{T}\left(x_2\right)\right\|^{2}}{1+\left\|\hat{\phi}_{u2}\right\|^{2}}\right]  \|\tilde{W}_{V2}(t)\|^2\|\hat{\phi}_{V2}\|^2\nonumber\\

&+3\left[\alpha_{m} \frac{\beta_{7} L_{\Psi_{m2}}\left\|W_{m2}\right\|^{2}}{1+\left\|\hat{\Psi}_{m2}\right\|^{2}}  +\alpha_u \beta_{8} \frac{\left\|R_2^{-1} g_2^{T}\left(x_2\right)\right\|^{2}}{1+\left\|\hat{\phi}_{u2}\right\|^{2}}\right]  \|\varepsilon_{HJI2}\|^2 
\end{align}

Combine the terms in \eqref{eq54} yields:
\begin{align}\label{eq55}
     &\cdot{L}_{sys}(t)\leq -\frac{\gamma_1}{2}\beta_1\|x_1\|^2 -\left[\frac{\alpha_{u} \beta_{4}}{4} \frac{\left\|\hat{\phi}_{u1}\right\|^{2}}{1+\left\|\hat{\phi}_{u1}\right\|^{2}} -\frac{6g^2_{M1}\beta_1}{\gamma_1} \|\hat{\phi}_{u1}\|^2 \right] \|\tilde{W}_u(t)\|^2 \nonumber\\
     
     &-\left[\frac{\alpha_{h} \beta_{2}}{4}  \frac{\left\|\hat{\Psi}_{V1}\right\|^{2}}{1+\left\|\hat{\Psi}_{V1}\right\|^{2}} -3\left[\alpha_{m} \frac{\beta_{3} L_{\Psi_{m1}}\left\|W_{m1}\right\|^{2}}{1+\left\|\hat{\Psi}_{m1}\right\|^{2}}  +\alpha_u \beta_{4} \frac{\left\|R_1^{-1} g_1^{T}\left(x_1\right)\right\|^{2}}{1+\left\|\hat{\phi}_{u1}\right\|^{2}}\right]\|\hat{\phi}_{V1}\|^2  \right]\| \tilde{W}_{V1}(t)\|^2 \nonumber\\
     
     &-\left[\frac{\alpha_{m} \beta_{3}}{2} \frac{\left\|\hat{\Psi}_{m1}\right\|^{2}}{1+\left\|\hat{\Psi}_{m1}\right\|^{2}}   -2\left[\begin{aligned}
 &3\left[\alpha_{m} \frac{\beta_{3} L_{\Psi_{m1}}\left\|W_{m1}\right\|^{2}}{1+\left\|\hat{\Psi}_{m1}\right\|^{2}}  +\alpha_u \beta_{4} \frac{\left\|R_1^{-1} g_1^{T}\left(x_1\right)\right\|^{2}}{1+\left\|\hat{\phi}_{u1}\right\|^{2}}\right]  L^2_{\phi v1}\|W_{V1}\|^2 \nonumber\\
&+\alpha_{h} \frac{\beta_{2}\left[L_\Phi+L_{\Psi V1}\left\|W_{V1}\right\|^{2}\right]}{1+\left\|\hat{\Psi}_{V1}\right\|^{2}} +\frac{6g^2_{M1}\beta_1}{\gamma_1}L^2_{\phi u}\|W_{u1}\|^2  \nonumber
\end{aligned}\right] \|\tilde{m}_2\|^2\|\phi_{m1}\|^2 \right]\left\|\tilde{W}_{m1}\right\|^{2} \nonumber\\

&+2\left[\begin{aligned}
 &3\left[\alpha_{m} \frac{\beta_{3} L_{\Psi_{m1}}\left\|W_{m1}\right\|^{2}}{1+\left\|\hat{\Psi}_{m1}\right\|^{2}}  +\alpha_u \beta_{4} \frac{\left\|R_1^{-1} g_1^{T}\left(x_1\right)\right\|^{2}}{1+\left\|\hat{\phi}_{u1}\right\|^{2}}\right]  L^2_{\phi v1}\|W_{V1}\|^2 \nonumber\\
&+\alpha_{h} \frac{\beta_{2}\left[L_\Phi+L_{\Psi V1}\left\|W_{V1}\right\|^{2}\right]}{1+\left\|\hat{\Psi}_{V1}\right\|^{2}}  +\frac{6g^2_{M1}\beta_1}{\gamma_1}L^2_{\phi u}\|W_u\|^2  \nonumber
\end{aligned}\right] \|\tilde{m}_2\|^2\|\varepsilon_{FPK}\|^2   +\frac{6g^2_{M1}\beta_1}{\gamma_1}\|\varepsilon_{u1}\|^2\nonumber \\

&+\left[\alpha_{m} \frac{\beta_{3} L_{\Psi_{m1}}\left\|W_{m1}\right\|^{2}}{1+\left\|\hat{\Psi}_{m1}\right\|^{2}}  +\alpha_u \beta_{4} \frac{\left\|R_1^{-1} g_1^{T}\left(x_1\right)\right\|^{2}}{1+\left\|\hat{\phi}_{u1}\right\|^{2}}\right]  \|\varepsilon_{HJI1}\|^2 +\beta_{4}\varepsilon_{Nu1} +\beta_{3} \varepsilon_{N F P K1}+\beta_{2} \varepsilon_{VHJI1}  \nonumber\\

&-\frac{\gamma_2}{2}\beta_5\|x_2\|^2-\left[\frac{\alpha_{u} \beta_{8}}{4} \frac{\left\|\hat{\phi}_{u2}\right\|^{2}}{1+\left\|\hat{\phi}_{u2}\right\|^{2}} -\frac{6g^2_{M2}\beta_5}{\gamma_2} \|\hat{\phi}_{u2}\|^2 \right] \|\tilde{W}_{u2}(t)\|^2 \nonumber\\
     
     &-\left[\frac{\alpha_{h} \beta_{6}}{4}  \frac{\left\|\hat{\Psi}_{V2}\right\|^{2}}{1+\left\|\hat{\Psi}_{V2}\right\|^{2}} -3\left[\alpha_{m} \frac{\beta_{7} L_{\Psi_{m2}}\left\|W_{m2}\right\|^{2}}{1+\left\|\hat{\Psi}_{m2}\right\|^{2}}  +\alpha_u \beta_{8} \frac{\left\|R_2^{-1} g_2^{T}\left(x_2\right)\right\|^{2}}{1+\left\|\hat{\phi}_{u2}\right\|^{2}}\right]\|\hat{\phi}_{V2}\|^2  \right]\| \tilde{W}_{V2}(t)\|^2 \nonumber\\
     
     &-\left[\frac{\alpha_{m} \beta_{7}}{2} \frac{\left\|\hat{\Psi}_{m2}\right\|^{2}}{1+\left\|\hat{\Psi}_{m2}\right\|^{2}}   -2\left[\begin{aligned}
 &3\left[\alpha_{m} \frac{\beta_{7} L_{\Psi_{m2}}\left\|W_{m2}\right\|^{2}}{1+\left\|\hat{\Psi}_{m2}\right\|^{2}}  +\alpha_u \beta_{8} \frac{\left\|R_2^{-1} g_2^{T}\left(x_2\right)\right\|^{2}}{1+\left\|\hat{\phi}_{u2}\right\|^{2}}\right]  L^2_{\phi v2}\|W_{V2}\| \nonumber\\
&+\alpha_{h} \frac{\beta_{6}\left[L_\Phi+L_{\Psi V2}\left\|W_{V2}\right\|^{2}\right]}{1+\left\|\hat{\Psi}_{V2}\right\|^{2}} +\frac{6g^2_{M2}\beta_5}{\gamma_2}L^2_{\phi u2}\|W_{u2}\|^2  \nonumber
\end{aligned}\right]  \|\tilde{m}_1\|^2\|\phi_{m2}\|^2 \right]\left\|\tilde{W}_{m2}\right\|^{2} \nonumber\\

&+2\left[\begin{aligned}
 &3\left[\alpha_{m} \frac{\beta_{7} L_{\Psi_{m2}}\left\|W_{m2}\right\|^{2}}{1+\left\|\hat{\Psi}_{m2}\right\|^{2}}  +\alpha_u \beta_{8} \frac{\left\|R_2^{-1} g_2^{T}\left(x_2\right)\right\|^{2}}{1+\left\|\hat{\phi}_{u2}\right\|^{2}}\right]  L^2_{\phi v2}\|W_{V2}\| \nonumber\\
&+\alpha_{h} \frac{\beta_{6}\left[L_\Phi+L_{\Psi V2}\left\|W_{V2}\right\|^{2}\right]}{1+\left\|\hat{\Psi}_{V2}\right\|^{2}}  +\frac{6g^2_{M2}\beta_5}{\gamma_2}L^2_{\phi u2}\|W_{u2}\|^2  \nonumber
\end{aligned}\right]  \|\tilde{m}_1\|^2\|\varepsilon_{FPK2}\|^2   +\frac{6g^2_{M2}\beta_5}{\gamma_2}\|\varepsilon_{u2}\|^2\nonumber \\

&+\left[\alpha_{m} \frac{\beta_{7} L_{\Psi_{m2}}\left\|W_{m2}\right\|^{2}}{1+\left\|\hat{\Psi}_{m2}\right\|^{2}}  +\alpha_u \beta_{8} \frac{\left\|R_2^{-1} g_2^{T}\left(x_2\right)\right\|^{2}}{1+\left\|\hat{\phi}_{u2}\right\|^{2}}\right]  \|\varepsilon_{HJI2}\|^2 +\beta_{8}\varepsilon_{Nu2} +\beta_{7} \varepsilon_{N F P K2}+\beta_{6} \varepsilon_{VHJI2} \nonumber\\\nonumber\\

&\leq -\frac{\gamma_1\beta_1}{2} -\frac{\gamma_2\beta_5}{2} -\kappa_{u1}\|\tilde{W}_{u1}\|^2 -\kappa_{m1}\|\tilde{W}_{m1}\|^2 -\kappa_{V1}\|\tilde{W}_{V1}\|^2 -\kappa_{u2}\|\tilde{W}_{u2}\|^2 -\kappa_{m2}\|\tilde{W}_{m2}\|^2 -\kappa_{V2}\|\tilde{W}_{V2}\|^2 \nonumber\\

&+\varepsilon_{CLS1} +\varepsilon_{CLS2}
\end{align} 
with $\kappa$ and $\varepsilon$ parameters defined as 
\begin{align*}
    &\kappa_{u1}=\left[\frac{\alpha_{u} \beta_{4}}{4} \frac{\left\|\hat{\phi}_{u1}\right\|^{2}}{1+\left\|\hat{\phi}_{u1}\right\|^{2}} -\frac{6g^2_{M1}\beta_1}{\gamma_1} \|\hat{\phi}_{u1}\|^2 \right] \\
    
    &\kappa_{m1}=\left[\frac{\alpha_{m} \beta_{3}}{2} \frac{\left\|\hat{\Psi}_{m1}\right\|^{2}}{1+\left\|\hat{\Psi}_{m1}\right\|^{2}}   -2\left[\begin{aligned}
 &3\left[\alpha_{m} \frac{\beta_{3} L_{\Psi_{m1}}\left\|W_{m1}\right\|^{2}}{1+\left\|\hat{\Psi}_{m1}\right\|^{2}}  +\alpha_u \beta_{4} \frac{\left\|R_1^{-1} g_1^{T}\left(x_1\right)\right\|^{2}}{1+\left\|\hat{\phi}_{u1}\right\|^{2}}\right]  L^2_{\phi v1}\|W_{V1}\|^2 \nonumber\\
&+\alpha_{h} \frac{\beta_{2}\left[L_\Phi+L_{\Psi V1}\left\|W_{V1}\right\|^{2}\right]}{1+\left\|\hat{\Psi}_{V1}\right\|^{2}} +\frac{6g^2_{M1}\beta_1}{\gamma_1}L^2_{\phi u}\|W_{u1}\|^2  \nonumber
\end{aligned}\right] \|\tilde{m}_2\|^2\|\phi_{m1}\|^2 \right] \\

&\kappa_{V1}=\left[\frac{\alpha_{h} \beta_{2}}{4}  \frac{\left\|\hat{\Psi}_{V1}\right\|^{2}}{1+\left\|\hat{\Psi}_{V1}\right\|^{2}} -3\left[\alpha_{m} \frac{\beta_{3} L_{\Psi_{m1}}\left\|W_{m1}\right\|^{2}}{1+\left\|\hat{\Psi}_{m1}\right\|^{2}}  +\alpha_u \beta_{4} \frac{\left\|R_1^{-1} g_1^{T}\left(x_1\right)\right\|^{2}}{1+\left\|\hat{\phi}_{u1}\right\|^{2}}\right]\|\hat{\phi}_{V1}\|^2  \right] \\

&\kappa_{u2}=\left[\frac{\alpha_{u} \beta_{8}}{4} \frac{\left\|\hat{\phi}_{u2}\right\|^{2}}{1+\left\|\hat{\phi}_{u2}\right\|^{2}} -\frac{6g^2_{M2}\beta_5}{\gamma_2} \|\hat{\phi}_{u2}\|^2 \right] \\

&\kappa_{m2}=\left[\frac{\alpha_{m} \beta_{7}}{2} \frac{\left\|\hat{\Psi}_{m2}\right\|^{2}}{1+\left\|\hat{\Psi}_{m2}\right\|^{2}}   -2\left[\begin{aligned}
 &3\left[\alpha_{m} \frac{\beta_{7} L_{\Psi_{m2}}\left\|W_{m2}\right\|^{2}}{1+\left\|\hat{\Psi}_{m2}\right\|^{2}}  +\alpha_u \beta_{8} \frac{\left\|R_2^{-1} g_2^{T}\left(x_2\right)\right\|^{2}}{1+\left\|\hat{\phi}_{u2}\right\|^{2}}\right]  L^2_{\phi v2}\|W_{V2}\| \nonumber\\
 
&+\alpha_{h} \frac{\beta_{6}\left[L_\Phi+L_{\Psi V2}\left\|W_{V2}\right\|^{2}\right]}{1+\left\|\hat{\Psi}_{V2}\right\|^{2}} +\frac{6g^2_{M2}\beta_5}{\gamma_2}L^2_{\phi u2}\|W_{u2}\|^2  \nonumber
\end{aligned}\right]  \|\tilde{m}_1\|^2\|\phi_{m2}\|^2 \right] \\

&\kappa_{V2}=\left[\frac{\alpha_{h} \beta_{6}}{4}  \frac{\left\|\hat{\Psi}_{V2}\right\|^{2}}{1+\left\|\hat{\Psi}_{V2}\right\|^{2}} -3\left[\alpha_{m} \frac{\beta_{7} L_{\Psi_{m2}}\left\|W_{m2}\right\|^{2}}{1+\left\|\hat{\Psi}_{m2}\right\|^{2}}  +\alpha_u \beta_{8} \frac{\left\|R_2^{-1} g_2^{T}\left(x_2\right)\right\|^{2}}{1+\left\|\hat{\phi}_{u2}\right\|^{2}}\right]\|\hat{\phi}_{V2}\|^2  \right] \\\\

&\varepsilon_{CLS1}=2\left[\begin{aligned}
 &3\left[\alpha_{m} \frac{\beta_{3} L_{\Psi_{m1}}\left\|W_{m1}\right\|^{2}}{1+\left\|\hat{\Psi}_{m1}\right\|^{2}}  +\alpha_u \beta_{4} \frac{\left\|R_1^{-1} g_1^{T}\left(x_1\right)\right\|^{2}}{1+\left\|\hat{\phi}_{u1}\right\|^{2}}\right]  L^2_{\phi v1}\|W_{V1}\| \nonumber\\
&+\alpha_{h} \frac{\beta_{2}\left[L_\Phi+L_{\Psi V1}\left\|W_{V1}\right\|^{2}\right]}{1+\left\|\hat{\Psi}_{V1}\right\|^{2}}  +\frac{6g^2_{M1}\beta_1}{\gamma_1}L^2_{\phi u}\|W_u\|^2  \nonumber
\end{aligned}\right] \|\tilde{m}_2\|^2\|\varepsilon_{FPK}\|^2   +\frac{6g^2_{M1}\beta_1}{\gamma_1}\|\varepsilon_{u1}\|^2\nonumber \\

&+\left[\alpha_{m} \frac{\beta_{3} L_{\Psi_{m1}}\left\|W_{m1}\right\|^{2}}{1+\left\|\hat{\Psi}_{m1}\right\|^{2}}  +\alpha_u \beta_{4} \frac{\left\|R_1^{-1} g_1^{T}\left(x_1\right)\right\|^{2}}{1+\left\|\hat{\phi}_{u1}\right\|^{2}}\right]  \|\varepsilon_{HJI1}\|^2 +\beta_{4}\varepsilon_{Nu1} +\beta_{3} \varepsilon_{N F P K1}+\beta_{2} \varepsilon_{VHJI1}  \nonumber\\\\

&\varepsilon_{CLS2}=2\left[\begin{aligned}
 &3\left[\alpha_{m} \frac{\beta_{7} L_{\Psi_{m2}}\left\|W_{m2}\right\|^{2}}{1+\left\|\hat{\Psi}_{m2}\right\|^{2}}  +\alpha_u \beta_{8} \frac{\left\|R_2^{-1} g_2^{T}\left(x_2\right)\right\|^{2}}{1+\left\|\hat{\phi}_{u2}\right\|^{2}}\right]  L^2_{\phi v2}\|W_{V2}\| \nonumber\\
&+\alpha_{h} \frac{\beta_{6}\left[L_\Phi+L_{\Psi V2}\left\|W_{V2}\right\|^{2}\right]}{1+\left\|\hat{\Psi}_{V2}\right\|^{2}}  +\frac{6g^2_{M2}\beta_5}{\gamma_2}L^2_{\phi u2}\|W_{u2}\|^2  \nonumber
\end{aligned}\right]  \|\tilde{m}_1\|^2\|\varepsilon_{FPK2}\|^2   +\frac{6g^2_{M2}\beta_5}{\gamma_2}\|\varepsilon_{u2}\|^2\nonumber \\

&+\left[\alpha_{m} \frac{\beta_{7} L_{\Psi_{m2}}\left\|W_{m2}\right\|^{2}}{1+\left\|\hat{\Psi}_{m2}\right\|^{2}}  +\alpha_u \beta_{8} \frac{\left\|R_2^{-1} g_2^{T}\left(x_2\right)\right\|^{2}}{1+\left\|\hat{\phi}_{u2}\right\|^{2}}\right]  \|\varepsilon_{HJI2}\|^2 +\beta_{8}\varepsilon_{Nu2} +\beta_{7} \varepsilon_{N F P K2}+\beta_{6} \varepsilon_{VHJI2}
\end{align*}

Note that the coefficient functions $\kappa_{u1}$, $\kappa_{m1}$, $\kappa_{V1}$, $\kappa_{u2}$, $\kappa_{m2}$, and $\kappa_{V2}$ are all positive definite, and the terms $\varepsilon_{CLS1}$ and $\varepsilon_{CLS2}$ go to zero if the reconstruction errors $\varepsilon_{HJI1}$, $\varepsilon_{FPK1}$, $\varepsilon_{u1}$, $\varepsilon_{HJI2}$, $\varepsilon_{FPK2}$, $\varepsilon_{u2}$ go to zero. The meaning of reconstruction error goes to zero means that the neural network structure and activation functions are perfectly selected. In that case, the first derivative of the Lyapunov function is negative definite which means the closed loop system is asymptotically stable. In the case where the reconstruction error is not zero, the closed loop system is Uniformly Ultimately Bounded (UUB).
\end{proof}

\bibliographystyle{IEEEtran}
\bibliography{reference.bib}

\begin{thebibliography}{10}
\providecommand{\url}[1]{#1}
\csname url@samestyle\endcsname
\providecommand{\newblock}{\relax}
\providecommand{\bibinfo}[2]{#2}
\providecommand{\BIBentrySTDinterwordspacing}{\spaceskip=0pt\relax}
\providecommand{\BIBentryALTinterwordstretchfactor}{4}
\providecommand{\BIBentryALTinterwordspacing}{\spaceskip=\fontdimen2\font plus
\BIBentryALTinterwordstretchfactor\fontdimen3\font minus
  \fontdimen4\font\relax}
\providecommand{\BIBforeignlanguage}[2]{{%
\expandafter\ifx\csname l@#1\endcsname\relax
\typeout{** WARNING: IEEEtran.bst: No hyphenation pattern has been}%
\typeout{** loaded for the language `#1'. Using the pattern for}%
\typeout{** the default language instead.}%
\else
\language=\csname l@#1\endcsname
\fi
#2}}
\providecommand{\BIBdecl}{\relax}
\BIBdecl

\bibitem{vlahov2018developing}
B.~Vlahov, E.~Squires, L.~Strickland, and C.~Pippin, ``On developing a uav
  pursuit-evasion policy using reinforcement learning,'' in \emph{2018 17th
  IEEE International Conference on Machine Learning and Applications
  (ICMLA)}.\hskip 1em plus 0.5em minus 0.4em\relax IEEE, 2018, pp. 859--864.

\bibitem{ramana2017pursuit}
M.~V. Ramana and M.~Kothari, ``Pursuit-evasion games of high speed evader,''
  \emph{Journal of intelligent \& robotic systems}, vol.~85, no.~2, pp.
  293--306, 2017.

\bibitem{camci2016game}
E.~Camci and E.~Kayacan, ``Game of drones: Uav pursuit-evasion game with type-2
  fuzzy logic controllers tuned by reinforcement learning,'' in \emph{2016 IEEE
  International Conference on Fuzzy Systems (FUZZ-IEEE)}.\hskip 1em plus 0.5em
  minus 0.4em\relax IEEE, 2016, pp. 618--625.

\bibitem{wilson2017pursuit}
B.~Wilson, S.~Sundaram, and A.~Prasad, ``Pursuit evasion with multiple
  pursuers: Capturing a ground vehicle on a road network with multiple
  drones,'' 2017.

\bibitem{turetsky2016target}
V.~Turetsky and T.~Shima, ``Target evasion from a missile performing multiple
  switches in guidance law,'' \emph{Journal of Guidance, Control, and
  Dynamics}, pp. 2364--2373, 2016.

\bibitem{makkapati2018pursuit}
V.~R. Makkapati, W.~Sun, and P.~Tsiotras, ``Pursuit-evasion problems involving
  two pursuers and one evader,'' in \emph{2018 AIAA Guidance, Navigation, and
  Control Conference}, 2018, p. 2107.

\bibitem{sun2017multiple}
W.~Sun, P.~Tsiotras, T.~Lolla, D.~N. Subramani, and P.~F. Lermusiaux,
  ``Multiple-pursuer/one-evader pursuit--evasion game in dynamic flowfields,''
  \emph{Journal of guidance, control, and dynamics}, vol.~40, no.~7, pp.
  1627--1637, 2017.

\bibitem{faguo1}
J.-M. Lasry and P.-L. Lions, ``Mean field games,'' \emph{Japanese journal of
  mathematics}, vol.~2, no.~1, pp. 229--260, 2007.

\bibitem{olfati2007consensus}
R.~Olfati-Saber, J.~A. Fax, and R.~M. Murray, ``Consensus and cooperation in
  networked multi-agent systems,'' \emph{Proceedings of the IEEE}, vol.~95,
  no.~1, pp. 215--233, 2007.

\bibitem{frankbook}
F.~L. Lewis, D.~Vrabie, and K.~G. Vamvoudakis, ``Reinforcement learning and
  feedback control: Using natural decision methods to design optimal adaptive
  controllers,'' \emph{IEEE Control Systems Magazine}, vol.~32, no.~6, pp.
  76--105, 2012.

\bibitem{zhou2019intelligent}
Z.~Zhou, L.~Qian, and H.~Xu, ``Intelligent decentralized dynamic power
  allocation in manet at tactical edge based on mean-field game theory,'' in
  \emph{MILCOM 2019-2019 IEEE Military Communications Conference
  (MILCOM)}.\hskip 1em plus 0.5em minus 0.4em\relax IEEE, 2019, pp. 604--609.

\bibitem{zhou2019decentralized}
Z.~Zhou and H.~Xu, ``Decentralized adaptive optimal tracking control for
  massive multi-agent systems with input constraint,'' in \emph{2019 IEEE
  Symposium Series on Computational Intelligence (SSCI)}.\hskip 1em plus 0.5em
  minus 0.4em\relax IEEE, 2019, pp. 1--8.

\bibitem{nourian2012mean}
M.~Nourian, P.~E. Caines, R.~P. Malham{\'e}, and M.~Huang, ``Mean field lqg
  control in leader-follower stochastic multi-agent systems: Likelihood ratio
  based adaptation,'' \emph{IEEE Transactions on Automatic Control}, vol.~57,
  no.~11, pp. 2801--2816, 2012.

\bibitem{vamvoudakis2012online}
K.~G. Vamvoudakis and F.~L. Lewis, ``Online solution of nonlinear two-player
  zero-sum games using synchronous policy iteration,'' \emph{International
  Journal of Robust and Nonlinear Control}, vol.~22, no.~13, pp. 1460--1483,
  2012.

\bibitem{b9}
G.~Cybenko, ``Approximation by superpositions of a sigmoidal function,''
  \emph{Mathematics of control, signals and systems}, vol.~2, no.~4, pp.
  303--314, 1989.

\end{thebibliography}

\end{document}